\documentclass[reprint,prx,aps,twocolumn,superscriptaddress,groupedaddress,longbibliography,footinbib]{revtex4-2}
\usepackage{amsmath,amsfonts,amsthm,bbm,graphicx,color,enumerate,hyperref,xcolor,tikz,pgfplots,ulem,subfigure}
\usepackage{ulem}

\hypersetup{
	colorlinks=true,  
	linkcolor=blue,   
	citecolor=blue,   
	urlcolor=blue     
}

\def\E{ {\mathcal E} }

\def\h{ {\mathcal H} }
\def\H{ {\mathcal H} }

\def\J{ {\mathcal J} }
\def\U{ {\mathcal U} }

\def\G{ {\mathcal G} }

\def\F{ {\mathcal F} }

\def\B{ {\mathcal B} }

\def\D{ {\mathcal D} }
\def\S{ {\mathcal S} }
\def\T{ {\mathcal T} }
\def\V{ {\mathcal V} }
\def\I{ {\mathcal I} }

\def\Tr{{\rm{tr}}}
\def\tr{ \mbox{tr} }
\newcommand{\tra}[1]{\mathrm{tr}\left( #1 \right)}
\newcommand{\trb}[2]{\mathrm{tr}_{#1}\left( #2 \right)}
\def\>{\rangle}
\def\<{\langle}

\def\hc{^{\dagger}}

\renewcommand{\emph}{\textit}

\newcommand{\bra}[1]{\langle {#1} |}
\newcommand{\ket}[1]{| {#1} \rangle}
\newcommand{\vect}[1]{|#1\rangle\!\rangle}
 
\newcommand{\ketbra}[2]{\ensuremath{\left|#1\right\rangle\!\!\left\langle#2\right|}}
\newcommand{\braket}[2]{\ensuremath{\left\langle#1\right|\left.\!#2\right\rangle}}

\newcommand{\iden}{\mathbb{I}} 
\newcommand{\vv}[1]{\ensuremath{\boldsymbol #1}}
\newcommand{\vc}[1]{|vec(#1)\rangle}
\newcommand{\vcb}[1]{\langle vec(#1)|}
\newcommand{\clebsch}[6]{\langle #1,#2;#3,#4\hspace{0.5px} | \hspace{0.5px}#5,#6\rangle}

\def\non{ \nonumber\\}

\newcommand{\cc}[1]{\textcolor{black}{#1}}

\newcommand{\kk}[1]{{\color{black}#1}}

\newtheorem{theorem}{Theorem}

\newtheorem{lemma}[theorem]{Lemma}
\newtheorem{definition}[theorem]{Definition}

\newtheorem{result}{Result}
\pgfplotsset{compat=1.16}
\newlength{\figheight}
\newlength{\figwidth}

\begin{document}

	
	\title{Robustness of Noether's principle: Maximal disconnects between conservation laws \& symmetries in quantum theory}
	\author{Cristina C\^{i}rstoiu}
	\email{ccirstoiu@gmail.com }
	\affiliation{Department of Physics, Imperial College London, London SW7 2AZ, UK}
	\affiliation{Department of Computer Science, University of Oxford, Oxford, OX1 3PU, UK}
	\affiliation{ Cambridge Quantum Computing Ltd,
		9a Bridge Street, Cambridge, United Kingdom}
	\author{Kamil Korzekwa}
	\affiliation{Centre for Engineered Quantum Systems, School of Physics, The University of Sydney, Sydney, NSW 2006, Australia}
	\affiliation{International Centre for Theory of Quantum Technologies, University of Gda{\'n}sk, 80-308 Gda{\'n}sk, Poland}	
	\affiliation{Faculty of Physics, Astronomy and Applied Computer Science, Jagiellonian University, 30-348 Kraków, Poland}
	\author{David Jennings}
	\affiliation{School of Physics and Astronomy, University of Leeds, Leeds, LS2 9JT, UK}
	\affiliation{Department of Physics, University of Oxford, Oxford, OX1 3PU, UK}
	\affiliation{Department of Physics, Imperial College London, London SW7 2AZ, UK}
	
	\begin{abstract}
	
		To what extent does Noether's principle apply to quantum channels? Here, we quantify the degree to which imposing a symmetry constraint on quantum channels implies a conservation law, and show that this relates to physically impossible transformations in quantum theory, such as time-reversal and spin-inversion. In this analysis, the convex structure and extremal points of the set of quantum channels symmetric under the action of a Lie group~$G$ becomes essential. It allows us to derive bounds on the deviation from conservation laws under any symmetric quantum channel in terms of the deviation from closed dynamics as measured by the unitarity of the channel~$\E$. In particular, we investigate in detail the U(1) and SU(2) symmetries related to energy and angular momentum conservation laws. In the latter case, we provide fundamental limits on how much a spin-$j_A$ system can be used to polarise a larger spin-$j_B$ system, and on how much one can invert spin polarisation using a rotationally-symmetric operation. Finally, we also establish novel links between unitarity, complementary channels and purity that are of independent interest.

	\end{abstract}
	
	\maketitle
	
	
\section{Introduction}
\label{sec:intro}
	
\subsection{Symmetry principles versus conservation laws}

Noether's theorem in classical mechanics states that for every continuous symmetry of a system there is an associated conserved charge~\cite{noether1918invariante,baez2013noether,gough2015noether}. This fundamental result forms the bedrock for a wide range of applications and insights for theoretical physics in both non-relativistic and relativistic settings. Quantum theory incorporates Noether's principle at a fundamental level, where for unitary dynamics generated by a Hamiltonian $H$ we have that an observable $A$ is conserved, in the sense of $\<\psi|A|\psi\>$ being constant under the dynamics for any state $|\psi\>$, if and only if $[A,H]=0$. In quantum field theory, Noether's theorem gets recast as the Ward-Takahashi identity~\cite{ward1950identity, takahashi1957generalized} for $n$-point correlations in momentum space. 

In all of the above cases a continuous symmetry principle is identified with some conserved quantity. However, the most general kind of evolution of a quantum state, for relativistic or non-relativistic quantum theory, is not unitary dynamics but instead a quantum channel. This broader formalism includes both unitary evolution and open system dynamics, but also allows more general quantum operations such as state preparation or discarding of subsystems. It is therefore natural to ask about the status of Noether's principle for those quantum channels that obey a symmetry principle.

A quantum channel $\E$~\cite{watrous2018theory} takes a quantum state $\rho_A$ of a system $A$ into some other valid quantum state \mbox{$\sigma_B=\E(\rho_A)$} of a potentially different system $B$. The channel respects a symmetry, described by a group $G$, if we have that 
\begin{equation}
	\label{eq:covariance_intro}
	\E(U_A(g) \rho_A U_A^\dagger (g)) = U_B(g) \sigma_B U_B^\dagger (g)
\end{equation}
for all \mbox{$g \in G$}, where $U(g)$ denotes a unitary representation of the group $G$ on the appropriate quantum system.

However, even in the simple case of the U(1) phase group $U(\theta) = e^{i \theta N}$ generated by the number operator~$N$, we know from quantum information analysis in asymmetry theory~\cite{gour2017quantum}, that situations arise in which the symmetry constraint is not captured by \mbox{$\<N\>:= \tr (N \rho)$} being constant~\cite{marvian2014extending}. Indeed, even if we were given all the moments $\<N^k\>$ of the generator $N$ of the symmetry, together with all the spectral data of the state $\rho_A$, this turns out to still be insufficient to determine whether $\rho_A$ may be transformed to some other state $\sigma_B$ while respecting the symmetry. Conversely, given a symmetry principle, there exist quantum channels that can change the expectation of the generators of the symmetry in non-trivial ways. These facts imply that a complex disconnect occurs between symmetries of a system and traditional conservation laws when we extend the analysis to open dynamics described by quantum channels, see Fig.~\ref{fig:symmetry_conservation}. Given this break-down of Noether's principle, our primary aim in this work is to address the following fundamental question:
\begin{center}
	{\it {\bf Q1.} What is the maximal disconnect between symmetry principles and conservation laws for quantum channels?}
\end{center}

Surprisingly, we shall see that this question relates to the distinction between the notion of an active transformation and a passive transformation of a quantum system.

\begin{figure}[t]
	\includegraphics[width=\columnwidth]{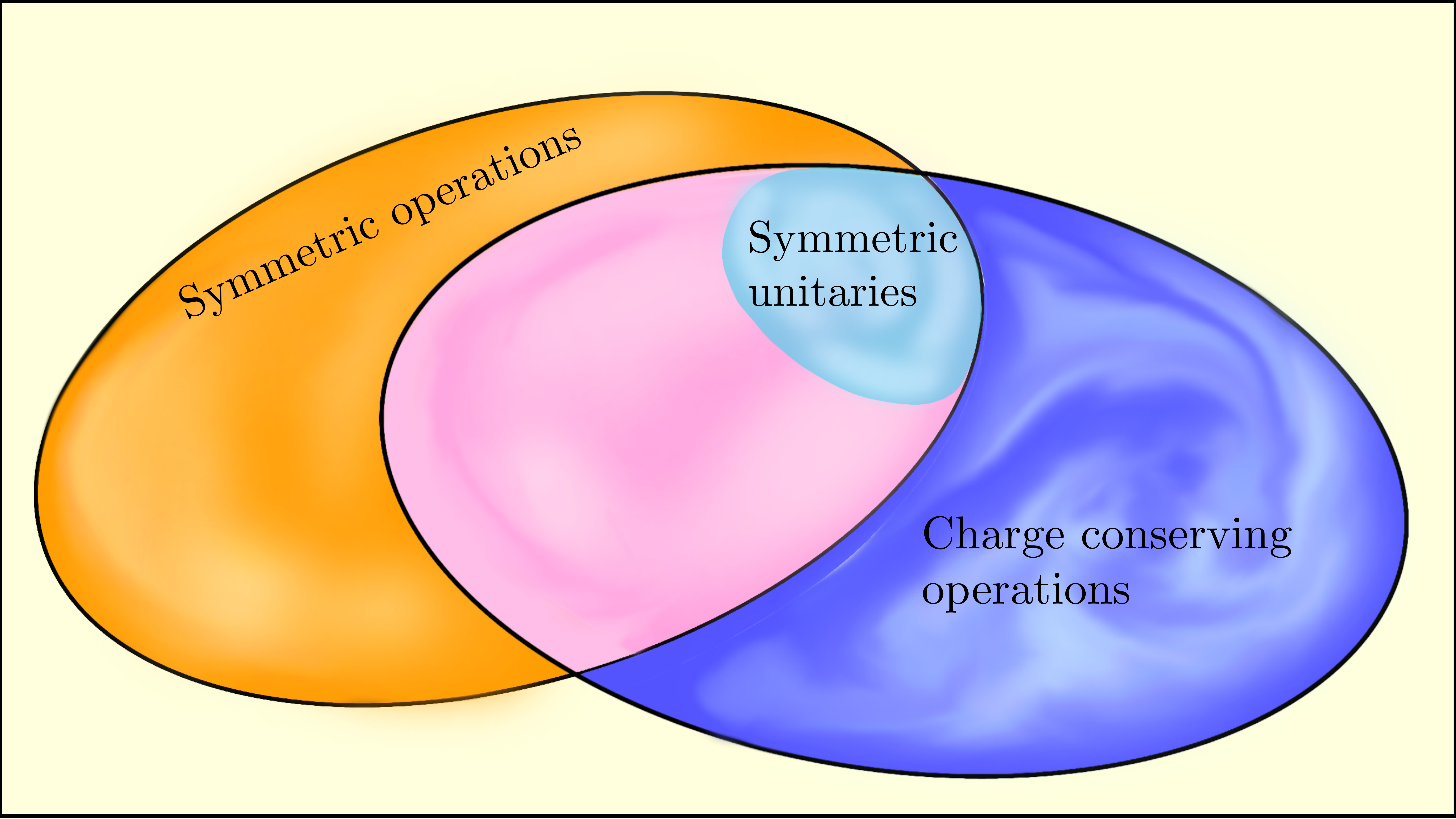}
	\caption{\label{fig:symmetry_conservation}\textbf{Disconnect between symmetries and conservation laws for open quantum dynamics.} Every continuous symmetry of the closed unitary evolution implies a conserved charge, but under the same symmetry constraints quantum channels may change the expectation value of such charges.}
\end{figure}

\subsection{Active versus passive: forbidden transformations in quantum mechanics.}

In quantum mechanics the time-reversal transformation $t \rightarrow -t$ is a stark example of a symmetry transformation that does not correspond to any physical transformation that could be performed on a quantum system $A$~\cite{peres2006quantum}. More precisely, within quantum theory time-reversal must be represented by an anti-unitary operator~$\Theta$, and so cannot be generated by any kind of dynamics acting on a quantum system. Instead, time-reversal is a passive transformation -- namely a change in our description of the physical system. On the other hand, active transformations, such as rotations or translations, are physical transformations with respect to a fixed description (coordinate system) that can be performed on the quantum system $A$. Time-reversal, therefore, constitutes an example of a passive transformation that is without any corresponding active realisation. This is in contrast to spatial rotations of $A$ which admit either passive or active realisations. 

If $A$ is a simple spin system, then the action of time-reversal on the spin angular momentum $\mathbf{J}$ degree of freedom coincides with \emph{spin-inversion}, which transforms states of the system as \mbox{$\rho_A \rightarrow \T(\rho_A)= \Theta \rho_A \Theta^\dagger$}. In the Heisenberg picture this transformation sends \mbox{$\mathbf{J} \rightarrow -\mathbf{J}$}. Indeed, while spin-inversion is seemingly less abstract than time-reversal, it constitutes another symmetry transformation in quantum theory that is forbidden in general -- a passive transformation with no active counterpart. 

The strength of this prohibition on spin-inversion actually depends on the fundamental structure of quantum theory itself. This can be seen if we ask the question: what is the \emph{best approximation} to spin-inversion that can be realised within quantum theory through an active transformation, given by a quantum channel $\E$, of an arbitrary state $\rho_A$ to some new state $\E(\rho_A)$? If we restrict to the simplest possible scenario of $A$ being a spin-1/2 particle system, we have that spin-inversion coincides with the \emph{universal-NOT} gate for a qubit. It is well known that such a gate is impossible in quantum theory~\cite{buzek99optimal}, and the best approximation of such a gate is a channel $\S_-$ that transforms any state $\rho$ with spin polarisation \mbox{$\mathbf{P}(\rho_A):= \tr(\mathbf{J} \rho_A)$} into a quantum state $\S_-(\rho_A)$ such that 
\begin{equation}
\mathbf{P}(\rho_A) \rightarrow \mathbf{P}(\S_-(\rho_A)) = -\frac{1}{3}\mathbf{P}(\rho_A).
\end{equation}
We refer to $\S_-$ as the optimal inversion channel for the system.

It is important to emphasise that the pre-factor of~\mbox{$-{1}/{3}$} is fundamental and cannot be improved on. Its numerical value can be determined by considering the application of quantum operations to one half of a maximally entangled quantum state -- anything closer to perfect spin-inversion would generate negative probabilities, and would thus be unphysical. Indeed, if we removed entanglement from quantum theory, by restricting to separable quantum states, then there would be no prohibition on spin-inversion of the system!\footnote{More precisely, spin-inversion is equivalent to $\rho \rightarrow \rho^T$ followed by a $\pi$-rotation around the $Y$ axis. However, the transpose map is known to be a positive but not completely-positive map~\cite{watrous2018theory}. If we restrict the state space to be the set of separable quantum states, however, such a strictly positive map will never generate negative probabilities and so would be an admissible physical transformation.}

While this limit is easily determined for spin-1/2 systems, it raises the more general question: 
\begin{center}
	{\it  {\bf Q2.} What are the limits imposed by quantum theory on approximate spin-inversion and other such inactive symmetries?}
\end{center}
Here, an \emph{inactive symmetry} transformation simply means a symmetry transformation that is purely passive and does not have an active counterpart. More precisely, and focusing on spin-inversion, the question becomes: given any quantum system $A$, what is the quantum channel $\E$ that optimally approximates spin-inversion on $A$? For a $d=2$ qubit spin system, this analysis essentially coincides with looking at depolarizing channels. However, for a $d>2$ spin system, this connection with depolarizing channels no longer holds and a more detailed analysis is required to account for the spin angular momentum of the quantum system.


\subsection{Structure and scope of the problem}

In this paper, our main focus will be on the maximal disconnects between symmetry principles and conservation laws. We will focus on symmetries corresponding to Lie groups, and the dominant case will be the $SU(2)$ rotational group. This provides an illustration of the non-trivial structures involved, but also shows that the problem of performing an optimal approximation to spin-inversion arises naturally. We do not consider more general inactive symmetries, but leave this to future work.

We first fully solve {\it {\bf Q2}} for the case of spin-inversion, and show that this can be better and better approximated at a state level as we increase the dimension of the spin. However, this has an information-theoretic caveat that things look quite differently at a quantum channel level. The solution of spin-inversion also connects with a seemingly paradoxical ability to perform \emph{spin-amplification} under rotationally symmetric channels. We diagrammatically present these results in Fig.~\ref{fig:spin_inv}. 

\begin{figure}[t]
	\includegraphics[width=\columnwidth]{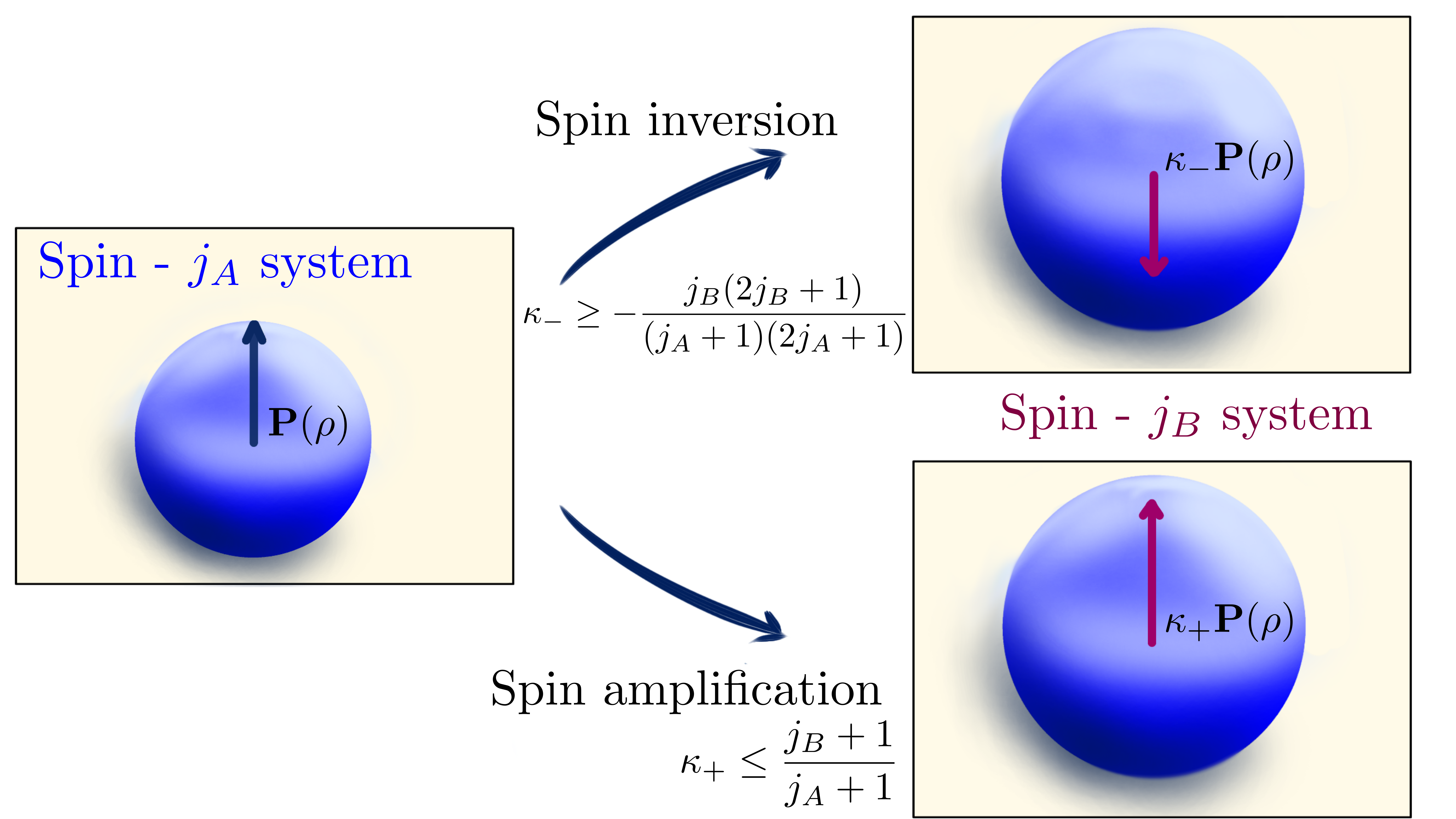}
	\caption{\label{fig:spin_inv}\textbf{Spin-inversion and amplification.} There exist quantum channels that can invert or amplify the polarisation of a spin system while exactly respecting $SU(2)$ rotational symmetry. The values $\kappa_\pm$ provide the ultimate limits of such processes and depend only on the dimension of the spin systems involved.}
\end{figure}

Both spin-inversion and spin-amplification turn out to be two extremal deviations from Noether's principle, and thus lead on to the central question {\it {\bf Q1}}. Here, we derive general bounds on deviations from conservation laws for general groups and systems. These describe the trade-off between allowed deviations and the departure from closed unitary dynamics as schematically portrayed in Fig.~\ref{fig:trade_off}. \kk{Crucially, the quantity we use to measure the departure from closed unitary evolution is extremely well suited to physical scenarios: not only it has a clean theoretical basis, but also it is experimentally measurable and avoids the exponential cost of full tomography of a quantum channel.}

The nature of the considered questions requires one to understand the structural aspects of the set of symmetric quantum channels and, in particular, to have a strong handle on the extremal points of this set. One also needs an operationally sensible way to cast questions {\it {\bf Q1}} and {\it {\bf Q2}} into quantitative and well-defined forms. To these ends we extend previous results on the structure of symmetric channels~\cite{holevo1993note,holevo1995structure,nuwairan2013su2,nuwairan2015su2,mozrzymas2017structure} and derive novel relations for the unitarity of a quantum channel~\cite{wallman2015estimating} -- both of which are of independent interest to the quantum information community. \kk{Our primary methodological advances lie in combining the concept of unitarity, which is efficiently estimable, with harmonic analysis tools for quantum channels. The value of this new methodology is that it provides means to address abstract features of covariant quantum channels, normally expressed in terms of irreducible tensor operators, diamond norm measures, resource measures, etc., with quantities that are readily accessible via experimental methods.}

\begin{figure}[t]
	\includegraphics[width=0.75\columnwidth]{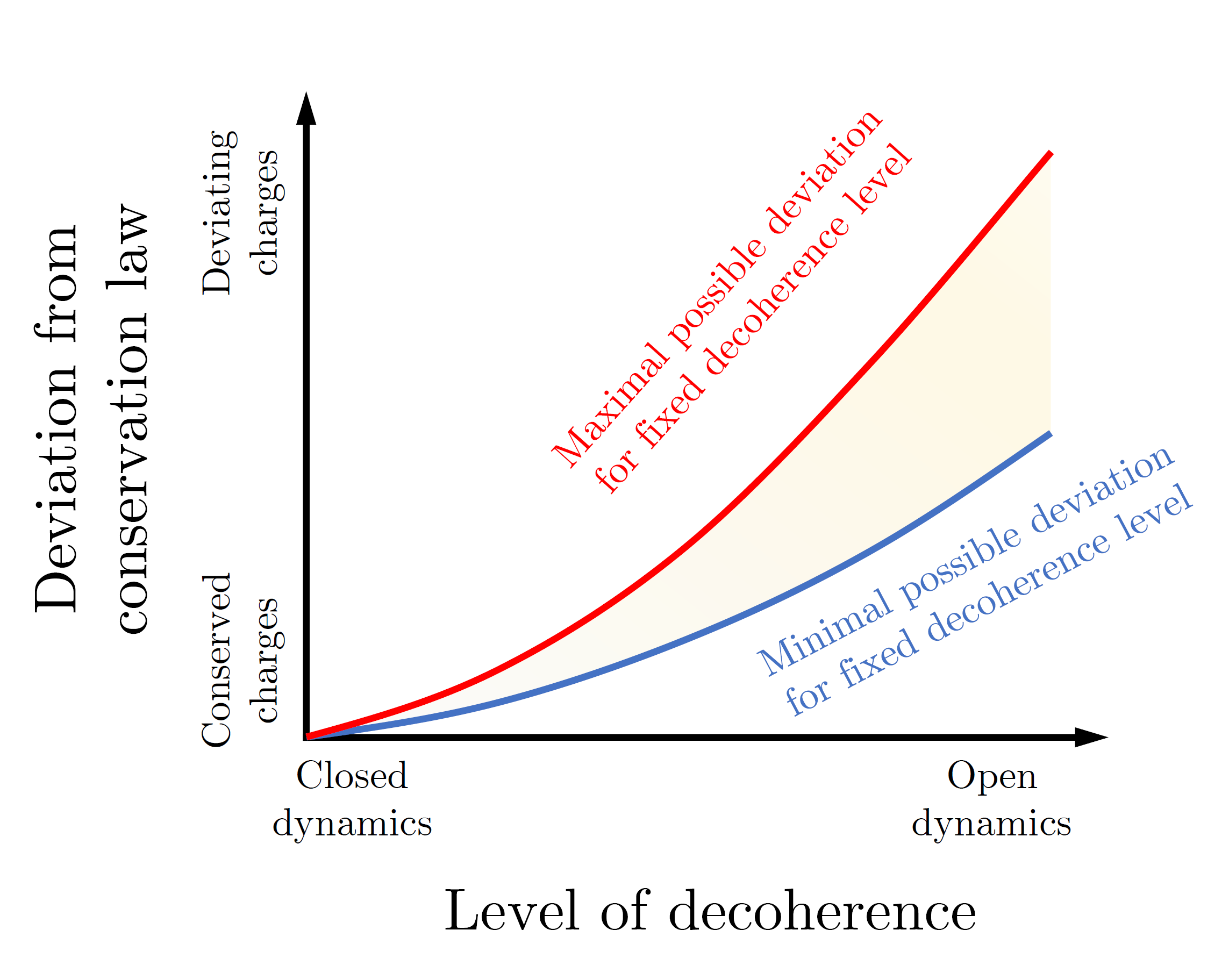}
	\caption{\label{fig:trade_off}\textbf{Robustness of Noether's principle \& trade-off relations.} A qualitative description of trade-off relations between deviation from conservation laws and level of decoherence under the dynamics of a symmetric channel. While the red upper bound exists for all symmetries described by connected Lie groups, the lower bound is present when quantum systems have multiplicity-free decompositions.}
\end{figure}

Since our results provide general bounds on the behaviour of expectation values of observables under symmetric dynamics, we believe that they may be of relevance to scientists working in quantum open systems, decoherence theory, and quantum technologies~\cite{breuer2002theory}. \kk{We explain in more detail how our work connects with the problem of benchmarking quantum devices~\cite{knill2008randomized,magesan2011scalable}; how it can be applied to improve error mitigation in quantum simulations~\cite{georgescu2014quantum}; how it could be extended to study quantum measurement theory~\cite{ozawa2002conservation}; and how it bounds the thermodynamic transformations of quantum systems~\cite{goold2016role}. In each of these cases we specify how concrete applications of our results can be made.} Moreover, as Noether's principle is fundamental and far-reaching, our studies are of potential interest to people investigating foundational topics and relativistic physics~\cite{hohn2016operational,vanrietvelde2018switching}. 

The structure of the paper is as follows. In the next section we give a detailed overview of our main results, and then in the rest of the paper we gradually introduce all the necessary ingredients that allow us to rigorously address the questions posed here and derive our results. In Sec.~\ref{sec:notation} we introduce the notation and provide preliminaries on covariant quantum channels. Next, in Sec.~\ref{sec:deviations} we define quantitative measures of the departure from conservation laws and from closed unitary dynamics. Section~\ref{sec:structure} contains the technical core of our paper with a detailed analysis of the convex structure of the set of symmetric channels. In Sec.~\ref{sec:reversal} we then use these mathematical tools to address the problem of spin-inversion and amplification, while in Sec.~\ref{sec:trade_off} we derive trade-off relations between conservation laws and decoherence. \kk{Section~\ref{sec:applications} is devoted to potential applications of our results to various fields of quantum information science.} Finally, Sec.~\ref{sec:conclusions} contains conclusions and outlook.


\section{Overview of Main Results}

\kk{The central message of our work is that we extend Noether's principle and the general relationship between symmetry and conservation laws to arbitrary quantum evolutions with a natural regulator to measure the openness of the dynamics, which can be efficiently estimated experimentally. We thus provide a concrete methodology to answer questions {\it {\bf Q1}} and {\it {\bf Q2}} that is framed in terms of \emph{experimentally accessible} quantities, and can be directly applied to the developing field of quantum devices and technologies. In what follows we describe the key specific features of this framework.}

\subsection{The optimal spin-inversion channel}

We first address question {\it {\bf Q2}} by studying in detail the problem of approximating spin polarisation inversion for spin-$j_A$ system $A$ with $2j_A+1$-dimensional Hilbert space~$\h_A$. The higher-dimensional spin angular momentum observables $\vv{J}_A:=(J^x_{A},J^y_{A},J^z_{A})$ along the three Cartesian coordinates generate rotations corresponding to elements $g\in \mathrm{SU(2)}$, which act on the system via the unitary representations $U_A(g)$ describing the underlying symmetry principle. A channel $\E:\B(\h_A)\rightarrow\B(\h_A)$ is symmetric under rotations, or \emph{SU(2)-covariant}, if it satisfies Eq.~\eqref{eq:covariance_intro} for all states $\rho_A\in \B(\h_A)$ and $g\in \mathrm{SU(2)}$ (since the input and output systems are the same we have $B=A$). Now, rotational invariance ensures that the symmetric channel $\E$ acts on single spin systems isotropically. As a result, spin polarisation vector $\vv{P}(\rho_A)$ of an initial state $\rho_A$ is simply scaled by the action of $\E$, i.e.,
\begin{equation}
\vv{P}(\E(\rho_A)) = f(\E) \vv{P}(\rho_A)
\end{equation}
for a single parameter $f(\E)$ that is independent on~$\rho_A$ or the spatial direction. The question {\it {\bf Q2}} thus amounts to determining the symmetric quantum channel $\S_{-}$ with coefficient $f(\S_{-})$ that is as close as possible to $-1$ (which can only be achieved by the unphysical spin-inversion operation). 

As the set of all symmetric channels is convex, this becomes a convex optimisation problem whose solution is attained on the boundary of the set. The convex structure of SU(2)-symmetric quantum channels on spin systems has been previously examined by Nuwairan in Ref.~\cite{nuwairan2013su2}, where a characterisation of extremal channels is given. We review these results in Sec.~\ref{sec:su2_structure} and extend the analysis in terms of the Liouville and Jamio{\l}kowski representations of channels (see Sec.~\ref{sec:notation} for details). This, in turn, allows us to directly compute the scaling factors $f(\E)$ for any symmetric channel. 

The convex set of SU(2)-covariant quantum channels on a spin-$j_A$ system forms a simplex with $2j_A+1$ vertices, each corresponding to a CPTP map $\E^{L}$ labelled by an integer $L\in \{0,\dots,2j_A\}$. Therefore, any such symmetric channel $\E$ is a convex combination of these extremal covariant channels:
\begin{equation}
	\E = \sum_{L=0}^{2j_A} p_{L} \E^{L},
\end{equation}
where $\{p_{L}\}_{L=0}^{2j_A}$ forms a probability distribution. 

The following result gives the best physical approximation to spin-inversion, and is proved and generalised to different input and output systems in Theorem~\ref{thm:spinreversalmax} of Sec.~\ref{sec:reversal_B}.
\begin{result}
	The optimal spin polarisation inversion channel is achieved by \mbox{$\S_{-} := \E^{2j_A}$}, the extremal point of SU(2)-covariant channels with the largest dimension $2j_A+1$ of the environment required to implement it. It results in an inversion factor:
	\begin{equation}
		f(\S_{-})=-\frac{j_A}{j_A +1} = -1+  O(1/j_A).
	\end{equation}	
\end{result}

This generalises the previous result on optimal approximations of universal-NOT under rotational symmetry, and determines a fundamental limit that quantum theory imposes on the specific task of (universally) inverting the spin of a quantum system. The higher the dimension of the system, the larger is the maximal spin-inversion factor. Specifically, the optimal channel $\S_{-}$ in the limit $j_A\rightarrow \infty$ approaches $f(\S_{-})\rightarrow -1$, which is the value obtained under the inactive spin-inversion transformation. However, this feature alone does not imply that the channel $\S_{-}$ behaves more like spin-inversion as the dimension of the system increases. As shown previously \cite{marvian2014extending}, once one goes beyond unitary dynamics, the angular momentum observables do not provide a complete description of symmetry principles and information-theoretic aspects become crucial. 

To explicitly quantify this aspect, in Sec.~\ref{sec:inversion_reversal} we compare the fidelity between the output of an active symmetric channel versus the passive transformation of spin-inversion~$\mathcal{T}$. We restrict to input states $\rho_A$ within the convex hull of spin coherent states as these behave classically in the sense of saturating the Heisenberg bound. We find that the output fidelity is given by
\begin{equation}
	F(\E(\rho_A), \T(\rho_A)) = p_{2j_A} \left(1 - \frac{2j_A}{1+4j_A}\right),
\end{equation}
which is maximised whenever $p_{2j_A} = 1$, i.e., whenever $\E$ coincides with the optimal spin-inversion channel $\S_{-}$. Notice that while $f(\S_{-})$ approaches $-1$ as we increase $j_A$, the fidelity only achieves $F(\S_{-}(\rho_A), \T(\rho_A)) \rightarrow 1/2$ in the limit, with the highest bound occurring for $j_A = 1/2$. In other words, the actions of the symmetric channel $\E$ and the passive transformation $\T$ on quantities beyond~$\mathbf{P}(\rho_A)$ distinguish the two, and limit the fidelity at the state level.


\subsection{Spin amplification}

The simple structure of the extremal points of SU(2)-covariant channels generalises to the situation where the input and output spaces correspond to different irreducible spin systems. We discuss all these aspects in Sec.~\ref{sec:su2_structure}, and extensions to general compact Lie groups in Sec.~\ref{sec:general_structure}. The convex set of symmetric channels \mbox{$\E:\B(\h_A) \rightarrow \B(\h_B)$}, where $\h_A$ and $\h_B$ are Hilbert spaces for spin-$j_A$ and spin-$j_B$ systems, forms a simplex now with \mbox{$2\max(j_A, j_B) + 1$} extremal points. In this scenario, it also holds that the spin polarisation of any input state is scaled isotropically by a constant parameter $f(\E)$, which depends only on the particular symmetric channel $\E$. While for $j_A = j_B$, it was always the case that $f(\E)\leq 1$, this no longer holds true for \mbox{$j_{B} > j_A$}, and the spin can be amplified under a symmetric open dynamics. The ultimate limits of this are derived in Theorem~\ref{thm:spinamplify}, and are summarised as follows.

\begin{result}
	Let us denote by $\kappa_{+} = {\max_{\E}} f(\E)$, where the maximisation occurs over the convex set of SU(2)-covariant channels $\E:\B(\h_A) \rightarrow \B(\h_B)$. Then the maximal spin-amplification factor $\kappa_{+}$ is given by:
	\begin{subequations}
		\begin{align}
			\kappa_{+} = \frac{j_B}{j_A} \quad &\mathrm{~for~} j_A\geq j_B,\\
		\kappa_{+} = \frac{j_B+1}{j_A+1} \quad &\mathrm{~for~} j_A<j_B.		
		\end{align}
	\end{subequations}
\end{result}

The above result may initially seem paradoxical: using purely rotationally invariant transformations on a quantum system, we are free to \emph{arbitrarily increase} the expectation value of angular momentum. This provides a dramatic example of the disconnect between symmetry principle and conservation laws. This surprising spin-amplification effect requires that the dynamics is not unitary, but is instead given by a quantum channel with non-trivial Kraus rank, and the intuitions we acquire while dealing with unitary evolution fail badly when we look at more general open quantum dynamics.

But where does this new angular momentum come from? Here, the ability to perform approximate spin-inversion comes in. Any symmetric quantum channel can be purified to a Stinespring dilation involving a symmetric unitary $V$ and an environment $E$ in a pure state $|\eta\>_E$ with zero angular momentum \cite{scutaru, keyl1999optimal, marvian-thesis},
\begin{equation}
\E(\rho_A) = \tr_C V (\rho_A \otimes |\eta\>_E\<\eta|)V^\dagger,
\end{equation}
where we have that $AE$ and $ BC$ denote the two different ways of factoring the global system.
Since angular momentum is exactly conserved across the joint system $AE$ we see that we must have 
\begin{equation}\label{conserve}
\mathbf{P}(\rho_A) = \mathbf{P}(\E(\rho_A)) + \mathbf{P}(\tilde{\E}(\rho_A)),
\end{equation}
where $\tilde{\E}$ denotes the complementary channel to $\E$ obtained by tracing out $B$ after the action of the global unitary $V$ \cite{watrous2018theory}. We now see that spin-inversion and spin-amplification are complementary to each other. Namely, given any spin-amplification for which \mbox{$f(\E) >1$}, Eq.~\eqref{conserve} necessarily implies that the complementary channel must have \mbox{$f(\E) <0$}, and thus is a spin-inversion channel. Some of these features have been discussed previously from the perspective of asymmetry theory~\cite{marvian2014modes}, and earlier in relation to optimal cloning and the universal-NOT gate \cite{vanEnk}. In particular, the complementary channel of the optimal spin polarisation inversion channel $\S_{-}$ will be the maximal spin amplification \mbox{$\tilde{\S}_{-} :\B(\h_A) \rightarrow\B(\h_{B})$} between a spin $j_A$ system and a spin \mbox{$j_B = 2j_A$} system. This generalises to optimal spin polarisation inversion channels between spin systems of different dimensions.

From the perspective of asymmetry theory, every resource measure is monotonically non-increasing under symmetric channels, and thus the fact that polarisation can be increased implies that spin polarisation cannot be a proper measure of asymmetry~\cite{marvian2014modes}. The polarisation may increase, but its ability to encode a spatial direction must become inherently noisier. This is also in agreement with the No-Stretching Theorem \cite{NoStretching} for spin systems.

\subsection{Conservation laws vs decoherence: Quantitative trade-off relations}

Starting from {\it {\bf Q2}}, we analysed to what degree a spin-inversion is possible within quantum theory. This led us to consider symmetric quantum channels and we found that both spin-inversion and spin-amplification are directly related and can be approximately performed under the symmetry constraint. These two examples are maximal disconnects between symmetric dynamics and conservation laws, and thus bring us to the broader issue of question {\it {\bf Q1}}. 

In order to address it properly, we first need to define measures quantifying the deviations from conservation laws and from unitary dynamics. We also generalise the discussion to symmetries described by an arbitrary compact Lie group $G$, and introduce quantitative measures for probing how much the conserved charges associated with symmetry generators, $\{J_A^k\}_{k=1}^{n}$ and $\{J_B^k\}_{k=1}^{n}$, can fluctuate between initial and final states, $\rho_A$ and $\E(\rho_A)$, for a $G$-covariant channel $\E$. To that end, in Sec.~\ref{sec:deviations} we introduce the notion of \emph{average total deviation} from a conservation law, which we define as the average $L_2$ norm of the difference in expectation values between \mbox{$\psi =\ketbra{\psi_A}{\psi_A}$} and $\E(\psi)$ of the generators. Explicitly:
\begin{equation}
	\Delta(\E) :=  \sum_{k=1}^n \int |\Tr(\E(\psi)J_B^k - \psi J_A^k) |^{2} d\, \psi,
\end{equation}
where the integration is with respect to the standard Haar measure on pure states.

To quantify how close a channel $\E$ is to a unitary dynamics we employ the notion of \emph{unitarity}, first defined in Ref.~\cite{wallman2015estimating}. It is defined as the average output purity over all pure states with the identity component subtracted, i.e., 
\begin{equation}
	u(\E) := \frac{d_A}{d_A-1}\int \Tr\left(\E\left(\psi-\frac{\iden_{A}}{d_A}\right)^2\right)\, d\,\psi,
\end{equation}
and satisfies $u(\E)\leq 1$ with equality if and only if $\E$ is a unitary channel. Note that previously this was defined only for channels between the same input and output spaces but, as we explain in Sec.~\ref{sec:deviations}, the definition can be generalised. We also provide a simple characterisation of unitarity in terms of the complementary channel, describing the back-flow of information from the environment, and relate it to the conditional purity of the corresponding Jamio{\l}kowski state. These results, which may be of independent interest, can be summarised as follows.
\begin{result}
	Let $u(\E)$ be the unitarity of an arbitrary quantum channel $\E$ from input system $A$ to output system $B$, then
\begin{enumerate}
\item (Purity representation) 
	\begin{equation}
		u(\E) = \frac{d_A^2}{d_A^2-1} \gamma_{B|A}({\J(\E)}),
	\end{equation}
	where $\gamma_{A|B}(\rho) := \gamma(\rho_{AB}) - \frac{1}{d_A} \gamma(\rho_B)$ is the conditional purity of a bipartite state, with $\gamma(\rho):=\Tr(\rho^2)$, and $\J(\E)$ is the Jamio{\l}kowski state of quantum channel $\E$.
\item (Complementary channel representation)
	\begin{align}
		u(\E) &= \frac{d_A}{d_A^2-1} \left(d_A\Tr(\tilde{\E}(\iden_A/d_A)^2)\right.\non
		&\left.\qquad\qquad\quad\quad - \Tr(\E(\iden_A/d_A)^2)\phantom{\tilde{E}}\!\!\!\!\!\right).
	\end{align}
	where $\tilde{\E}$ is the complementary channel to $\E$ in any Stinespring dilation. 
\item (Zero decoherence) We have that $u(\E) =1$ if and only if $\E$ is an isometry channel.
\end{enumerate}

\end{result}
Thus, unitarity can be understood both as a purity-based measure of correlations in the Jamio{\l}kowski state, or alternatively as a trade-off between the output purities for the channel and its complement. This result is independent of symmetry-based questions and holds for arbitrary quantum channels.

\emph{When do conservation laws hold?} For a unitary symmetric dynamics, the corresponding conservation laws will always hold, but generally this is no longer true for symmetric quantum channels. There will be situations, however, when the degrees of freedom that decohere through interactions with the environment have no effect on the expectation values of the generators. In Sec.~\ref{sec:lowerbounds} we give the most general form of such a covariant channel that is unital and for which conservation laws always hold. Such behaviour would require the presence of decoherence-free subspaces, so that the information is protected from leaking into the environment. It follows that conservation laws will hold for symmetric dynamics that protects the degrees of freedom associated with the symmetry generators from leaking the information into the environment. More precisely, suppose that $\{J_A^k\}_{k=1}^{n}$ generate a unitary representation $U_A$ acting on the Hilbert space $\h_A$ that describes the quantum system. Any symmetric channel $\E:\B(\h_A)\rightarrow\B(\h_A)$ for which $\Delta(\E) = 0$ will protect the subspace \mbox{$\mathfrak{S}:={\rm{span}} \{ \iden, J_A^{k}\}\subset \B(\h_A)$}, so $\E(\rho_A) = \rho_A$ for any state $\rho_A$ in $\mathfrak{S}$. In this sense, conservation laws may be viewed as a form of information preserving structures \cite{blume2010information}.

Consider also a simple example of a two-qubit system $AA'$, where only $A$ carries spin angular momentum, so the symmetry generators are $J^x_A\otimes \iden_{A'}$, $J^y_A\otimes \iden_{A'}$ and $J^z_A\otimes \iden_{A'}$. Any channel of the form $\E_{AA'} = \I_A\otimes \E_{A'}$ is symmetric, with $\I_A$ the identity channel on system $A$ and $\E_{A'}$ an arbitrary quantum channel on system $A'$. Moreover, $\E_{AA'}$ satisfies $\Delta(\E_{AA'}) = 0$, so that the associated conservation laws hold despite the fact that $\E_{AA'}$ can be arbitarily far from unitary dynamics. This example illustrates that probing conservation laws for a physical realisation of symmetric dynamics will not always be sufficient to decide whether there are decoherence effects present. In other words, robustness of conservation laws does not occur for all types of systems. Nevertheless, there are regimes that guarantee robustness for conservation laws. In such cases, approximate conservation laws hold if and only if the dynamics is close to a unitary symmetric evolution. For example, whenever $\B(\h_A)$ contains a single trivial subspace then there is no symmetric channel other than identity for which conservation laws hold (which is the case, e.g., when $\h_A$ carries an irreducible representation of SU(2)).

\emph{What does it mean for conservation laws to be robust under decoherence?} If for all channels $\E$ obeying a given symmetry principle, it holds that $\Delta(\E) \approx 0$ if and only if $u(\E) \approx 1$, we say that the associated conservation laws are robust. This can be established by finding upper and lower bounds on the deviation $\Delta(\E)$ that coincide when $u(\E) \rightarrow 1$. In Sec.~\ref{sec:upperbounds} we show in Theorem~\ref{thm:upperbounds} that for all types of symmetries described by connected compact Lie groups, one can find such an upper bound (and the result extends to different input and output systems).
\begin{result}
	Given any connected compact Lie group, for a symmetric channel $\E$ approximating a symmetric unitary the associated conservation laws will hold approximately. In other words, there exists an upper bound on the deviation from conservation law in terms of unitarity:
	\begin{equation}
		\Delta(\E) \leq M(1-u(\E))
	\end{equation}
	for some constant $M>0$ \cc{that is independent of $\E$, and depends only on the dimensions of the systems involved and the symmetry generators.}
\end{result}

In order to obtain lower bounds, however, additional assumptions are required. It is clear from the previous discussion that conservation laws can hold beyond unitary dynamics, and in those situations we cannot expect to obtain lower bounds on the deviation in terms of unitarity. However, there exist symmetries for which conservation laws only hold for symmetric unitary dynamics, and then robustness is achieved. This happens in the case of spin-$j$ system with symmetry generators given by higher-dimensional spin angular momenta generating an irreducible representation of SU(2). We prove the following result in Theorem~\ref{thm:su2_bounds} of Sec.~\ref{sec:lowerbounds}.
\begin{result}
	For a spin-$j$ system, spin angular momentum conservation laws are robust to noise described by a symmetric channel $\E$ and the following bounds hold:
	\begin{subequations}
		\begin{align}
			\sqrt{\Delta(\E)}\geq&\frac{\sqrt{2}j^{1/2}}{(2j+1)^2}(1-u(\E)),\\
			\sqrt{\Delta(\E)}\leq&\frac{3\sqrt{2} j^{3/2}}{2j+1} (1-u(\E)).
		\end{align}		
	\end{subequations}
\end{result}

More generally, we prove in Theorem~\ref{thm:lowerbounds} that whenever the quantum system carries a representation $U_A$ of a Lie group $G$ for which $U_A\otimes U_A^{*}$ has a multiplicity-free decomposition, then the associated conservation laws are robust under any open system dynamics given by the symmetric channel \mbox{$\E:\B(\H_A)\rightarrow\B(\H_A)$}.

Finally, in Sec.~\ref{sec:U1bounds} we obtain specific upper bounds on the deviation from a conservation law for energy that generates a U(1) symmetry constraint, in terms of the unitarity of the U(1)-symmetric channel. We also explain why a lower bound cannot hold because of the many multiplicities that appear in the decomposition of $\B(\h_A)$. This analysis relies on the structure of convex set of U(1)-covariant channels, which we expand on in Sec.~\ref{sec:u1_structure}.


\section{Notation and preliminaries}
\label{sec:notation}
	
\subsection{Quantum channels and their representations}
	
A state of a finite-dimensional quantum system $A$ is described by a density operator $\rho_A\in\B(\H_A)$, with $\B(\H_A)$ denoting the space of bounded operators on a $d_A$-dimensional Hilbert space $\H_A$, that also satisfies \mbox{$\rho_A\geq 0$} and \mbox{$\tra{\rho_A}=1$}. The space $\B(\H_A)$ is itself a Hilbert space with the Hilbert-Schmidt inner product $\langle X,Y\rangle=\tra{X^\dagger Y}$. General evolution between $d_A$-dimensional and $d_B$-dimensional quantum systems is described by a quantum channel $\E$ given by a linear superoperator \mbox{$\E:\B(\h_A)\rightarrow\B(\h_B)$} that is completely positive and trace-preserving (CPTP). More broadly, we will also consider CP maps, i.e., linear superoperators that are only completely positive (CP), but not trace-preserving (TP). A quantum channel $\E^\dagger$ is called the \emph{adjoint} of $\E$ if for all $X\in\B(\H_A)$ and $Y\in\B(\H_B)$ we have
\begin{equation}
	\label{eq:adjoint}
	\tra{\E(X)Y}=\tra{X \E^\dagger(Y)}.
\end{equation}
Closed dynamics is described by a unitary channel \mbox{$\V(\cdot)=V(\cdot) V\hc$}, where $V$ is a unitary operator. 
	
The \emph{Liouville representation} of $X\in\B(\h_A)$ is defined by a unique column vector $\vect{X}\in\mathbb{C}^{d_{A}^2}$ (as opposed to vectors in $\H_A$ denoted by $\ket{\cdot}$) with entries given by the inner product $\tr(T_{k}\hc X)$, where $\{T_k\}_{k=1}^{d_A^2}$ is a fixed orthonormal basis of $\B(\h_A)$. By analogously denoting a fixed orthonormal basis of $\B(\h_{B})$ by \mbox{$\{S_k\}_{k=1}^{d_B^2}$}, the Liouville representation of the superoperator $\E:\B(\h_A)\rightarrow\B(\h_B)$ is a $d_B^2$ by $d_A^2$ matrix $L(\E)$ defined uniquely via the relation:
\begin{equation}
	L(\E)\vect{X}=\vect{\E(X)}
\end{equation}
for any $X\in \B(\h_A)$. It is then straightforward to show that the entries of $L(\E)$ are given by
\begin{equation}
	\label{eq:liouville}
	\!\!\!\!L(\E)_{jk}\!=\!\langle\!\bra{S_j}L(\E)\vect{T_{k}}\!=\!\langle\!\langle S_{j} \vect{\E(T_{k}) }\!=\!{\Tr}(S_{j}\hc\E(T_k)).
\end{equation}
Note that, in the Liouville representation, the composition of quantum channels becomes matrix multiplication, i.e., $L(\E\circ \F)=L(\E)L(\F)$.
	
One can also represent a quantum channel $\E$ via its \emph{Jamio{\l}kowski state} \mbox{$\J(\E)\in\B(\H_B)\otimes\B(\H_A)$} defined by
\begin{equation}
	\label{eq:jamiolkowski}	
	\J(\E):=\E\otimes\I_A \ketbra{\Omega}{\Omega},\quad \ket{\Omega}=\frac{1}{\sqrt{d_A}}\sum_{j=1}^{d_A} \ket{jj},
\end{equation}
where $\I_A$ denotes the identity channel acting on $\B(\H_A)$. The condition for complete positivity of $\E$ is equivalent to the positivity of $\J(\E)$, while the trace-preserving property of $\E$ correspond to \mbox{$\trb{B}{\J(\E)}=\iden_A/d_A$}. We note that we may pass from the Liouville representation to the Jamio{\l}kowski representation via
\begin{equation}
	\label{eq:reshuffling}
	L(\E)^R = \J(\E),
\end{equation}
where $R$ is the \emph{reshuffling} operation defined as the linear operation for which \mbox{$\ketbra{ab}{cd}^R = \ketbra{ac}{bd}$} for all computational basis states.
	
Finally, any quantum channel $\E$ admits a \emph{Stinespring representation} in terms of an isometry \mbox{$V:\h_A\rightarrow\h_{B}\otimes \h_{E}$} with $\h_{E}$ describing the environment system such that
\begin{equation}
	\label{eq:stinespring}
	\E(X)= \Tr_{E} (V X V\hc)
\end{equation}
for all $X\in \B(\h_A)$. The isometry $V$ that defines the quantum channel $\E$ is unique up to a local isometry on the environment. Note that, using the above, the adjoint channel $\E^\dagger$ is given by
\begin{equation}
	\label{eq:stinespring_adjoint}
	\E^\dagger(Y)=V^\dagger(Y\otimes \iden_E)V
\end{equation}
for all $Y\in \B(\h_{B})$.
	
Stinespring representation allows one to introduce the concept of a \emph{complementary channel}: a quantum channel $\tilde{\E}$ is complementary to $\E$, defined by Eq.~\eqref{eq:stinespring}, if its action is given by
\begin{equation}
	\label{eq:complementary}
	\tilde{\E}(X) = \Tr_{B}(V X V\hc),
\end{equation}
We also note that the adjoint of the complementary channel, which we denote by $\tilde{\E}\hc$, is given by
\begin{equation}
	\label{eq:adjoint_complementary}
	\tilde{\E}\hc(X) = V\hc (\iden_{B}\otimes X)V
\end{equation}
for all $X\in \B(\h_E)$.
	
	
\subsection{Symmetries and $G$-covariant channels}
\label{sec:introsymmops}
	
Consider a group $G$ that acts on $\h_A$ and $\H_{B}$ via unitary representations \mbox{$g\rightarrow U_A(g)$} and \mbox{$g\rightarrow U_{B}(g)$}, so that the group action on quantum states is given by unitary channels \mbox{$\U_A^g(\cdot):=U_{A}(g)(\cdot) U_{A}^\dagger(g)$} and \mbox{$\U_{B}^g(\cdot)=U_{B}(g)(\cdot) U_{B}^\dagger(g)$}. Recall that every finite-dimensional unitary representation on a Hilbert space is the direct sum of irreducible representations, or irreps. We say that a quantum system $A$ is an \emph{irreducible system} if $\H_A$ carries an irrep of~$G$, i.e., if $\H_A$ has no non-trivial subspace closed under the action of $U_A(g)$. 
	
We say that a quantum channel $\E:\B(\h_A)\rightarrow\B(\h_B)$ is \emph{$G$-covariant} (or simply that it is a \emph{symmetric} channel when the group $G$ is fixed) if it satisfies
\begin{equation}
	\label{eq:covariance}
	\forall g\in G:~\U_{B}^{g\dagger}\circ\E\circ \U_{A}^{g}=\E.
\end{equation}

To explain how the covariant constraint affects different representations of quantum channels, we rely on the following well-known result \cite{hall2013lie}.

\begin{lemma}[Schur's lemma]
	\label{lem:schur}
	Let $U(g)$ be an irreducible representation of a group $G$ on a Hilbert space $\H$. Then, any operator $X\in\B(\H)$ satisfying $[X,U(g)]=0$ for all $g$ is a scalar multiple of identity on $\H$. Moreover, if $V(g)$ is another inequivalent representation of $G$, then \mbox{$U(g)YV^\dagger(g)=Y$} for all $g$ implies $Y=0$. 
\end{lemma}
	
Let us start with the structure of the Liouville representation of $G$-covariant channels.
\begin{theorem}
	\label{thm:liouville}
	Let $U_A(g)$ and $U_B(g)$ be the unitary representations of $G$ on $\H_A$ and $\H_B$. Then, the Liouville representation of a $G$-covariant channel \mbox{$\E:\B(\H_A)\rightarrow\B(\H_B)$} is given by
	\begin{equation}
		\label{eq:liouville_decomp}
		L(\E)=\bigoplus_{\lambda}\mathbf{\iden}^{\lambda}\otimes L^{\lambda}(\E),
	\end{equation}
	where $\lambda$ ranges over all irreps that appear in both irrep decompositions of tensor representations \mbox{$U_{A}(g)\otimes U^*_{A}(g)$} and \mbox{$U_{B}(g)\otimes U^*_{B}(g)$}, $\iden^{\lambda}$ are the identity matrices acting within the irrep subspaces, and $L^{\lambda}$ denote non-trivial $m_B^{\lambda}\times m_A^\lambda$ block matrices acting on the multiplicity spaces.	
\end{theorem}
\begin{proof}
	First, using the Liouville representation, the covariance condition is equivalent to
	\begin{equation}
		\label{eq:covariance_liouville}
		\forall g\in G:~L(\U_{B}^{g\dagger})L(\E)L(\U_A^g)=L(\E).
	\end{equation}
	Note that $L(\U_{A}^g)=U_{A}(g)\otimes U_{A}^{*}(g)$ is itself a (tensor) representation of $G$, and an analogous statement holds for $L(\U_{B}^g)$. Therefore, we can decompose them into irreps as 
	\begin{subequations}
		\begin{align}
			\label{eq:tensor_decomp_A}
			U_{A}(g)\otimes U^*_{A}(g)&= \bigoplus_{\lambda}V^{\lambda}(g)\otimes \iden_{m_A^{\lambda}},\\
			\label{eq:tensor_decomp_B}
			U_{B}(g)\otimes U^*_{B}(g)&= \bigoplus_{\lambda}V^{\lambda}(g)\otimes \iden_{m_B^{\lambda}},
		\end{align}
	\end{subequations}
	where $\lambda$ ranges over all irreps that appear in each decomposition, and the group acts trivially on the multiplicity spaces of dimensions $m_A^\lambda$ and $m_B^\lambda$. Now, since the covariance condition means that $L(\E)$ commutes with group representations having the above decompositions, the Schur's lemma implies that $L(\E)$ acts non-trivially only on the multiplicity spaces, leading to the decomposition given in Eq.~\eqref{eq:liouville_decomp}.
\end{proof}
	
Next, let us proceed to the Jamio{\l}kowski representation of a covariant channel $\E$. 
\begin{theorem}
	\label{thm:jamiolkowski}
	Let $U_A(g)$ and $U_B(g)$ be the unitary representations of $G$ on $\H_A$ and $\H_B$. Then, the Jamio{\l}kowski representation of a $G$-covariant channel \mbox{$\E:\B(\H_A)\rightarrow\B(\H_B)$} is given by
	\begin{equation}
		\label{eq:jamiolkowski_decomp}
		\J(\E)=\bigoplus_{\lambda}\mathbf{\iden}^{\lambda}\otimes \J^{\lambda}(\E),
	\end{equation}
	where, $\lambda$ ranges over all irreps that appear in the irrep decomposition of tensor representation \mbox{$U_{B}(g)\otimes U_{A}^{*}(g)$}, $\iden^{\lambda}$ are the identity matrices acting within the irrep subspaces, and $\J^{\lambda}$ denote non-trivial square matrices of size $m_{BA}^{\lambda}\times m_{BA}^{\lambda}$ that act on the multiplicity spaces.
\end{theorem}
\begin{proof}
	The covariance condition means that for all $g\in G$ we have
	\begin{equation}
		\label{eq:covariance_jamiol_part_1}
		(U^\dagger_{B}(g)\otimes \iden_A)\J(\E\circ \U_{A}^{g})(U_{B}(g)\otimes \iden_A)=\J(\E).
	\end{equation}
	By employing the fact that for any unitary $U$ we have \mbox{$U\otimes\iden \ket{\Omega}=\iden\otimes U^{*\dagger} \ket{\Omega}$}, we get
	\begin{equation}
		\label{eq:covariance_jamiol_part_2}
		\J(\E\circ \U_{A}^{g})=(\iden_B\otimes U_A^{*\dagger}(g))\J(\E)(\iden_B\otimes U^*_{A}(g)).
	\end{equation}
	Combining the above two equations we find that covariance of $\E$ is equivalent to $\J(\E)$ satisfying the following commutation relation:
	\begin{equation}
		\label{eq:covariance_jamiol}
		\forall g\in G:~[\J(\E),U_{B}(g)\otimes U^*_{A}(g)]=0.
	\end{equation}
	As in the proof of Theorem~\ref{thm:liouville}, we can decompose the tensor representation appearing in the above commutator into irreps,
	\begin{equation}
		\label{eq:tensor_decomp_AB}
		U_{B}(g)\otimes U_{A}^{*}(g)= \bigoplus_{\lambda}V^{\lambda}(g)\otimes \iden_{m_{BA}^{\lambda}}.
	\end{equation}
	Once again, by using the Schur's lemma, we arrive at the block-diagonal decomposition of $\J(\E)$ given in Eq.~\eqref{eq:jamiolkowski_decomp}.
\end{proof}
	
Finally, there is also a very particular form of the Stinespring representation of a $G$-covariant channel given by the following theorem. 
\begin{theorem}
	\label{thm:stinespring}
	Given a $G$-covariant channel $\E$, there exists an environment system $E$, with a Hilbert space $\h_{E}$ and a unitary representation \mbox{$U_E(g)$}, together with a $G$-covariant isometry \mbox{$V:\h_A\rightarrow \h_{B}\otimes \h_{E}$}, such that:
	\begin{equation}
		\label{eq:covariance_stinespring}
		\E(X) = \trb{E}{V X V^\dagger}
	\end{equation}
	for all $X\in \B(\h_A)$.	
\end{theorem}
\noindent The proof of the above result can be found in Ref.~\cite{keyl1999optimal}.
	
	
\subsection{Irreducible tensor operators}
	
The set of operators $\{T^{\lambda,\alpha}_{k}\}_{\lambda,\alpha,k}$ in $\B(\h_A)$ are called \emph{irreducible tensor operators (ITOs)} if they transform irreducibly under the group action,
\begin{equation}
	\label{eq:ITOs}
	\U_A^g(T^{\lambda,\alpha}_{k})=\sum_{k'}v^{\lambda}_{k'k}(g)T^{\lambda,\alpha}_{k'},
\end{equation}
where $\lambda$ labels irreducible representations of $G$ with matrix elements $v^{\lambda}_{kk'}$, and $\alpha$ denotes multiplicities. From the above property it can be deduced via Schur's orthogonality theorem that the set of ITOs must be orthonormal,
\begin{equation}
	\label{eq:ITOs_ortho}
	\Tr((T^{\lambda',\alpha'}_{k'})\hc T^{\lambda,\alpha}_{k})\propto\delta_{\lambda\lambda'}\delta_{\alpha\alpha'}\delta_{kk'}.
\end{equation}
Throughout the paper we will denote the normalised ITOs for the input system, living in $\B(\H_A)$, by $T^{\lambda,\alpha}_{k}$, and the normalised ITOs for the output system, living in $\B(\H_{B})$, by $S^{\lambda,\alpha}_{k}$. 
	
These yield symmetry-adapted bases for $\B(\h_A)$ and $\B(\h_B)$ that are particularly useful for the studies of $G$-covariant channels. More precisely, by employing the block diagonal structure of the Liouville representation for such channels stated in Theorem~\ref{thm:liouville}, and using the defining property of ITOs, we have
\begin{equation}
	\label{eq:action_on_ITOs}
	L(\E) \vect{T^{\lambda,\alpha}_k}= \sum_{\beta} L^\lambda_{\beta\alpha}(\E) \vect{S^{\lambda,\beta}_k}.
\end{equation}
Moreover, since ITOs are orthonormal, any density matrix in $\B(\h_A)$ (and analogously for $\B(\h_B)$) can be written as
\begin{equation}
	\label{eq:ITO_decomp}
	\rho_A=\frac{\iden_A}{d_A}+\sum_{\lambda,\alpha} \mathbf{r}^{\lambda,\alpha}\cdot \mathbf{T}^{\lambda,\alpha},
\end{equation}
where we denoted the vector of ITOs transforming under a $\lambda$-irrep by $\mathbf{T}^{\lambda,\alpha}=(T^{\lambda,\alpha}_{1},...,T^{\lambda,\alpha}_{d_\lambda})$, with $d_\lambda$ being the dimension of the $\lambda$-irrep. 
	
	
\subsection{Continuous symmetries and conserved charges}
	
Continuous symmetries of the system $A$ are related to compact Lie groups. The representation of such a group $G$ can be generated by \emph{infinitesimal generators} $\{J_{A}^k\}_{k=1}^n$. For simply connected Lie groups, representations of the group are in a one-to-one correspondence with representations of the Lie algebra $\mathfrak{g}$ via the exponentiation map. More precisely, we have
\begin{equation}
	\label{eq:lie_action}
	U_A(\vv{g})=e^{i\vv{J}_A\cdot \vv{g}}
\end{equation}
with $g_k\in\mathbb{R}$ continuously parametrizing the group action. In such a Lie algebraic setting, by considering infinitesimal group action, $g_k\rightarrow 0$, one can show that the covariance of a linear map $\E:\B(\h_A)\rightarrow\B(\h_{B})$, specified by Eq.~\eqref{eq:covariance}, is equivalent to 
\begin{equation}
	\label{eq:lie_commutation}
	[\E(X),J_{B}^{k}] = \E([X,J_A^k])
\end{equation}
for all $k\in\{1,\dots,n\}$ and $X\in\B(\h_A)$, with \mbox{$[X,Y]$} denoting a commutator. 
	
By taking the Liouville representation of the operators on both sides of the above equality and employing the identity \mbox{$\vect{XYZ}=X\otimes Z^{*\dagger} \vect{Y}$}, one can alternatively express the covariance condition as
\begin{equation}
	\label{eq:covariance_lie}
	L(\E)(\iden_A\otimes J_{A}^{k*}-J_{A}^{k}\otimes\iden_A)=(\iden_B\otimes J_{B}^{k*}-J_{B}^{k}\otimes\iden_B)L(\E),
\end{equation}
for all $k$. In particular, for a unitary $G$-covariant channel $\V:\B(\H_A)\rightarrow\B(\H_A)$, the condition becomes simply $[V,J_A^k]=0$. As a result, for all $k$ and for all quantum states $\rho_A\in\B(\H_A)$ we have
\begin{equation}
	\label{eq:conservation}
	\tra{\V(\rho_A) J_A^k}=\tra{\rho_A J_A^k},
\end{equation}
i.e., the generators of the symmetry, $\{J_A^k\}_{k=1}^n$, give the conserved (Noether) charges.

\begin{table*}
	\centering
	\cc{\begin{tabular}{ |c|c|c|c| } 
		\hline
		&Closed unitary evolution $U$& Open channel evolution $\E$ & Level \\  \hline
		1&$[U, J_k] = 0$ & $[\E, ad[J_k]]] = 0$& Defining symmetry\\  \hline
		2&$U\hc J_k U= J_k$ & $\E^{\dagger}(J_k) = J_k$ & Dynamical charge conservation  \\  \hline
		3&${\rm{Tr}}(\rho J_k) = {\rm{Tr}}(U\rho U^{\dagger} J_k)  \  \  \forall \rho$ & ${\rm{Tr}}(\rho J_k) = {\rm{Tr}}(\E(\rho) J_k ) \ \forall \rho $ & Charge conservation law \\ \hline
		4&$\Delta(U) = 0$ & $\Delta(\E) = 0$ & No average deviation from conservation
		\\	\hline
	\end{tabular}}
	\caption{ \cc{Definitions of symmetric dynamics and charge conservation for closed and open systems that are used throughout. For closed system $1\iff2\iff3\iff4$ while for open systems $2\iff 3 \iff 4$, but there is no such equivalence with respect to 1 (i.e., defining symmetry). A symmetric dynamics is one that commutes with the action of the generators. In here, $ad[J_k]$ represents the adjoint action given by Eq.~\eqref{eq:lie_commutation}.}}
	\label{table}
\end{table*}

\cc{To be more precise, we can only talk about ``symmetry'' when we have a set of generators  (or representations) that determine exactly what that symmetry principle is. Traditionally, both in quantum and classical mechanics, charge operators are generators of particular symmetry. Mathematically, charge operators act on the system forming a representation of a particular Lie algebra. For unitary dynamics $U$ conservation of charges happens if and only if $U$ commutes with the charge operators. Equivalently, viewed in the Heisenberg picture, charge operators are fixed points of the unitary evolution. The problem is that while for closed systems all these formulations are the same and often interchangable in the literature, this is no longer the case for open systems. This calls out for a precision of language, and so we require that:
	\begin{enumerate}
		\setcounter{enumi}{-1}
		\item Charge operators are generators that define a symmetry group action.
		\item Dynamics commutes with the generators to define a symmetry principle.
		\item Generators are fixed points of the dynamics in the Heisenberg picture and define (dynamical) charge conservation.
		\item Expectation values of generators remain constant under the dynamical evolution of every input state and define charge conservation.
		\item $\Delta (\E) = 0$ defines no average total deviation from a conservation law.
	\end{enumerate}
One should note that these distinctions have also been made for dissipative dynamics described by Lindbladian master equations~\cite{albert2014symmetries}, with different terminology in other works where the symmetry described here was called weak symmetry in Refs.~\cite{zhang2020stationary, buvca2012note}. As we see from Table~\ref{table}, the equivalence of charge conservation in either the Heisenberg or Schr\"{o}dinger picture with no average deviation from a conservation law motivates our focus on this quantity. Therefore, unless one starts talking about \emph{particular} states for which the expectation value of the generators remains unchanged under dynamics, then there is no pressing need to differentiate between formulations 2, 3 and~4. Whether one would like to talk about charge conservation for particular states that is a different question altogether, one that cannot be equivalently related to the state-independent definitions above. }


\section{Deviations from closed dynamics and from conservation laws}
\label{sec:deviations}
	
The main aim of this paper is to quantitatively investigate the deviation from conservation laws as the symmetric dynamics deviates from being closed. In order to achieve this, we obviously need to understand the structure of covariant quantum channels that model symmetric open dynamics, and we will pursue this task from Sec.~\ref{sec:structure} onwards. However, there is also one more crucial ingredient needed for our analysis: namely, we need quantitative measures of how much a given dynamics deviates from being closed, and how much it deviates from satisfying the conservation law. In this section we introduce such measures and provide their basic properties.
	
		
\subsection{Quantifying the deviation from closed dynamics}
	
In order to quantify how much the dynamics generated by a given quantum channel $\E$ deviates from the closed unitary dynamics we employ the notion of \emph{unitarity}. It was originally introduced in Ref.~\cite{wallman2015estimating} as a way to quantify how well a quantum channel preserves purity on average. We extend these results to allow for distinct input and output system dimensions for a quantum channel $\E:\B(\h_A) \rightarrow\B(\h_B)$.
\begin{definition}
	Unitarity of a quantum channel \mbox{$\E:\B(\h_A) \rightarrow\B(\h_B)$} is defined as the average output purity with the identity component subtracted:
	\begin{equation}
		\label{eq:unitarity}
		u(\E) := \frac{d_A}{d_A-1}\int \Tr\left(\E\left(\psi-\frac{\iden_{A}}{d_A}\right)^2\right)\, d\,\psi,
	\end{equation}
	the integral is taken over all pure states \mbox{$\psi=\ketbra{\psi}{\psi}\in \B(\h_A)$} distributed according to the Haar measure.
\end{definition}
	
As we prove in Appendix~\ref{appendix:unitarity}, the above extension of unitarity satisfies the original condition \mbox{$u(\E)\leq 1$} with equality if and only if the operation is an isometry (as opposed to a unitary in the original formulation). This means \mbox{$u(\E) = 1$} is equivalent to the existence of an isometry $V:\h_A\rightarrow\h_B$ such that \mbox{$\E(\rho) = V\rho V\hc$}. Furthermore, as shown by the authors of Ref.~\cite{wallman2015estimating}, unitarity can be efficiently estimated using a process similar to randomised benchmarking, and can be calculated using the Jamio{\l}kowski representation of $\E$. This characterisation through $\J(\E)$ carries over to the extended version we discuss here and, moreover, we find a novel characterisation of $u(\E)$ in terms of the output purity of $\E$ and its complementary channel $\tilde{\E}$. These results are summarised in the following lemma (see Appendix~\ref{appendix:unitarity} for the proof).
\begin{lemma}
	\label{lem:unitaritychoi}
	Unitarity of a channel $\E:\B(\h_A) \rightarrow \B(\h_B)$ can be equivalently expressed by the following relations:
	\begin{align}
		\label{eq:unitarity_jamiolkowski}
		\!\!\!u(\E)&\!=\!\frac{d_A}{d_A^2-1}\left(d_A\gamma(\J(\phantom{\tilde \E}\!\!\!\E))\!-\!\gamma(\E(\iden_A/d_A))\right)\!,\! \\			
		\label{eq:unitarity_complementary}
		\!\!\!u(\E) &\!=\! \frac{d_A}{d_A^2-1}\left(d_A\Tr(\tilde{\E}(\iden_A/d_A)^2) \!-\! \Tr(\E(\iden_A/d_A)^2)\right)\!,\!
	\end{align}
	with $\gamma(\rho) = \Tr(\rho^2)$ denoting the purity of a state $\rho$.
\end{lemma}

Finally, let us remark that Eq.~\eqref{eq:unitarity_jamiolkowski} suggests defining the notion of \emph{conditional purity} for a bipartite system,
\begin{equation}
	\gamma_{B | A}(\rho_{AB}) :=\gamma\left(\rho_{AB}\right)-\frac{1}{d_A} \gamma\left(\rho_{A}\right).
\end{equation}
Then, unitarity of a channel is simply expressed by the scaled conditional purity of its Jamio{\l}kowski state:
\begin{equation}
	u(\E)= \frac{d_A^2}{d_A^2-1} \gamma_{B | A}(\J(\E)).
\end{equation}	

	
\subsection{Quantifying the deviation from conservation laws}
	
Typically, the expectation values of symmetry generators, $\{J^k\}_{k=1}^n$, are not constant under non-unitary $G$-covariant dynamics. In order to quantify this deviation from conservation laws we need to introduce appropriate measures. For any quantum operation $\E$ we define the \emph{directional deviation} $\Delta_k$ for the expectation value of the $J^{k}$ generator with respect to the state $\rho_A$ as
\begin{equation}
	\label{eq:dev_dir_1}
	\Delta_{k}(\rho_A,\E):=\tra{\E(\rho_A)J_B^k-\rho_A J_A^{k}}.
\end{equation}
Note that by introducing \kk{the finite deviation operator}
\begin{equation}
	\label{eq:delta_J}
	\delta J^k_A:=\E^\dagger(J_B^{k})-J_A^k,
\end{equation}
with $\E^\dagger$ denoting the adjoint of $\E$ that describes its action in the Heisenberg picture, we can rewrite Eq.~\eqref{eq:dev_dir_1} as
\begin{equation}
	\label{eq:dev_dir_2}
	\Delta_{k}(\rho_A,\E)=\Tr(\rho_A \delta J^k_A).
\end{equation}
	
As we are equally interested in the deviation from a conservation law for all conserved charges, we define the \emph{total deviation} $\Delta_\mathrm{tot}$ as the $l_2$ norm of directional deviations for all generators:
\begin{equation}	
	\label{eq:total_deviation}
	\Delta_\mathrm{tot}(\rho_A,\E):=\sum_{k=1}^n |\Delta_{k}(\rho_A,\E)|^2.
\end{equation} 
Finally, since we aim at quantifying how much a channel deviates from conservation law, independently of the input state, we introduce the \emph{average total deviation} $\Delta(\E)$:
\begin{equation}
	\label{eq:dev_avg}
	\Delta(\E): = \int d\psi \Delta_{\mathrm{tot}}(\psi,\E) = \sum_{k=1}^n\int d\psi  |\!\bra{\psi}\delta J^k_A\ket{\psi}\!|^2,
\end{equation}
where we integrate with respect to the induced Haar measure over all pure states $\psi\in\H_A$.
	
The above expression for the average total deviation $\Delta$ can clearly be rewritten in the following form
\begin{equation}
	\Delta(\E) = \sum_{k=1}^n \int d\psi \tra{(\psi \otimes \psi) (\delta J_A^k \otimes \delta J_A^k)}.
\end{equation}
Next, we can employ the identity \cite{harrow2013church}
\begin{equation}
	\int d\psi \ \psi^{\otimes N} = \frac{(d_A-1)!}{(d_A+N-1)!} \sum_\pi P_{\pi},
\end{equation}
where $\pi$ is any permutation on $N$ symbols and $P_{\pi}$ is the corresponding Hilbert space unitary. In our case \mbox{$N=2$}, so we only have the identity $\I$ and the flip unitary operation $\F$. Thus,
\begin{align}
	\label{eq:dev_su2_partxxxxx}
	\!\!\!\Delta(\E) &= \frac{1}{d_A(d_A+1)}\sum_{k=1}^n \tra{(\I+\F)(\delta J_A^k \otimes \delta J_A^k)}\nonumber\\
	&= \frac{1}{d_A(d_A+1)} \sum_{k=1}^n\left(\tra{\delta J_A^k}^2 + \tra{(\delta J_A^k)^2}\right). 	
\end{align}
	
	
\section{Convex structure of symmetric channels}
\label{sec:structure}
	
We now proceed to investigate the convex structure of the set of symmetric channels \mbox{$\E:\B(\h_A)\rightarrow\B(\h_B)$}, with a particular focus on its extremal points. We start with a specific example of SU(2)-covariant channels, the convex structure of which was investigated before in Ref.~\cite{nuwairan2013su2}. In this case, we provide a full characterisation of the extremal symmetric channels between irreducible systems, i.e., with Hilbert spaces of the input and output systems, $\h_A$ and $\h_B$, corresponding to spin-$j_A$ and spin-$j_B$ systems with $d_A=2j_A+1$ and $d_B=2j_B+1$.  \cc{We will refer to SU(2) symmetric channels between irreducible systems as \emph{SU(2)-irreducibly symmetric (covariant) channels}, in order to differentiate from the more general SU(2) covariant channels which need not have the extra irreduciblity assumption}. The technical results derived here will be then employed in Sec.~\ref{sec:reversal} to study optimal covariant channels for spin-inversion and spin amplification. Next, we switch to a generic case of a compact group~$G$. Here, we describe a useful decomposition of symmetric channels, which will be crucial in Sec.~\ref{sec:trade_off} to analyse the trade-off between deviations from conservations laws and deviations from closed symmetric dynamics. We also explain how, under the assumption of multiplicity-free decomposition, this leads to a complete characterisation of the extremal points of $G$-covariant channels: the corresponding Jamio{\l}kowski states are then given by normalised projectors onto irreducible subspaces. Finally, we investigate the U(1) group, which is the extreme example of a group that does not have a multiplicity-free decomposition (i.e., since U(1)-irreps are one-dimensional, all the non-trivial dynamics happens within the multiplicity spaces). In this particular case, which is physically relevant due to its connection with conservation law for energy, we find an incomplete set of extremal channels, which is however large enough to generate arbitrary action on the multiplicities of the trivial irrep $\lambda=0$ (which physically encodes the action of the channel on energy eigenstates).
	
	
\subsection{Extremal SU(2)-covariant channels between irreducible systems}
\label{sec:su2_structure}	
	
The Lie group related to rotations of a system $A$ in physical three-dimensional space is the SU(2) group. It has three generators, $\{J_A^x,J_A^y,J_A^z\}$, corresponding to angular momentum operators along three perpendicular axes, which generate general rotations. The unitary representation of such a rotation on the Hilbert space $\H_A$ is given by
\begin{equation}
	\label{eq:su2_action}
	U_A(\vv{g})=e^{i\vv{J}_A\cdot\vv{g}},
\end{equation}
with $g_k\in[0,2\pi]$ parametrising the rotation angles. 
	
Irreducible representations of SU(2) group can be classified according to total angular momentum $j$, which is either an integer or half-integer. The $j_A$-irrep is $(2j_A+1)$-dimensional and the corresponding subspace of $\H_A$ is spanned by $\{\ket{j_A,m}\}_{m=-j_A}^{j_A}$, which are the simultaneous eigenstates of total angular momentum, \mbox{$J_A^2=(J_A^x)^2+(J_A^y)^2+(J_A^z)^2$}, and $J_A^z$, with eigenvalues $j_A$ and $m$, respectively. Here, we focus on systems whose Hilbert space $\H_A$ carries a $j_A$-irrep, i.e., $\H_A$ is spanned by $d_A=2j_A+1$ vectors $\ket{j_A,m}$ that transform as the $j_A$-irrep (also meaning that there is no subspace of $\H_A$ that is left invariant under the action of $U_A(\vv{g})$). Physically this corresponds to a simple spin-$j_A$ system rather than to the one composed of many spin-$j$ systems. 
	
The set of SU(2)-covariant channels between a system whose Hilbert space carries $j_A$-irrep and a system whose Hilbert space carries $j_B$-irrep has a particularly simple structure. This is because the representations \mbox{$U_A(g)\otimes U_A^{*}(g)$} and \mbox{$U_B(g)\otimes U_B^{*}(g)$} on \mbox{$\H_A\otimes \H_A$} and \mbox{$\H_B\otimes \H_B$} have a multiplicity-free decomposition into irreps. More precisely, the tensor representation \mbox{$U_A(g)\otimes U_A^*(g)$} can be decomposed into $l$-irreps with $l$ varying between $0$ and $2j_A$. In other words,
\begin{equation}
	\label{eq:su2_irrep_decomp_A}
	\H_A\otimes\H_A=\bigoplus_{l=0}^{2j_A} \H^l,
\end{equation}
where $\H^l$ is a $(2l+1)$-dimensional Hilbert space carrying irrep $l$. Analogous statement holds for the output system~$B$. This means that the symmetry-adapted basis of ITOs for the input and output systems have no multiplicities and are given by 
\begin{subequations}
	\begin{align}
		\label{eq:su2_irrep_ITOs_1}
		\!\!\!\{T_m^l\}_{m,l}	&\quad \mathrm{with~}m\in\{-l,\dots,l\},~l\in\{0,\dots,2j_A\},\\
		\label{eq:su2_irrep_ITOs_2}
		\!\!\!\{S_m^l\}_{m,l}	&\quad \mathrm{with~}m\in\{-l,\dots,l\},~l\in\{0,\dots,2j_B\}.		
	\end{align}
\end{subequations}
We note that we can choose $T^0_0$ and $S^0_0$, corresponding to the trivial irrep $l=0$, to be given by $\iden_A/\sqrt{d_A}$ and $\iden_B/\sqrt{d_B}$, and $T^1_m$ to be related to the angular momentum operators in the following way
\begin{equation}
	\label{eq:su2_ITO}
	T_0^1=\frac{J^z_A}{(\|\vv{J}_A\|/\sqrt{3})},\quad T^1_{\pm 1}=\frac{J_A^x\pm i J_A^y}{\sqrt{2}(\|\vv{J}_A\|/\sqrt{3})} ,
\end{equation}
where
\begin{equation}	
	\label{eq:spin_length}
	\|\vv{J}_A\|:=\kk{\tra{3(J_A^z)^2}^{\frac{1}{2}}}=\sqrt{j_A(j_A+1)(2j_A+1)},
\end{equation}
and with analogous expressions for the output system $B$ with $S^1_m$.
	
Moreover, as the multiplicity spaces are 1-dimensional, the operators \mbox{$L^\lambda(\E)$} from Eq.~\eqref{eq:liouville_decomp} of Theorem~\ref{thm:liouville} become scalars $f_\lambda(\E)$. Therefore, the block diagonal decomposition of the Liouville representation of an SU(2)-covariant channel between irreducible systems has a simple block structure given by
\begin{equation}
	\label{eq:liouville_su2_irrep}
	L(\E)=\bigoplus_{l=0}^{2\min (j_A,j_B)} f_{l}(\E)\iden^{l}.
\end{equation}
Employing the symmetry adapted basis of ITOs through Eq.~\eqref{eq:action_on_ITOs}, we can equivalently express the above by
\begin{equation}
	\label{eq:su2_irrep_transform_ITOs}
	\E(T^{l}_{m})=f_{l}(\E)S^{l}_{m}.
\end{equation}
In other words, the covariant channel $\E$ transforms irreducible systems by simply scaling ITOs with irrep-dependent magnitudes encoded in the scaling vector $\vv{f}(\E)$. As a result, the initial state $\rho_A$, given by Eq.~\eqref{eq:ITO_decomp} (without the sum over multiplicities~$\alpha$), is transformed into 
\begin{equation}
	\label{eq:su2_irrep_transform_rho}
	\E(\rho_A)=\frac{\iden_B}{d_B}+\sum_{l=1}^{2\min (j_A,j_B)} f_{l}(\E)\mathbf{r}^{l}\cdot\mathbf{S^{l}},
\end{equation}
where we have used the fact that ITOs $T^0_0$ and $S^0_0$ are given by identities and, due to \cc{the trace preserving condition}, $f_0(\E)=1$. 
	
At this point we know that the action of an SU(2)-covariant channel $\E$ between irreducible systems is fully described by a scaling vector $\vv{f}(\E)$ through Eq.~\eqref{eq:su2_irrep_transform_rho}, but to understand the relation between deviations from conservation laws and unitarity of $\E$, we need to find the constraints on $\vv{f}(\E)$. In particular, we will be interested in possible values of $f_1(\E)$, since this number quantifies how much the angular momentum of the system changes under the action of $\E$. To achieve this, we will look at the Jamio{\l}kowski state $\J(\E)$, enforce its positivity (to ensure CP condition), and \mbox{$\trb{B}{\J(\E)}=\iden_A/d_A$} (to ensure TP condition), thus finding constraints on $\vv{f}(\E)$ which ensure that it corresponds to a valid quantum channel. 
	
Using Theorem~\ref{thm:jamiolkowski}, we find that the Jamio{\l}kowski state is also block-diagonal, and the structure of the blocks is again very simple. This is because the tensor representation \mbox{$U_B(g)\otimes U_A^*(g)$} can be decomposed into $L$-irreps with \mbox{$L\in\{|j_A-j_B|,\dots,j_A+j_B\}$} and no multiplicities. In other words,
\begin{equation}
	\label{eq:su2_irrep_decomp_AB}
	\H_B\otimes\H_A=\bigoplus_{L=|j_A-j_B|}^{j_A+j_B} \H^L,
\end{equation}
where $\H^L$ is a $(2L+1)$-dimensional Hilbert space carrying irrep $L$. As the multiplicity spaces are 1-dimensional, the operators \mbox{$\J^\lambda(\E)$} from Eq.~\eqref{eq:jamiolkowski_decomp} become scalars, and thus we have
\begin{equation}
	\label{eq:jamiolkowski_su2_irrep}
	\!\!\!	\J(\E)=\sum_{L=|j_A-j_B|}^{j_A+j_B} p_L(\E) \J(\E^{L}),\quad\! \J(\E^{L})=\frac{\iden^L}{2L+1}.	\!
\end{equation}
Crucially, each $\J(\E^{L})$ corresponds to a valid Jamio{\l}kowski state: it is clearly positive semi-definite, and the trace-preserving condition can be shown as follows. First, observe that for all $g$ we have
\begin{equation}
	U_A(g) \trb{B}{\J(\E^L)} U_A^\dagger(g)=\trb{B}{\J(\E^L)}.
\end{equation}
Then, since \mbox{$\trb{B}{\J(\E^L)}\in\B(\H_A)$} commutes with an irrep $U_A(g)$ for all $g$, we can use Schur's Lemma~\ref{lem:schur} to conclude that \mbox{$\trb{B}{\J(\E^L)}$} must be proportional to identity. Finally, normalisation of $\J(\E^L)$ ensures that \mbox{$\trb{B}{\J(\E^L)}=\iden_A/d_A$}. Moreover, since the supports of $\J(\E^{L})$ are disjoint, $\E^L$ correspond to extremal channels,
\begin{equation}
	\label{eq:decomposition_cptp}
	\E=\sum_{L=|j_A-j_B|}^{j_A+j_B} p_L(\E) \E^{L}.
\end{equation}
Clearly,
\begin{equation}
	\vv{f}(\E)=\sum_{L=|j_A-j_B|}^{j_A+j_B} p_L(\E) \vv{f}(\E^{L}),
\end{equation}
so that in order to find constraints on $\vv{f}(\E)$, we only need to find the values of $\vv{f}(\E^L)$ for all $L$. More precisely, the set of allowed $\vv{f}(\E)$ is then given by a convex set with extremal points given by $\vv{f}(\E^L)$.
	
We will find $\vv{f}(\E^L)$ by deriving the explicit action of $\E^L$ on the basis elements \mbox{$\{\ketbra{j_A,m}{j_A,n}\}$} with \mbox{$m,n\in\{-j_A,\dots,j_A\}$}. First, note that $\iden^L$ appearing in the expression for $\J(\E^L)$ is given by
\begin{equation}
	\iden^L=\sum_{k=-L}^L \ketbra{L,k}{L,k}.
\end{equation}
Next, using Clebsch-Gordan expansion for the above total angular momentum states $\ket{L,k}$ in terms of the angular momentum states of $\H_A$ and $\H_B$, we write
\begin{equation}
	\hspace{-0.28cm}\ket{L,k}\!=\!\sum_{m=-j_B}^{j_B}\sum_{n=-j_A}^{j_A} \!\!\!\braket{j_B,m;j_A,n}{L,k}   \ket{j_B,m;j_A,n}.
\end{equation}
Now, employing the identity
\begin{equation}
	\E(X)=d_A\trb{A}{\J(\E)(\iden_B\otimes X^{*\dagger})}
\end{equation}
that holds for all $X\in\B(\H_A)$, as well as the following two properties of Clebsch-Gordan coefficients,
\begin{subequations}
	\begin{align}
		&\braket{j_B,m;j_A,n}{L,k}\propto \delta_{m+n,k},\\
		&\braket{j_B,m;j_A,n}{L,k}=\nonumber\\
		&\hspace{-0.3cm}\quad=(-1)^{j_A-L+m}\sqrt{\frac{2L+1}{2j_A+1}}\braket{j_B,-m;L,k}{j_A,n},
	\end{align}
\end{subequations}
we arrive at
\begin{align}
	\label{eq:su2_extremal}
	&\E^{L}(\ketbra{j_A,n}{j_A,m})=\nonumber\\
	&\quad=\sum_{k=-L}^{L}\left(\, \clebsch{j_B}{n-k}{L}{k}{j_A}{n}\langle j_A,m|j_B,m-k;L,k\rangle\nonumber\right.\\
	&\left.\quad\quad\quad\quad\quad\quad ~\ket{j_B,n-k}\!\bra{j_B,m-k} \, \right).
\end{align}
Note that the action of an extremal channel $\E^{L}$ can be physically interpreted as first splitting the original system with total angular momentum $j_A$ into two subsystems with total angular momenta $j_B$ and $L$ (using Clebsch-Gordan coefficients), and then discarding the second subsystem. These extremal channels have been examined in detail in previous literature under the name of EPOSIC channels~\cite{nuwairan2013su2}. 
	
Finally, using Eq.~\eqref{eq:su2_irrep_transform_ITOs} and noting that there exists $m',n'$ and $k$ such that \mbox{$\bra{j_B,n'}S^{l}_{k}\ket{j_B,m'}\neq 0$}, we can write
\begin{equation}
	f_{l}(\E^L)=\frac{\bra{j_B,n'}\E^{L}(T^{l}_{k})\ket{j_B,m'}}{\bra{j_B,n'}S^{l}_{k}\ket{j_B,m'}}.
\end{equation} 
We emphasise that that the quantity above is independent of $m',n'$ and $k$. Now, by expanding $T^{l}_{k}$ in the basis \mbox{$m,n\in\{-j_A,\dots,j_A\}$}, using Eq.~\eqref{eq:su2_extremal}, and employing Wigner-Eckart theorem, we can derive the following expression for $f_{l}(\E^L)$:
\begin{align}
	\label{eq:coefficients}
	\!\!\!f_{l}(\E^L)=\frac{\bra{j_A}|T^{l}|\ket{j_A}}{\bra{j_B}|S^{l}|\ket{j_B}}\sum_{k=-L}^{L}&\frac{\clebsch{j_A}{j_B+k}{l}{0}{j_A}{j_B+k}}{\clebsch{j_B}{j_B}{l}{0}{j_B}{j_B}}\nonumber\\
	&\!\!\!\!\cdot\clebsch{j_B}{j_B}{L}{k}{j_A}{j_B+k}^2,\!
	\end{align}
where $\bra{j_A}|T^{l}|\ket{j_A}$ and $\bra{j_B}|S^{l}|\ket{j_B}$ are reduced matrix elements independent of $n',m'$ or $k$. It simplifies significantly when $j_A=j_B=j$:
\begin{align}
	\label{eq:f_components}
	f_{l}(\E^L)&=\sum_{k=-L}^{L}&\hspace{-0.3cm}\frac{\clebsch{j}{j+k}{l}{0}{j}{j+k}}{\clebsch{j}{j}{l}{0}{j}{j}}\clebsch{j}{j}{L}{k}{j}{j+k}^{2}.
\end{align}
We provide the step-by-step derivation of the above expressions in Appendix~\ref{app:su2}, where we also show how to obtain the explicit formula for $f_{1}(\E^L)$, 
\begin{align}
	\label{eq:explicit_f1}
	f_{1}(\E^L)=&\left(\frac{j_A(j_A+1)+ j_B(j_B+1)-L(L+1) }{2j_B(j_B+1)} \right)\cdot \nonumber\\
	&\cdot\sqrt{\frac{j_B(j_B+1)(2j_A+1)}{j_A(j_A+1)(2j_B+1)}},
\end{align}
which will be crucial for our analysis of spin-inversion and spin-amplification.
	
Let us conclude this section by re-iterating the main result in the form of the following theorem. 
\begin{theorem}
	An SU(2)-covariant channel $\E$ between two irreducible systems, carrying irreps $j_A$ and $j_B$, is fully specified by a probability distribution $\vv{p}(\E)$ of size {\cc{$2\,{\rm{\min}}(j_A,j_B) +1$}}. Its action on $X\in\B(\H_A)$ is then given by
	\begin{align}
		\E(X)&=\sum_{L=|j_A-j_B|}^{j_A+j_B} p_L(\E) \E^L(X)\nonumber\\
		&=\sum_{L=|j_A-j_B|}^{j_A+j_B}\sum_{l=-j_A}^{j_A}\sum_{k=-l}^{l} p_L(\E) f_l(\E^L) x^l_k S^l_k,
	\end{align}
	where $x^l_k=\tra{T^{l\dagger}_k X}$ and $f_l(\E^L)$ are specified by Eq.~\eqref{eq:coefficients}.
\end{theorem}
	
	
\subsection{General decomposition of G-covariant channels}
\label{sec:general_structure}
	
Let $U_A$ and $U_B$ be unitary representations of a compact group $G$ acting on $\h_A$ and $\h_B$, respectively. We are interested in quantum channels \mbox{$\E:\B(\h_A)\rightarrow\B(\h_B)$} that are symmetric under these actions. As explained in Sec.~\ref{sec:introsymmops}, the corresponding Jamio{\l}kowski state $\J(\E)$ will commute with the tensor product representation $U_{B}\otimes U_A^{*}$, which decomposes the Hilbert space $\h_B\otimes \h_A$ into
\begin{equation}
	\label{eq:isomorphism}
	\h_B\otimes \h_A \cong \bigoplus_{\lambda\in \Lambda} \h^{\lambda}\otimes \mathbb{C}^{m_{\lambda}}.	
\end{equation}
Here, $\Lambda$ is a subset of all non-equivalent irreducible representations labelled generically by $\lambda$ that appear with multiplicities $m_{\lambda}$ (denoting the dimension of the multiplicity space). 
	
From Theorem~\ref{thm:jamiolkowski} we know that under such a decomposition the Jamio{\l}kowski state of a symmetric channel has a block-diagonal structure:
\begin{equation}
	\J(\E) = \bigoplus_{\lambda\in \Lambda} \frac{\iden^{\lambda}}{d_{\lambda}}\otimes \J^{\lambda}(\E).
\end{equation}
Note that in the above \mbox{$\J^{\lambda}(\E)$} are bounded operators on a $m_{\lambda}$ dimensional complex space, and $\iden^{\lambda}$ acts as identity on the $\lambda$-irrep representation space $\h^{\lambda}$. Let us now define 
\begin{equation}
	\rho^{\lambda}(\E):= \frac{\J^{\lambda} (\E)}{p_{\lambda}(\E)},\quad p_{\lambda}(\E)=\Tr(\J^{\lambda}(\E)).
\end{equation}
Since $\E$ is completely positive, we have $\J(\E) \geq 0$ and thus $\rho^{\lambda}(\E) \geq 0 $. Moreover, the trace-preserving property of $\E$ implies that $\sum_{\lambda\in \Lambda} p_{\lambda}(\E) = 1$. Therefore, $p_{\lambda}(\E)$ is a probability distribution and $\rho^{\lambda}(\E)$ is a valid quantum state on $GL(\mathbb{C}^{m_{\lambda}})$. One should keep in mind, however, that there will be additional constraints on $\rho^{\lambda}(\E)$ coming from the trace-preserving condition.
	
We can thus write
\begin{equation}
	\J(\E) = \bigoplus_{\lambda\in \Lambda} p_{\lambda}(\E) \  \frac{\iden^{\lambda}}{d_{\lambda}}\otimes \rho^{\lambda}(\E).
	\label{eq:choigeneraldecomp}
\end{equation}
Now, recall that any state $\rho^{\lambda}(\E)\in GL(\mathbb{C}^{m_{\lambda}})$ can be viewed as a probability distribution over all pure state such that:
\begin{equation}
	\rho^{\lambda}(\E) = \int d\psi^{\lambda} \ r_{\E}(\psi^{\lambda}) \ketbra{\psi^{\lambda}}{\psi^{\lambda}} \ 
\end{equation}
where $|\psi^{\lambda}\> \in \mathbb{C}^{m_{\lambda}}$, integration is over all such pure states (according to the Haar measure) with \mbox{$\int d\psi^{\lambda}  \, r_{\E}(\psi^{\lambda}) = 1$} and $r_{\E}(\psi^{\lambda})\geq 0$. We can then define the following operators,
\begin{equation}
	\label{eq:pure_jamiolkowski}
	\J^{\lambda}_{\psi^{\lambda}}: = \frac{\iden^{\lambda}}{d_{\lambda}}\otimes \ketbra{\psi^{\lambda}}{\psi^{\lambda}},
\end{equation}
which should be viewed as elements of $\B(\h_B\otimes\h_A)$ that are positive and have trace one. Therefore, any symmetric Jamio{\l}kowski state $\J(\E)$ can be written as follows
\begin{equation}
	\J(\E) = \sum_{\lambda\in \Lambda}  p_{\lambda}(\E) \int d\psi^{\lambda} \ r_{\E}(\psi^{\lambda}) \J^{\lambda}_{\psi^{\lambda}}.
\end{equation}
This directly leads to the following decomposition of any $G$-covariant channel:
\begin{equation}
	\E = \sum_{\lambda\in \Lambda} p_{\lambda}(\E) \int d\psi^{\lambda} \ r_{\E} (\psi^{\lambda}) \E^{\lambda}_{\psi^{\lambda}}.
\end{equation}
Here, $\E^{\lambda}_{\psi^{\lambda}}$ are CP maps corresponding to Jamio{\l}kowski states \mbox{$\J^{\lambda}_{\psi^{\lambda}}$}. Note, however, that although the above resembles a convex decomposition over extremal channels $\E^{\lambda}_{\psi^{\lambda}}$, these are not necessarily trace-preserving. Therefore, the set of extremal $G$-covariant quantum channels may be much more complicated, e.g., with \mbox{$\ketbra{\psi^{\lambda}}{\psi^{\lambda}}$} in Eq.~\eqref{eq:pure_jamiolkowski} replaced by a mixed state.
	
More can be said about the structure of extremal channels under additional assumptions. The particular case we consider here is given by these symmetries for which representations $U_A$ and $U_B$ of a compact group $G$ (acting on the input and output Hilbert spaces $\h_A$ and $\h_B$) are such that $\h_{B}\otimes \h_A$, with the tensor product representation $U_B\otimes U_A^{*}$, has a multiplicity-free decomposition,
\begin{equation}
	\h_B\otimes\h_A = \bigoplus_{\lambda\in \Lambda} \h^{\lambda}.
\end{equation}
Moreover, we will also require that $U_A$ is an irrep. One example of a group satisfying these assumptions is the SU(2) symmetry with the input system being irreducible, which we studied in detail in Sec.~\ref{sec:su2_structure}. For completeness, we remark that previous works~\cite{stembridge2003multiplicity} have fully characterised under what conditions tensor products of irreducible representations have a multiplicity-free decomposition for all connected semisimple complex Lie groups. In particular, if $G$ is a simple Lie group (e.g., $SL(d)$), then either $U_B$ or $U_A^{*}$ must correspond to an irrep with the highest weight being a multiple of the fundamental representation. For example, for the group $SU(3)$ with the fundamental irrep labelled by $\mathbf{3}$, we have a multiplicity-free decomposition $\mathbf{3}\otimes \bar{\mathbf{3}}  = \mathbf{8}\oplus \mathbf{1}$. This stands to show that the assumptions can still include a large class of symmetries beyond the canonical SU(2) example, e.g., Ref.~\cite{mozrzymas2017structure} studies covariant channels with respect to finite groups with multiplicity-free decomposition. 
	
For groups satisfying these conditions, Eq.~\eqref{eq:choigeneraldecomp} simplifies significantly and takes the following form
\begin{equation}
	\label{eq:directdecompChoi}
	\!\!\!	\J(\E)=\bigoplus_{\lambda\in \Lambda} p_\lambda(\E) \J(\E^{L}),\quad\! \J(\E^{\lambda})=\frac{\iden^\lambda}{d_\lambda}.	\!
\end{equation}
Here, by the same argument as in Sec.~\ref{sec:su2_structure}, $p_\lambda(\E)$ is a probability distribution and each $\J(\E^{\lambda})$ is a positive operator satisfying \mbox{$\trb{B}{\J(\E^\lambda)}=\iden_A/d_A$}. Therefore, each $\J(\E^{\lambda})$ will uniquely correspond to a CPTP map $\E^{\lambda}:\B(\h_A)\rightarrow\B(\h_B)$ and, since $\J(\E^{\lambda})$ act on orthogonal subspaces, they will be linearly independent operators. Equivalently, this ensures that $\E^{\lambda}$ are extremal points of the set of $G$-covariant channels. We can thus characterise $\E^{\lambda}$ in terms of Jamio{\l}kowski states, Kraus operators and Stinespring dilation through the following theorem.
\begin{theorem}
	\label{thm:extremalchannelmultipfree}
	Let $G$ be a compact group with representations $U_A$ and $U_B$ acting on Hilbert spaces $\h_A$ and $\h_B$. Suppose that $U_B\otimes U_A^{*}$ is a multiplicity-free tensor product representation with non-equivalent irreps labelled by elements of a set $\Lambda$, and that $U_A$ is an irrep. Then, the convex set of $G$-covariant quantum channels $\E:\B(\h_A)\rightarrow\B(\h_B)$ has $|\Lambda|$ distinct isolated extremal points given by channels $\E^{\lambda}$ for $\lambda\in \Lambda$. Each $\E^{\lambda}$ can be characterised by the following: 	
	\begin{enumerate}
		\item A unique Jamio{\l}kowski state 
		\begin{equation}
			\J(\E^{\lambda}) = \frac{\iden^{\lambda}}{d_\lambda}.
		\end{equation} 
		\item Kraus decomposition $\{E^{\lambda}_{k}\}_{k=1}^{d_\lambda}$ such that:
		\begin{equation}
			\E^{\lambda}(\rho) = \sum_{k} E^{\lambda}_{k} \rho (E^{\lambda}_k)\hc,
		\end{equation}
		with $E^{\lambda}$ forming a $\lambda$-irreducible tensor operator transforming as \mbox{$U_{B}(g) E^{\lambda} U_{A}(g)\hc = \sum_{k'} v^{\lambda}_{k'k}(g) E^{\lambda}_{k}$}, where $v^{\lambda}_{k'k}$ are matrix coefficients of the $\lambda$-irrep.\\
		\item A symmetric isometry $W^{\lambda}:\h_A \rightarrow\h_B\otimes\h^{\lambda}$ such that
		\begin{equation}
			\E^{\lambda}(\rho) =\Tr_{\h^{\lambda}}( W^{\lambda}\rho (W^{\lambda})\hc).
		\end{equation}
		Also, the minimal Stinespring dilation dimension for $\E^\lambda$ is given by $d_\lambda$.
	\end{enumerate}
\end{theorem}
\noindent The details on how to obtain characterisations 2. and~3. from 1. can be found in Appendix~\ref{app:convexstruct}. 
	
	
\subsection{Decomposition of U(1)-covariant channels}
\label{sec:u1_structure}
	
We now proceed to the simplest example of a compact group that does not satisfy the multiplicity-free condition~--~the U(1) group. As we will see in a moment, channels symmetric with respect to U(1) group do not satisfy this condition in the strongest possible way: they act trivially on the irrep spaces (since those are one-dimensional) and are fully defined by their action within the multiplicity spaces. In that sense, the example investigated in this section is the exact opposite of SU(2)-irreducibly-covariant channels studied in Sec.~\ref{sec:su2_structure}, where the action within multiplicity spaces was trivial and channels were defined by their action within irrep spaces.
	
The U(1) group has a single generator $J_A^1$,
\begin{equation}
	\label{eq:u1_action}
	U_A(g)=e^{iJ_A^1 g},
\end{equation}
where $g\in[0,2\pi]$. For a finite-dimensional system\footnote{More precisely: for a system exhibiting cyclic dynamics, i.e., such that there exists time $t_0$ for which $e^{iH_At_0}=\iden_A$.} described by a Hilbert space $\H_A$, the U(1) group can be related to time-translations by choosing the generator to be given by the system Hamiltonian $H_A$,
\begin{equation}
	\label{eq:hamiltonian}
	H_A=\sum_{n} E_A^n \ketbra{E_A^{n}}{E_A^{n}},
\end{equation}
with $E_A^n$ denoting different energy levels, and where we restricted ourselves to non-degenerate Hamiltonians for the clarity of discussion. Indeed, substituting \mbox{$J_A^1\rightarrow H_A$} and \mbox{$g\rightarrow -t$}, we see that the group action,
\begin{equation}
	\label{eq:time_trans}
	U_A(t)=e^{-iH_At},
\end{equation}
evolves the system in time by $t$. The representation of the group on $\H_B$ is defined in an analogous way with the Hamiltonian $H_B$. Recall that, by Noether's theorem, closed unitary dynamics symmetric under time-translations, generated by $H_A$, conserves energy represented by Hamiltonian $H_A$.
	
As U(1) is an Abelian group, its irreducible representations are 1-dimensional, meaning that the symmetry adapted basis composed of ITOs satisfies
\begin{equation}
	\label{eq:u1_action_ITOs}
	\U_A^t(T^{\lambda,\alpha})=v^\lambda(t) T^{\lambda,\alpha}= e^{-i \lambda t} T^{\lambda,\alpha}.
\end{equation}
It follows that we can choose
\begin{equation}
	\label{eq:u1_ITOs}
	T^{\lambda,\alpha}=\ket{E_A^{n}}\!\bra{E_A^{n'}},\quad	S^{\lambda,\alpha}=\ket{E_B^{n}}\!\bra{E_B^{n'}}
\end{equation}
with $\lambda=E_{A/B}^n-E_{A/B}^{n'}$ and $\alpha$ enumerating multiplicities arising from the degeneracy of the Bohr spectrum of $H_A$ and $H_B$, i.e., various pairs $n,n'$ satisfying the same \mbox{$\lambda=E_{A/B}^n-E_{A/B}^{n'}$}. 
	
We consider a U(1)-covariant channel \mbox{$\E$}, with the representations of the U(1) group on the input and output spaces, $\H_A$ and $\H_B$, being given by $U_A(t)$ and $U_B(t)$, i.e., with the Hamiltonians of the input and output systems being $H_A$ and $H_B$. Employing Theorem~\ref{thm:liouville} we then get that the Liouville representation of $\E$ is block-diagonal,
\begin{equation}
	\label{eq:u1_liouville}
	L(\E)=\bigoplus_{\lambda} L^{\lambda}(\E),
\end{equation}
and from Eq.~\eqref{eq:action_on_ITOs} we find that
\begin{equation}
	\label{eq:u1_liouville_multi}
	L^{\lambda}_{\beta\alpha}(\E)= \bra{E_B^{m}} \E(\ket{E_A^{n}}\!\bra{E_A^{n'}}) \ket{E_B^{m'}},
\end{equation}
with $\lambda=E_A^n-E_A^{n'}=E_B^m-E_B^{m'}$ and $\alpha,\beta$ enumerating degeneracies, i.e., various pairs of $n,n'$ and $m,m'$ with the same energy difference $\lambda$.
	
We see that the block $\lambda=0$ describes the evolution of populations (in the energy eigenbases), while the remaining blocks describe the evolution of coherence terms between energy levels differing by $\lambda$. Therefore, $L^{\lambda=0}(\E)$ contains full information needed to study deviations from energy conservation induced by $\E$, while $L^{\lambda\neq 0}(\E)$ define how coherent $\E$ is, i.e., how close it is to a closed unitary dynamics. We note that the relation between $L^{\lambda=0}(\E)$ and $L^{\lambda\neq 0}(\E)$ has played a crucial role in the previous studies on optimal processing of coherence under thermodynamic~\cite{lostaglio2015quantum} and Markovian~\cite{lostaglio2017markovian} constraints. Here, we will use this relation to constrain the unitarity of a general U(1)-covariant channel inducing energy flows (deviating from energy conservation) described by a given stochastic matrix. Since $L^{\lambda=0}(\E)$ is crucial for our studies, we will use a shorthand notation $P^\E$ for it,
\begin{equation}
	\label{eq:population_transfer}
	P^{\E}_{mn}:=\bra{E_B^{m}} \E(\ketbra{E_A^{n}}{E_A^{n}}) \ket{E_B^{m}},
\end{equation}
and note that it is a $d_B\times d_A$ stochastic matrix, \mbox{$P^{\E}_{mn}\geq 0$} and \mbox{$\sum_m P^{\E}_{mn}=1$}. 
	
As our aim is to study the relation between deviations from conservation laws and unitarity of U(1)-covariant channels, we need to understand what are the constraints on $L^{\lambda}(\E)$. To answer this question we will look at the Jamio{\l}kowski state $\J(\E)$ and, by enforcing its positivity and \mbox{$\trb{B}{\J(\E)}=\iden_A/d_A$}, we will find constraints on matrices $L^{\lambda}$ ensuring that they correspond to a valid CPTP map. From Theorem~\ref{thm:jamiolkowski} we get that the Jamio{\l}kowski state is also block-diagonal,
\begin{equation}
	\label{eq:u1_jamiolkowski}
	\J(\E)=\bigoplus_{\lambda} \J^{\lambda}(\E).
\end{equation}
Moreover, the support of each $\J^{\lambda}(\E)$ is spanned by vectors \mbox{$\ket{E_B^n,E_A^{n'}}$} that transform as irrep $\lambda$ under \mbox{$U_{B}(t)\otimes U_{A}^{*}(t)$}, i.e., they satisfy \mbox{$E_B^n-E_A^{n'}=\lambda$}. More precisely, we have
\begin{align}
	\label{eq:u1_jamiolkowski_multi}
	&\J^{\lambda}(\E)=\frac{1}{d_A}\sum_{m,n} \bra{E_B^m}\E(\ketbra{E_B^m-\lambda}{E_B^n-\lambda})\ket{E_B^n}\nonumber\\
	&\quad\quad\quad\quad\quad\quad\quad~ \ketbra{E_B^m,E_B^m-\lambda}{E_B^n,E_B^n-\lambda},
\end{align}
where $\ket{E_B^m-\lambda}$ is a shorthand notation for $\ket{E_A^{m'}}$ with $m'$ such that \mbox{$E_A^{m'}=E_B^m-\lambda$}, and the summation is performed only over the indices $m,n$ for which $E_B^m-\lambda$ and $E_B^n-\lambda$ correspond to valid energies of $H_A$. 
	
The positivity of $\J(\E)$ is now equivalent to the positivity of $\J^{\lambda}(\E)$ for all $\lambda$, while the partial trace condition is fulfilled automatically as long as $P^\E$ is a stochastic matrix. Importantly, the diagonal of $\J^{\lambda}(\E)$ is given by $P^\E_{m,m-\lambda}$ ($m-\lambda$ is such $m'$ that satisfies \mbox{$E^{m'}_A=E^m_B-\lambda$}), while the off-diagonal terms describe transformation of coherences. One can now construct extremal U(1)-covariant channels by simply coherifying any stochastic matrix $\Gamma$ to a quantum channel with the constraint of preserving the block-diagonal structure~\cite{korzekwa2018coherifying}. More precisely, for every stochastic matrix $\Gamma$ and a set of phases $\{\phi_{\lambda,m}\}$, one can construct an extremal U(1)-covariant channel $\E^{\Gamma,\vv{\phi}}$ with the Jamio{\l}kowski state given by
\begin{equation}
	\label{eq:extremal_u1}
	\J(\E^{\Gamma,\vv{\phi}})=\sum_\lambda \ketbra{\psi^{\Gamma,\vv{\phi}}_\lambda}{\psi_\lambda^{\Gamma,\vv{\phi}}},
\end{equation}	
with
\begin{equation}
	\ket{\psi_\lambda^{\Gamma,\vv{\phi}}}:=\frac{1}{\sqrt{d_A}}\sum_{m} e^{i\phi_{\lambda,m}}\sqrt{\Gamma_{m,m-\lambda}}\ket{E_B^m,E_B^m-\lambda},
\end{equation}
where $\Gamma$ describes $P^{\E^{\Gamma,\vv{\phi}}}$ and the same notation applies to its elements $\Gamma_{mn}$.	It is a straightforward calculation to show that the corresponding map is CPTP. Moreover, since its Jamio{\l}kowski state is proportional to a projector on each block, it is extremal. 
	
We want to note, however, that the above construction in general does not produce all extremal U(1)-covariant channels. As a counterexample, consider the following Jamio{\l}kowski state
\begin{equation}
	\J(\E^{\Gamma',\vv{\phi}'}_*)=\sum_{\lambda\neq \lambda_*} 	\ketbra{\psi^{\Gamma',\vv{\phi}'}_\lambda}{\psi_\lambda^{\Gamma',\vv{\phi}'}}+\rho_{\lambda_*}.
\end{equation}	
Since the above is extremal on each block $\lambda\neq \lambda_*$, the possibility of decomposing it as a convex combination of Jamio{\l}kowski states from Eq.~\eqref{eq:extremal_u1} is equivalent to the possibility of decomposing $\rho_{\lambda_*}$ as
\begin{equation}
	\rho_{\lambda_*}=\int d\vv{\phi} \ \mu(\vv{\phi}) \ketbra{\psi^{\Gamma',\vv{\phi}}_\lambda}{\psi_\lambda^{\Gamma',\vv{\phi}}}.
\end{equation}
In other words, it would need to hold that every density matrix of size $d$ can be decomposed into a convex combination of pure states with the same diagonal. This, however, is not true in general (it holds for $d=2$, but counterexamples can be found already for $d=4$).
	
		
\section{Spin-inversion and amplification}
\label{sec:reversal}

\subsection{Setting}

The scenario investigated in this section is as follows. We consider input and output systems, described by Hilbert spaces $\H_A$ and $\H_B$, to be spin-$j_A$ and spin-$j_B$ systems. We denote the spin angular momenta operators (with respect to a Cartesian coordinate frame) by 
\begin{equation}
	\mathbf{J}_{A}=(J_{A}^{x},J_{A}^{y}, J_{A}^{z}),
\end{equation}
and analogously for $\mathbf{J}_{B}$. These are traceless and for every \mbox{$k\in\{x,y,z\}$} satisfy
\begin{equation}
	\tr( (J_{A}^{k})^2) =\frac{\|\vv{J_A}\|^2}{3},
\end{equation}
where $\|\vv{J}_A\|$ was defined in Eq.~\eqref{eq:spin_length}. Analogous conditions hold for system $B$. We recall that these spin operators are generators for the SU(2) irreducible representations on $\h_A$ and $\h_B$, and they span the adjoint irrep (i.e the three dimensional $1$-irrep) in the decomposition of the operator spaces $\B(\h_A)$ and $\B(\h_B)$. In other words, $J^k_A$ and $J^k_B$ are (unnormalised) ITOs $T_k^1$ and $S_k^1$, see Eq.~\eqref{eq:su2_ITO}. Now, for the input and output state, $\rho\in\B(\H_A)$ and $\E(\rho)\in\B(\h_B)$, we can define spin polarisation vectors, $\mathbf{P}({\rho})$ and $\mathbf{P}({\E(\rho)})$, to be given by expectation values of the spin operator along different Cartesian axes:
\begin{equation}
	\label{eq:spin_polarization}
	P_{k}(\rho) := \Tr( J_{A}^k\rho),
\end{equation}
and similarly for the system $B$ with $\rho$ replaced by $\E(\rho)$.

Our aim is to investigate operations that isotropically invert or amplify the spin operator, so that under their action the polarisation vector scales with either some negative factor $\kappa_-$, or a positive factor $\kappa_{+}>1$. In particular, we want to determine channels $\S_{-}$ and $\S_{+}$, representing the optimal spin-inversion and spin-amplification, which are those that achieve the largest values of $|\kappa_{-}|$ and $\kappa_{+}$:
\begin{align}
	\label{eq:spin_rev_amp}
	\vv{P}(\S_\pm(\rho))&= \kappa_\pm \vv{P}(\rho), \quad \kappa_-<0, \quad \kappa_{+}> 1.
\end{align}
Equivalently, $\S_{\pm}$ may be defined in terms of their action on the generators:	
\begin{align}
	\label{eq:spin_rev_amp_2}
	\S_\pm(J^k_A)&= \kappa_\pm J^k_B \cdot \frac{\|\vv{J}_A\|^2}{\|\vv{J}_B\|^2}.	
\end{align}

First, we will take the above equations as really defining $\S_{\pm}$, without specifying their action outside of the subspace spanned by the generators $J_{A}^{k}$. This will, in principle, correspond to a large class of operations that we need to optimise over. However, since $\S_{\pm}$ acts isotropically on all states, in the next section we will show that without loss of generality one may restrict considerations to SU(2)-covariant channels. This will allow us to employ results of Sec.~\ref{sec:su2_structure} to determine optimal inversion and amplification factors $\kappa_\pm$, and to relate $\kappa_-$ to the maximal allowed deviation from conservation law under covariant dynamics. Finally, we will focus on the decoherence induced by the optimal inversion channel by comparing the action of this channel with the action induced by time-reversal symmetry.


\subsection{Optimal transformations of spin polarisation}
\label{sec:reversal_B}

We want to analyse channels $\E:\B(\h_A)\rightarrow \B(\h_B)$ that send $\vv{P}(\rho)$ to $\kappa \vv{P}(\rho)$ for all $\rho$ and some independent real constant $\kappa$, while performing arbitrary transformation on the other irreducible subspaces (ITOs). As we will now show, for every such $\E$ there exists an SU(2)-covariant channel that has the exact same action on the polarisation vector. 
By assumption,
\begin{equation}
	\forall \rho: \quad\Tr(J_{B}^k \E(\rho)) = \kappa \Tr(J_{A}^k \rho).
\end{equation}
Now, with $\U_{A}$ denoting the SU(2) representations on $\B(\H_A)$, recall that the angular momentum operators transform under rotations as
\begin{equation}
	\U_A^{g}(J_A^{k}) = \sum_{k'} v^{1}_{k'k}(g) J_{A}^{k'},
\end{equation}
where $v^{1}_{k'k}(g)$ are matrix entries of the $1$-irrep. Analogous statement holds for system $B$. Therefore, it follows that
\begin{align}
	\Tr(\U^g_{B} (J_B^k) \E(\rho)) &= \sum_{k'} v^{1}_{k'k}(g)\Tr(J_B^{k'}\E(\rho))\non 
	& = \kappa\sum_{k'} v^{1}_{k'k}(g)\Tr(J_A^{k'}\rho)\non 
	& =\kappa \tr(\U_{A}^g(J_A^k) \rho).
\end{align}
Using the cyclic property of the trace and the fact that the above must hold for all $\rho$, so in particular for $\U_{A}^g(\rho)$, we arrive at
\begin{equation}
	\Tr(J_{B}^k\ \U^{g\dagger}_{B} \circ \E\circ \U^g_{A}(\rho)) = \kappa \Tr(J_A^k \rho),
\end{equation}
or equivalently:
\begin{equation}
	\vv{P}(\U^{g\dagger}_{B} \circ \E\circ \U^g_{A}(\rho))=\kappa\vv{P}(\rho).
\end{equation}
We note that the above could also be simply deduced from the fact that $\vv{P}$ transforms under SU(2) as a three-dimensional vector in real space, so that for all $g\in \mathrm{SU(2)}$ and all $\rho$ we have
\begin{equation}
	\vv{P}(\rho)=\vv{P}(\U_g(\rho)) \quad\Longrightarrow\quad \vv{P}(\U_g^\dagger(\rho))=\vv{P}(\rho).
\end{equation}

Next, by taking the group average and noting that $\vv{P}$ is linear (since it is defined through trace in Eq.~\eqref{eq:spin_polarization}), we obtain
\begin{align}
	\vv{P}(\G[\E](\rho))&= \kappa  \vv{P}(\rho),
\end{align}
where $\G[\E]$ is the \emph{twirling} of $\E$ over all rotations,
\begin{equation}
	\G[\E]:= \int dg \ \U^{g\dagger}_{B}\circ\E\circ \U^g_{A}.
\end{equation}
The twirled channel $\G[\E]$ is SU(2)-covariant (by construction), and it has the same scaling factor $\kappa$ as $\E$. Therefore, one may assume without loss of generality that the optimal spin-inversion and amplification operations are symmetric under SU(2).

In Sec.~\ref{sec:su2_structure} we have fully characterised SU(2)-covariant quantum channels for irreducible systems and we will now employ these results. First, recall that a symmetric channel \mbox{$\E:\B(\h_A)\rightarrow\B(\h_B)$} acts on any ITO \mbox{$\{T^{l}_{k}\}_{l,k}$} by a scaling factor depending only on the particular irrep (and the channel itself) such that 
\begin{equation}
	\E(T^{l}_{k} ) = f_{l}(\E) S^{l}_k.
\end{equation}
Taking into account the particular normalisation of the spin operators it follows that:
\begin{equation}
	\E(J_A^{k}) =  f_{1}(\E) J_{B}^k\cdot \frac{\|\vv{J_A}\|}{\|\vv{J_B}\|}.
	\label{eq:spintransf}
\end{equation}
Moreover, recall that every such SU(2)-covariant $\E$ decomposes into extremal channels according to Eq.~\eqref{eq:decomposition_cptp}, which results in
\begin{equation}
	f_1(\E) = \sum_{L=|j_A-j_B|}^{j_A+j_B} p_L f_{1}(\E^{L}),
	\label{eq:f1}
\end{equation}
where $f_{1}(\E^L)$ are the scaling factors explicitly given by Eq.~\eqref{eq:explicit_f1}. 

Now, we can compare the transformation of angular momentum operators under a general covariant channel, Eqs.~\eqref{eq:spintransf}-\eqref{eq:f1}, with the transformation under spin-inversion and spin-amplification, Eq.~\eqref{eq:spin_rev_amp_2}. We see that every SU(2)-covariant channel can act as a spin-inversion or amplification with
\begin{align}
	\kappa_\pm =\frac{\|\vv{J_B}\|}{\|\vv{J_A}\|}\cdot \sum_{L=|j_A-j_B|}^{j_A+j_B} p_L f_{1}(\E^{L}),
\end{align}
as long as $\kappa_+>1$ or $\kappa_-<0$. Our aim is thus to maximise and minimise the above expression over all probability distributions $p_{L}$. Since we are optimising over a convex region, the optima will be attained by one of the extremal points so that
\begin{equation}
	\kappa_\pm^{\mathrm{opt}} = \frac{\|\vv{J_B}\|}{\|\vv{J_A}\|}\cdot f_{1}(\E^{L_\pm})
\end{equation}
for some $L_{\pm}$. These can be easily found, as we derived explicit expressions $f_{1}(\E^L)$. 

We thus arrive at:
\begin{theorem}
	The maximal spin polarisation inversion, $\vv{P}(\rho)\rightarrow \kappa_- \vv{P}(\rho)$ with $\kappa_-<0$, is achieved by an SU(2)-irreducibly extremal channel $\E^{(j_A+j_B)}$. The inversion factor $\kappa_-$ is given by
	\begin{equation}
		\label{eq:spinreversalmax}
		\kappa_- =  - \frac{j_B(2j_B+1)}{(j_A+1)(2j_A+1)}.
		\end{equation}
	\label{thm:spinreversalmax}
\end{theorem}
It follows that the maximal spin-inversion is achieved by the extremal channel that requires the largest environment to be realised. Indeed, for every extremal channel~$\E^{L}$, its minimal Stinespring dilation (and thus the minimal number of Kraus operators) has dimension $2L+1$. Consequently, this means that the larger the environment, the more we can invert the spin. Note that in the classical macroscopic limit of an input and output systems given by massive spins, $j_A=j_B\rightarrow \infty$, we get $\kappa_-\rightarrow -1$, corresponding to perfect spin-inversion. While for finite-dimensional systems quantum theory does not allow for perfect spin-inversion, $\mathbf{P}\rightarrow - \mathbf{P}$, the above result yields fundamental limit on maximal spin-inversion.

Moreover, the optimal spin-inversion coincides with the channel leading to the largest allowed deviation from the conservation law under the constraint of SU(2) symmetry. To see this note that the total deviation resulting from the action of $\E$ on a given input state $\rho$ (defined in Eq.~\eqref{eq:total_deviation}), can be expressed by
\begin{equation}
	\Delta_{\mathrm{tot}}(\rho,\E)  =  \| \mathbf{P}(\E(\rho)) -\mathbf{P}(\rho) \|^{2}.
\end{equation}
Using the fact that covariant dynamics can only scale ITOs, we get
\begin{equation}
	\Delta_{\mathrm{tot}}(\rho,\E) = |\kappa-1|^2\ \|\mathbf{P}(\rho) \|^2,
\end{equation} 
and thus the deviation is maximised for smallest negative~$\kappa$, which is specified by Eq.~\eqref{eq:spinreversalmax}. From the equation above it is clear that also average total deviation, $\Delta(\E)$, will be maximised by the optimal spin-inversion channel. Of course, since we deal with symmetric channels, this deviation can come only for the price of decoherence (as the conserved charge can only come from an incoherent environment). In the next section, we will quantify this decoherence by comparing the action of the optimal spin-inversion channel with the transformation induced by time-reversal symmetry; while in Sec.~\ref{sec:trade_off} we will analyse in detail the trade-off between deviations from conservation laws and decoherence for general SU(2)-covariant operations.

Finally, we can obtain an analogous bounding result for spin amplification, captured by the following theorem.
\begin{theorem}
	\label{thm:spinamplify}
	The maximal spin polarisation amplification, $\vv{P}(\rho)\rightarrow \kappa_+ \vv{P}(\rho)$ with $\kappa_+>1$, is achieved by an SU(2)-irreducibly symmetric extremal channel $\E^{(|j_A-j_B|)}$. The amplification factor $\kappa_+$ is given by:
	\begin{equation}
		\kappa_+=\left\{\begin{array}{cc}
			\frac{j_B}{j_A} &\quad \mathrm{for~}j_A\geq j_B,\\~\\
			\frac{j_B+1}{j_A+1}&\quad \mathrm{for~}j_A< j_B.
		\end{array}\right.
	\end{equation}	
\end{theorem}
We remark that upper bounds on $\kappa_+$ have been previously reported in Ref.~\cite{marvian2014modes}, where the authors used resource monotones based on modes of asymmetry to show that \mbox{$\kappa_+\leq f(j_B)/f(j_A)$} with
\begin{align}
	f(j):=\left\{\begin{array}{lll}
		j+1/2~&:&~\mathrm{integer~}j,\\
		j(j+1)/(j+1/2)~&:&~\mathrm{half~integer~}j.
	\end{array}\right.
\end{align}
Note that, according to Theorem~\ref{thm:spinamplify} that provides the optimal amplification channel explicitly, these bounds are loose, i.e., the upper bound cannot be achieved by any SU(2)-covariant channel. In this sense our result can be seen as an ultimate improvement over the previously known bounds.


\subsection{Optimal spin-inversion and time-reversal symmetry}
\label{sec:inversion_reversal}

So far we have considered the action of a channel on spin polarisation vector as the defining property of the spin-inversion channel. We have thus focused on the maximal deviation from the conservation law, but ignored the decoherence induced by such a channel, which is described by the action of the channel on the remaining ITOs. Here, we will quantify this decoherence by comparing the action of the optimal spin-inversion channel to the action of a passive symmetry that naturally realises spin-inversion -- the time-reversal symmetry~$\T$.

Under the action of \mbox{$\mathcal{T}:\B(\H_A)\rightarrow\B(\H_A)$} the spin of a single particle flips sign and, generally, an odd number of particles will experience a sign change, while an even number will not. This manifests itself at the level of ITOs, which are mapped according to whether they correspond to even or odd dimensional irreducible representations of the rotation group:
\begin{equation}
	\mathcal{T}(T^{\lambda}_k) = (-1)^{\lambda}  T^{\lambda}_{-k}.
\end{equation} 
This fully captures the action of time-reversal on general mixed states of spin-$j_A$ systems described by the Hilbert space $\h_A$. In particular, the spin degrees of freedom under time-reversal will acquire a minus sign:
\begin{equation}
	\mathcal{T} (J_A^{k}) = - J_A^{k}.
\end{equation}
Therefore, for a single particle, time-reversal symmetry induces perfect spin reversal, as for any $\rho\in \B(\h_A)$ the spin polarisation vector satisfies \mbox{$\vv{P}(\T(\rho)) = -\vv{P}(\rho)$}. Moreover, $\T$ does not induce any decoherence, since it leaves the eigenvalues of $\rho$ unchanged. It is thus meaningful to compare the optimal physical spin-inversion channel \mbox{$\E^{(2j_A)}$} from Theorem~\ref{thm:spinreversalmax} with the perfect unphysical spin-inversion operation realised by time-reversal symmetry $\T$. We will see that \mbox{$\E^{(2j_A)}$}, although it inverts spin polarisation almost perfectly in the limit of large $j_A$, is always far away from realising $\T$, and thus induces unavoidable decoherence as expected.

In order to measure the distance between \mbox{$\E^{(2j_A)}$} and $\T$ let us introduce the concept of  a spin-coherent state. It is simply given by a rotation of $|j_A,j_A\>$, the state with maximal angular momentum along the $z$ axis. Suppose that the group element $g\in \mathrm{SU(2)}$ is characterised by the Euler angles $\theta,\phi$, corresponding to a spatial direction $\hat{n}$. Then the spin-coherent state associated to this direction is given by:
\begin{equation}
	|\hat{n}_A\> = U_A(g) |j_A,j_A\> = \sum_{k=-j_A}^{j_A} u_{kj_A}(g) |j_A,k\>.
\end{equation}
The behaviour of spin coherent states under time-reversal symmetry is particularly simple and reads
\begin{equation}
	\T(\ketbra{\hat{n}_A}{\hat{n}_A}) = \ketbra{-\hat{n}_A}{-\hat{n}_A}.
\end{equation}

In order to quantify how much the optimal spin-inversion channel \mbox{$\E^{(2j_A)}$} resembles the passive symmetry transformation $\mathcal{T}$ we will employ the notion of quantum fidelity,
\begin{equation}
	F(\rho,\sigma) := \Tr\left(\sqrt{\sqrt{\rho}\sigma\sqrt{\rho}})\right)^2.
\end{equation}
Namely, we will calculate the fidelity $\bar{F}$ between the outputs of the two channels averaged over all input spin-coherent states. Notice that the fidelity between two states is a unitarily invariant measure, so that
\begin{equation}
	\!\!\!	F(\E^{(2j_A)}(\rho),\T(\rho)) = F(\U_A^g(\E^{(2j_A)}(\rho)), \U_A^g(\T(\rho)) ),\!
\end{equation}
and, since both $\E^{(2j_A)}$ and $\T$ are SU(2)-covariant, it follows that the considered fidelity remains the same for all spin coherent input states. Therefore, it suffices to analyse the fidelity for the input state $\ket{j_A,j_A}$, i.e.,
\begin{align}
	\bar{F} = \<j_A,-j_A|\E^{(2j_A)}(\ketbra{j_A,j_A}{j_A,j_A})|j_A,-j_A\>.
\end{align}
Now, we can use the explicit form of \mbox{$\E^{(2j_A)}$} given in Eq.~\eqref{eq:su2_extremal} to arrive at
\begin{equation}
	\bar{F}=  |\<j_A,-j_A;2j_A ,2j_A| j_A,j_A\>|^2.
\end{equation}
Finally, employing Clebsch-Gordan coefficients identities we obtain
\begin{equation}
	\bar{F}=\frac{1+2j_A}{1+4j_A} = 1 - \frac{2j_A}{1+4j_A}.
\end{equation}

The above fidelity is monotonically decreasing as a function of $j_A$ and in the limit $j_A\rightarrow \infty$ it converges to $1/2$. Therefore, despite the fact that for macroscopic spins it is possible to almost perfectly invert their polarisation vector, the channel that achieves this is far from realising time-reversal symmetry. We remark that the above calculation only assumes that the action of $\T$ on the spin-coherent state $|j_A,j_A\>$ gives $|j_A,-j_A\>$ and that it is rotationally invariant. Therefore, the same result will hold for a general perfect and unphysical spin-inversion operation which satisfies these two constraints (without committing to the full exact form that the time-reversal operator takes). Moreover, note that the rotational invariance and linearity\footnote{\cc{Note that we can invoke linearity here as we restricted to the convex hull of spin-coherent states (i.e a classical state space), and on this subspace the anti-linear property of the map $\T$ is not relevant since the scalar field is real.}} ensures that the expression for $\bar{F}$ remains unchanged for any state in the convex hull of spin-coherent states.
	

\section{Trade-off relations between conservation laws and decoherence}
\label{sec:trade_off}
	
Building up on the results developed so far, we now address the core questions of interest: how much can open symmetric dynamics deviate from conservation laws? Do small perturbations from closed symmetric dynamics result in small corrections to the conservation laws? When does the converse also hold?
		
Our aim is therefore to analyse when each of the following two qualitative statements holds given an a priori symmetry principle:
\begin{itemize}
	\item If $\E$ is close to a symmetric unitary then the average total deviation from conservation law, $\Delta(\E)$, is small.
	\item If the average total deviation $\Delta(\E)$ is small then $\E$ is close to a symmetric unitary.
\end{itemize}
	
Whenever both of the above properties hold for any dynamics with the appropriate symmetry, we say that the conservation laws are robust with respect to decoherence. Quantitatively, we can analyse such robustness by deriving bounds on the average deviation induced by a channel in terms of its distance from a symmetric unitary process. In what follows, we first derive general upper bounds on the deviation in terms of the diamond distance (for arbitrary dimension of input and output spaces, $d_A$ and $d_B$) and unitarity (for $d_B\leq d_A$), showing that the first property holds in general. Then, we argue why a lower bound does not need to exist for a general group $G$, and so the second property does not need to hold. Nevertheless, we show that for symmetries with multiplicity-free decomposition, the lower bound can also be derived for $d_A=d_B$, and thus conservation laws are robust under decoherence in such cases. Finally, we analyse in detail the two special examples investigated in Sec.~\ref{sec:structure}: SU(2)-irreducibly-covariant channels and U(1)-covariant channels.


\subsection{Upper bounds on deviating charges for $G$-covariant open dynamics}
\label{sec:upperbounds}
	
Before we present our main result upper bounding the average total deviation $\Delta(\E)$ as a function of departure from unitarity $(1-u(\E))$, we want to present a simple argument showing that open dynamics that is close to symmetric unitary (isometry) must approximately conserve relevant charges. Consider $\rho\in\B(\H_A)$ and a $G$-covariant channel \mbox{$\E:\B(\h_A)\rightarrow\B(\h_B)$} with the symmetry generated by $\{J_{A}^{k}\}_{k=1}^n$ for the input system and $\{J_{B}^{k}\}_{k=1}^n$ for the output system. Now, take any isometry $W:\h_A\rightarrow \h_B$ that is symmetric, i.e., $WJ_{A}^k = J_{B}^k W$. Since the conservation laws hold under a dynamics generated by $W$, we have
\begin{align}
	\Delta_k (\rho_A, \E) &= \Tr(\E(\rho_A)J_{B}^k - \rho_A J_{A}^k)\\ \nonumber
	& = \Tr((\E(\rho_A) - W\rho_A W\hc) J_{B}^k). 
\end{align}
Using H\"{o}lder's inequality for the Hilbert-Schmidt inner product,
\begin{equation}
	\Tr(A\hc B) \leq \|B\|_{1} \|A\|_{\infty},
\end{equation}
with \mbox{$\|B\|_1:= \Tr(\sqrt{B\hc B})$} and \mbox{$\|A\|_{\infty}$} the operator norm, we obtain the following bound:
\begin{equation}
	\Delta_k(\rho_A,\E)\leq \|(\E- \mathcal{W})(\rho_A)\|_1\, \|J_{B}^{k}\|_{\infty}, 
\end{equation}
where $\mathcal{W} (\cdot) = W (\cdot) W\hc$. Thus, the total deviation for a given input state $\rho_A$ is bounded by
\begin{equation}
	\Delta_{\mathrm{tot}}(\rho_A,\E)\leq \|(\E- \mathcal{W})(\rho_A)\|^2_1\, \sum_{k=1}^n\|J_{B}^{k}\|^2_{\infty}. 
	\label{eq:1normbound}
\end{equation}
Finally, we can get a state-independent bound by employing a diamond norm,
\begin{equation}
	\|\mathcal{C}\|_{\diamond}^{2} :=\max_{\rho_{AA'}} \|\mathcal{C}_{A}\otimes \mathcal{I}_{A'}(\rho_{AA'})\|_{1},
\end{equation}
so that we arrive at the bound for the average total deviation
\begin{equation}
	\label{eq:diamondupper}
	\Delta(\E) \leq \|\E - \mathcal{W}\|_{\diamond}^2 \sum_{k} \|J_{B}^{k}\|_{\infty}^{2}.	
\end{equation}
Operationally the above can be interpreted as follows: the more indistinguishable a given covariant channel becomes from any symmetric isometry, the smaller the deviations from conservation laws.
		
Obviously, the above simple analysis has significant drawbacks. Not only is the diamond norm particularly difficult to calculate, but also Eq.~\eqref{eq:diamondupper} involves either an unknown symmetric isometry $\mathcal{W}$, or a minimisation of the quantity \mbox{$\|\E-\mathcal{W}\|_{\diamond}$} over all such isometries $\mathcal{W}$. The latter will generally be difficult to estimate from the properties of the channel $\E$ alone, leading to very loose upper bounds on the average total deviation. For these reasons, in the following theorem we provide an explicit inequality that captures robustness of conservation laws in terms of the unitarity of a symmetric channel.
\begin{theorem}
	\label{thm:upperbounds}	
	Let $G$ be a connected compact Lie group with unitary representations $U_A$ and $U_B$ acting on Hilbert spaces $\h_A$ and $\h_B$, and generated by traceless generators $\{J_{A}^{k}\}_{k=1}^{n}$ and $\{J_{B}^{k}\}_{k=1}^{n}$. For every $G$-covariant quantum channel \mbox{$\E:\B(\h_A)\rightarrow \B(\h_B)$} with $d_A\geq d_B$ the following holds:
	\begin{align}
		\Delta(\E) \leq 2 n d_A (d_A-1) \max_{k} &(\|J_{B}^k\|_1 + \|J_{A}^k\|_1)^2\nonumber\\
		& \qquad\times(1-u(\E)).\label{eq:upper}		
	\end{align}
	Moreover, the above also holds for $d_A<d_B$ whenever \mbox{$\Tr(\E(\iden/d_{A})^2)\geq \frac{1}{d_A}$}.
\end{theorem}
To prove the above theorem one starts from Eq.~\eqref{eq:choigeneraldecomp} that yields the general decomposition of a $G$-covariant map into a convex mixture of CP maps $\E^\lambda$ with probabilities $p_\lambda(\E)$. Employing this decomposition and Lemma~\ref{lem:unitaritychoi}, one can then lower bound the deviation from closed dynamics, $(1-u(E))$, with a dimensional constant times $(1-p_{\lambda=0}(\E))^2$. Next, one notes that $\E^{\lambda=0}$ conserves charges (generators), and thus using standard inequalities (e.g., the triangle inequality) the deviation can be upper bounded by a dimension-dependent constant times $(1-p_{\lambda=0}(\E))^2$. Finally, one combines both inequalities to bound $\Delta(\E)$ with $(1-u(\E))$ as in Eq.~\eqref{eq:upper}. The details of necessary calculations can be found in Appendix~\ref{app:upperbounds}.
		
		
\subsection{Lower bounds on deviating charges for $G$-covariant open dynamics}
\label{sec:lowerbounds}

We would now like to find a lower bound on the average total deviation in terms of unitarity. First, however, we need to note that decoherence does not need to lead to the deviation from conservation law. In other words, there may be open (non-unitary) symmetric dynamics that nevertheless conserves charges (generators) for all input states. To illustrate this, let us start with the following semi-trivial example of a non-unitary symmetric dynamics $\E$ for which all conservation laws relevant for the symmetry hold. Consider a two-qubit system where the first qubit transforms under the 1/2-irrep of SU(2) and the second transforms trivially. The conserved charges generating the symmetry are the spin operators on the first system. Let \mbox{$\E_{AB}(\rho\otimes\sigma):=\rho\otimes \E_{B}(\sigma)$}. This is covariant under the symmetry, the conservation laws hold for all states, however it is not a unitary operation as we are free to choose any CPTP $\E_{B}$ on system $B$.
	
More generally, there may exist whole families of non-trivial symmetric channels that are not unitary, but preserve conserved charges for all input states. For example, it is relatively simple to find such a family among unital covariant channels. Theorem~(4.25) from Ref.~\cite{watrous2018theory} tells us that for a unital CPTP map \mbox{$\E:\B(\H_A)\rightarrow\B(\H_A)$} with Kraus operators $\{K_i\}_i$ we have $\E(X) = X$ if and only if $[X,K_i]=0$ for all $i$. Extensions of this result to rotating fixed points may also be found in Ref.~\cite{albert2019asymptotics}. Recall also that any symmetric channel $\E$ admits a Kraus decomposition consisting of ITOs $\{E^{\lambda,\alpha}_m\}_{\lambda,m,\alpha}$, where $\lambda$ labels irreducible representations in $\B(\h_A)$ of multiplicity $\alpha$ and vector component~$m$. Then, it follows that $\E\hc(J_A^k) = \sum_{\lambda,m,\alpha} (E^{\lambda,\alpha}_m)\hc J_A^k E^{\lambda,\alpha}_m $ for all symmetry generators $J^k_A$. Now, since we assumed that $\E$ is unital CPTP map, also $\E^\dagger$ is a unital CPTP map. Thus, we can use the result quoted above and conclude that $\E\hc(J_A^k) = J_A^{k}$ if and only if $[E^{\lambda,\alpha}_m, J_A^k] = 0$ for all $\lambda,m,\alpha$ and $k$. However, $E^{\lambda,\alpha}_m$ transform as ITOs and only $\lambda=0$, corresponding to the trivial representation, commutes with the generators. Therefore, for unital symmetric channels $\E$, conservation laws hold if and only if $\E$ takes the general form:
\begin{equation}
	\E(\cdot) = \sum_{\alpha} E^{\lambda=0,\alpha} (\cdot) (E^{\lambda=0,\alpha})\hc,
\end{equation}
where each Kraus operator $E^{\lambda=0,\alpha}$ commutes with the group action. In general, it may also be possible for conservation laws to hold for non-unital operations, but a full characterisation of the dynamics for which this happens remains open.

As the examples above conserve charges despite decoherence by acting on the multiplicity spaces of the trivial representation $\lambda=0$, one could hope that for groups with multiplicity-free decomposition such a situation will be impossible (and so the conservation law would be robust to decoherence). However, this is not the case. To see this, recall that in Sec.~\ref{sec:reversal} we found the extremal SU(2)-irreducibly covariant channel $\E^{|j_A-j_B|}$ that allowed for spin amplification whenever $d_B>d_A$. At the same time, we showed that there also exists an optimal spin reversal channel $\E^{j_A+j_B}$. Thus, one can always find a parameter $q\in[0,1]$ such that \mbox{$q\E^{|j_A-j_B|}+(1-q)\E^{j_A+j_B}$} preserves all spin components, while at the same time being far from unitary evolution.
	
The above discussion illustrates that probing conservation laws for a physical realisation of a symmetric dynamics is usually not sufficient to decide if there are decoherence effects present. In other words, robustness of conservation laws does not occur for all types of symmetries. Nevertheless, there are particular conditions that guarantee a certain robustness of conservation laws. In such cases, approximate conservation laws hold if and only if the dynamics is close to a unitary symmetric evolution. In particular, for channels with equal input and output dimensions $d_A=d_B$, whenever $\B(\h_A)$ contains a single trivial subspace then there is no symmetric channel other than identity for which conservation laws hold. This is the case for example when $\h_A$ carries an irreducible representation of SU(2). More generally, however, we have the following theorem that provides lower bounds on the deviation from conservation laws in terms of the unitarity. 
\begin{theorem}
	\label{thm:lowerbounds}
	Let $G$ be a connected compact Lie group with unitary representation $U_{A}$ acting on a Hilbert space $\h_A$, and generated by traceless generators $\{J^k_{A}\}_{k=1}^{n}$. Moreover, assume that $\B(\h_A)$ has a multiplicity-free decomposition in terms of irreducible representations. Then, for every $G$-covariant quantum channel \mbox{$\E:\B(\h_A) \rightarrow\B(\h_A)$} the following holds:
	\begin{equation}
		\sqrt{\Delta(\E)}\geq K \|\vv{J}_{A}\| (1-u(\E)) \frac{(d_A-1)(d_A+1)^{1/2}}{2d_A^{5/2}}.
	\end{equation}
	where $K$ is a constant independent of $\E$, defined by 
	\begin{equation}
		K:= \min_{\lambda\neq 0 \in \Lambda} |1-f(\lambda)|,
	\end{equation}
	with $f(\lambda)$ being constant coefficients such that the extremal isolated channel $\E^{\lambda}$ satisfies \mbox{$\E^{\lambda\dagger}(J_{A}^k) = f(\lambda) J_{A}^k$}.
\end{theorem}
\noindent The proof of the above theorem can be found in Appendix~\ref{app:lowerbounds}.
	
		
\subsection{Bounds on deviating charges for SU(2)-covariant open dynamics}		
	
We now turn to investigating robustness of conservation laws for SU(2)-irreducibly covariant channels. We focus on a particular case of covariant quantum channels between spin-$j$ systems, i.e., for $j_A=j_B=j$. In this case it is possible to deduce both upper and lower bounds on average total deviation in terms of unitarity. One of the reasons for this is that dissipation, as given by a symmetric channel $\E$ that is not a unitary, cannot hide in the multiplicity subspace of the trivial representation, as $U_A\otimes U_A^{*}$ has a multiplicity-free decompositions into irreps. 
	

\subsubsection{Expressions for unitarity and deviations}
	
The structure of general SU(2)-irreducibly-covariant channels $\E:\B(\h_A)\rightarrow\B(\h_B)$ presented in Sec.~\ref{sec:su2_structure} gives a simple way to calculate their unitarity. Employing Lemma~\ref{lem:unitaritychoi}, using the decomposition of the Jamio{\l}kowski state given in Eq.~\eqref{eq:jamiolkowski_su2_irrep} and the fact that irreducibly-covariant channels are unital (so that \mbox{$\E(\iden_A/d_A)=\iden_B/d_B$}), one obtains
\begin{align}
	\label{eq:unitarity_su2}
	u(\E)=&\left((2j_A+1)^2\sum_{L=|j_A-j_B|}^{j_A+j_B} \frac{p_L(\E)^2}{2L+1}-\frac{2j_A+1}{2j_B+1}\right)\non
	&\quad \times \frac{1}{(2j_A+1)^2-1},
\end{align}
where $\vv{p}(\E)$ characterises a given SU(2)-covariant channel according to Eq.~\eqref{eq:decomposition_cptp}. In the particular case when \mbox{$j_A=j_B=j$}, so that both input and output dimensions $d=2j+1$, the above yields
\begin{align}
	\label{eq:unitarity_su2_same_dim}
	u(\E)=&\frac{1}{d^2-1}\left(d^2\sum_{L=0}^{2j} \frac{p_L(\E)^2}{2L+1}-1\right).
\end{align}
	
Now, in order to get an expression for the average deviation from a conservation, we first look at \mbox{$\delta J_A^k = \E\hc(J_{B}^k) - J_{A}^k$}, where $k\in\{x,y,z\}$ correspond to the spin angular momentum operators $J_x,J_y,J_z$. We start by noting that, due to Eq.~\eqref{eq:su2_irrep_transform_ITOs} and the fact that ITOs are orthonormal, we have
\begin{equation}
	\E^\dagger(S^l_m)=f_l(\E) T^l_m.
\end{equation}
Next, using the relation between angular momentum operators and ITOs from Eq.~\eqref{eq:su2_ITO}, we can re-write the above expression to arrive at
\begin{equation}
	\E^\dagger(J_B^k)=f_1(\E) \frac{\|\vv{J}_B\|}{\|\vv{J}_{A}\|} J_{A}^{k},
\end{equation}
where we recall that $\|\vv{J}_B\| = \sqrt{j_B(j_B+1)(2j_B+1)}$ and analogously for input system $A$. Next, we use convex decomposition of $\E$ into extremal channels $\E^L$, Eq.~\eqref{eq:decomposition_cptp}, to get
\begin{equation}
	\E^\dagger(J_B^k)=\ \left(\frac{\|\vv{J}_B\|}{\|\vv{J}_{A}\|}\!\!\sum_{L=|j_A-j_B|}^{j_A+j_B} p_L(\E) f_1(\E^L) \right) J_A^k .
\end{equation}
However, we have determined specific closed formulas for  $f_{1}(\E^L)$ in Eq.~\eqref{eq:explicit_f1}, so that combining with the above relations we end up with 
\begin{equation}
	\!\!\!\delta J_A^k = \frac{J_A^k}{2j_A(j_A+1)}\left(\beta_{AB}-\!\!\!\!\!\!\!\sum_{L=|j_A-j_B|}^{j_A+j_B}\!\!\!\!\!\!\! p_L(\E) L(L+1)\!\right)\!,\!
	\label{eq:deltajk}
\end{equation}
with
\begin{equation}
	\beta_{AB}=j_B(j_B+1)-j_A(j_A+1).
\end{equation}
Finally, the spin angular momenta satisfy the following relations (and similar ones for system $B$):
\begin{equation}
	\!\!\tra{J_A^k}=0,\quad \tra{(J_A^k)^2}=\frac{\|\vv{J}_A\|^2}{3}.
\end{equation}
Combining the above with Eq.~\eqref{eq:deltajk} and substituting to general expression for the average total deviation, Eq.~\eqref{eq:dev_su2_partxxxxx}, leads to
\begin{equation}
	\label{eq:dev_su2_part2}
	\!\!\!\!\Delta(\E) = \frac{1}{8j_A(j_A+1)^2}\!\left(\!\beta_{AB}-\!\!\!\!\!\!\!\!\sum_{L=|j_A-j_B|}^{j_A+j_B}\!\!\!\!\!\!\! p_L(\E) L(L+1)\!\right)^{\!\!2}\!\!\!.\!\!\!
\end{equation}
In the particular case when \mbox{$j_A=j_B=j$} the above yields
\begin{equation}
	\label{eq:dev_su2_part3}
	\Delta(\E) = \frac{1}{8j(j+1)^2} \left(\sum_{L=0}^{2j} p_L(\E) L(L+1)\right)^{2}\!\!\!.
\end{equation}
		
		
\subsubsection{Deriving trade-off relations}

We will now show how the unitarity and deviations from conservation laws are related, and obtain both lower and upper bounds on the average deviation from a conservation law of spin angular momenta under a rotationally invariant irreducible channel in terms of its unitarity.
	
\begin{theorem}
	\label{thm:su2_bounds}
	Let $\E:\B(\h_A) \rightarrow\B(\h_A)$ be an SU(2)-irreducibly covariant quantum channel acting on a $j$-spin system. Then, the average total deviation $\Delta(\E)$ from conservation law for spin angular momenta is bounded by the unitarity of the channel $u(\E)$ via the following trade-off inequalities:
	\begin{subequations}
		\begin{align}
			 \sqrt{\Delta(\E)}&\geq\frac{\sqrt{2}j^{1/2}}{(2j+1)^2}(1-u(\E)),\\
			 \sqrt{\Delta(\E)}& \leq\frac{3\sqrt{2} j^{3/2}}{2j+1} (1-u(\E)).
		\end{align}
	\end{subequations}
\end{theorem}
\begin{proof}
	First, using Eq.~\eqref{eq:dev_su2_part3}, we note that given a fixed $p_0$ the deviation $\Delta(\E)$ is maximised for $p_{2j}=1-p_0$. Therefore,
	\begin{equation}
		\Delta(\E)\leq\frac{j}{2}\frac{(2j+1)^2}{(j+1)^2}  (1-p_0)^2,
	\end{equation}
	resulting in the following bound
	\begin{equation}
		\label{eq:inequality_aux}
		1-p_0\geq \frac{j+1}{2j+1}\sqrt{\frac{2\Delta(\E)}{j}}.
	\end{equation}
	To shorten the notation we will now use $d=2j+1$ and $j$ simultaneously. Using Eq.~\eqref{eq:unitarity_su2} we have the following series of equalities and inequalities:
	\begin{align}
		u(\E)&= \frac{1}{d^2-1}\left(d^2\sum_{l=0}^{2j}\frac{p_l^2}{2l+1}-1\right)\nonumber\\
		&\leq \frac{1}{d^2-1}\left(d^2\left(p_0^2+\frac{1}{3}\sum_{l=1}^{2j}p_l^2\right)-1\right)\nonumber\\
		&\leq \frac{1}{d^2-1}\left(d^2\left(p_0^2+\frac{1}{3}\left(\sum_{l=1}^{2j}p_l\right)^2\right)-1\right)\nonumber\\
		&= \frac{1}{d^2-1}\left(d^2\left(p_0^2+\frac{1}{3}\left(1-p_0\right)^2\right)-1\right)\nonumber\\
		&= 1-\frac{2}{3}\frac{d^2}{d^2-1}(1-p_0)(1+2p_0)\nonumber\\
		&\leq 1-\frac{2}{3}\frac{d^2}{d^2-1}(1-p_0)\nonumber\\
		&\leq 1 - \frac{1}{3\sqrt{2}}\frac{2j+1}{j^{3/2}}\sqrt{\Delta(\E)},
	\end{align}
	where the second inequality comes from the fact that the sum of squares of positive numbers is upper bounded by the square of the sum, and the final one from Eq.~\eqref{eq:inequality_aux}.
		
	On the other hand, using Eq.~\eqref{eq:dev_su2_part3} again, we note that given a fixed $p_0$ the deviation $\Delta(\E)$ is minimised for \mbox{$p_{1}=1-p_0$}. Therefore,
	\begin{equation}
		\Delta(\E)\geq\frac{1}{2j(j+1)^2}  (1-p_0)^2,
	\end{equation}
	resulting in the following bound
	\begin{equation}
		\label{eq:inequality_aux2}
		1-p_0\leq (j+1)\sqrt{2j\Delta(\E)}.
	\end{equation}
	To shorten the notation we will use $d=2j+1$ and $j$ simultaneously, and introduce \mbox{$A:=2j(4j+1)$}. We then have the following series of equalities and inequalities:
	\begin{align}
		u(\E)&= \frac{1}{d^2-1}\left(d^2\sum_{l=0}^{2j}\frac{p_l^2}{2l+1}-1\right)\nonumber\\
		&\geq \frac{1}{d^2-1}\left(d^2\left(p_0^2+\frac{1}{4j+1}\sum_{l=1}^{2j}p_l^2\right)-1\right)\nonumber\\
		&\geq \frac{1}{d^2-1}\left(d^2\left(p_0^2+\frac{1}{4j+1}\frac{(1-p_0)^2}{2j}\right)-1\right)\nonumber\\
		&= 1-\frac{d^2}{d^2-1}\frac{(A+1)p_0+(A-1)}{A}(1-p_0)\nonumber\\
		&\geq 1-\frac{2d^2}{d^2-1}(1-p_0)\nonumber\\
		&\geq 1-\frac{1}{\sqrt{2}}\frac{(2j+1)^2}{j^{1/2}} \sqrt{\Delta(\E)},
	\end{align}
	with the second inequality coming from the fact that the sum of squared probabilities, given a constraint on the total probability, is minimised for uniform distribution; and the final inequality coming from Eq.~\eqref{eq:inequality_aux2}.
\end{proof}

\begin{figure}[t!]
	\includegraphics[scale=1]{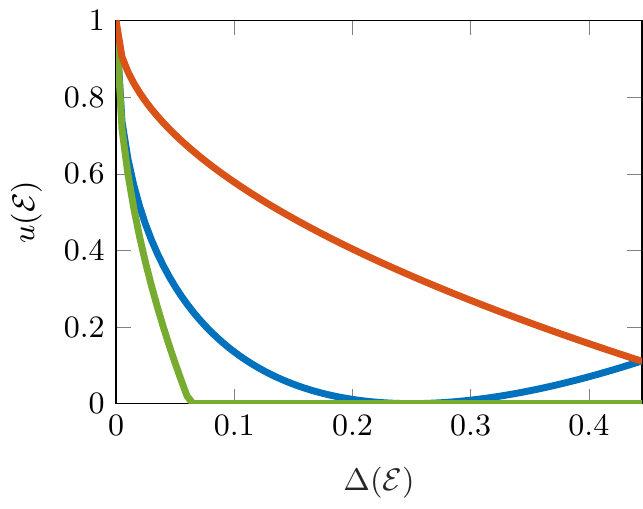}
	\includegraphics[scale=1]{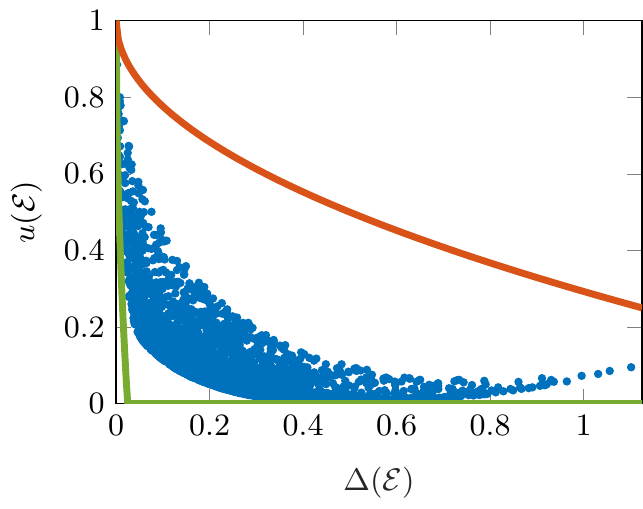}	
	\caption{\textbf{Trade-off between the deviation from angular momentum conservation, $\Delta(\E)$, and unitarity $u(\E)$ for SU(2)-irreducibly-covariant channels.} The middle blue line in the top panel (spin-$1/2$ system) and the blue dots in the bottom panel (spin-$1$ system) represent $[\Delta(\E),u(\E)]$ pairs realised by covariant channels. The top red and bottom green curves give the upper and lower bounds specialised to the case of SU(2) symmetry.\label{fig:su2_tradeoff} }
\end{figure}


\subsubsection{Examples}

Consider first the simplest example of a covariant channel for $j=1/2$. In this case the unitarity is given by
\begin{equation}
	u(\E) = \frac{1}{3}\left(4 p_0(\E)^2 + \frac{4}{3}(1-p_0(\E))^2-1\right),
\end{equation}
and the deviation by
\begin{equation}
	\Delta(\E) = \frac{4}{9}\, \,(1-p_0(\E))^2.
\end{equation}
A straightforward calculation then yields a direct relation between $u(\E)$ and $\Delta(\E)$,
\begin{equation}
	u(\E)=1-4\sqrt{\Delta(\E)}(1-\sqrt{\Delta(\E)}),
\end{equation}
while the bounds from Theorem~\ref{thm:su2_bounds} read
\begin{equation}
\frac{1}{4}	(1-u(\E))\leq \sqrt{\Delta(\E)} \leq \frac{3}{4}(1-u(\E))
\end{equation}
We present the above dependence and bounds in Fig.~\ref{fig:su2_tradeoff}. 
	 
The next simplest case concerns spin-$j$ system with $j=1$. We then have
\begin{equation}
	u(\E) = \frac{1}{8}\left(9 \left(p_0(\E)^2 + \frac{p_1(\E)^2}{3} +\frac{p_2(\E)^{2}}{5} \right) -1\right), 
\end{equation}
where $p_2(\E) = 1- p_0(\E)-p_1(\E)$ and the deviation is given by:
\begin{equation}
	\Delta(\E) = \frac{1}{32}\left(2\,\,p_1(\E)^2 + 6\, \,p_2(\E)^2\right).
\end{equation}
Unitarity $u(\E)$ and deviation $\Delta{\E}$ are no longer directly related, but they constrain each other, so that only some pairs \mbox{$[\Delta(\E),u(\E)]$} are realised by SU(2)-covariant channels. Our bounds then take the form
\begin{equation}
\frac{\sqrt{2}}{9} 	(1-u(\E))\leq \sqrt{\Delta(\E)} \leq \sqrt{2}(1-u(\E)),
\end{equation}
and again we plot them in Fig.~\ref{fig:su2_tradeoff} together with possible pairs \mbox{$(\Delta(\E),u(\E))$}.


\subsection{Bounds on deviating charges for U(1)-covariant open dynamics}
\label{sec:U1bounds}

We now turn to our final example of U(1)-covariant dynamics and the corresponding trade-off between deviations from energy conservation and unitarity of the channel. Throughout this section we employ the notation introduced in Sec.~\ref{sec:u1_structure} while studying the convex structure of U(1)-covariant channels. For simplicity, we focus on the input and output systems of the same dimension, \mbox{$d_A=d_B=d$}, and described by the same Hamiltonian \mbox{$H_A=H_B=H$}.	As we will shortly explain, in this case it is impossible to lower bound unitarity $u(\E)$ given the deviation $\Delta(\E)$, and thus our aim is to upper bound the unitarity $u(\E)$ given the deviation $\Delta(\E)$, i.e., to find the minimal allowed departure from a closed symmetric dynamics that can explain a given deviation from energy conservation.	

	
\subsubsection{Expressions for unitarity and deviations}
	
Substituting the decompositions given in Eq.~\eqref{eq:u1_jamiolkowski} to Eq.~\eqref{eq:unitarity_jamiolkowski}, one obtains the following expression for unitarity of a general U(1)-covariant channel~$\E$:
\begin{equation}
	\label{eq:u1_unitarity}
	u(\E)=\frac{1}{d^2-1}\left(d^2\sum_{\lambda}\gamma(\J^{(\lambda)}(\E))-b^\E\right),
\end{equation}
	where 
\begin{equation}
	b^\E:=\frac{1}{d}\sum_m \left(\sum_n P^\E_{mn}\right)^2,
\end{equation}
describes how far $P^\E$ is from a bistochastic matrix, i.e., $b^\E=1$ when $P^\E$ is bistochastic and $b^\E>1$ otherwise.
	
The expression for the average total deviation $\Delta$ is given by Eq.~\eqref{eq:dev_su2_partxxxxx},
\begin{align}
	\label{eq:dev_u1_part}
	\Delta(\E) &= \frac{1}{d(d+1)} \left(\tra{\delta H}^2 + \tra{\delta H^2}\right). 	
\end{align}
Moreover, since $H$ is diagonal in the energy eigenbasis, the expression for $\E^\dagger(H)$ only involves $\lambda=0$ block, and so $\delta H$ can be easily calculated explicitly. More precisely, one gets
\begin{subequations}
	\begin{align}
		\tra{\delta H}^2 &=\left( \sum_{m,n}  P^\E_{nm} (E_n-E_m) \right)^2,\\
		\tra{\delta H^2} &=\sum_m\left( \sum_{n}  P^\E_{nm} (E_n-E_m) \right)^2.
	\end{align}
\end{subequations}
	
	
\subsubsection{Deriving trade-off relations}
	
First of all, we note that the deviation $\Delta(\E)$ depends only on $P^\E$, while the unitarity $u(\E)$ depends both on $P^\E$ (forming diagonals of $\J^{(\lambda)}(\E)$) and on $L^{(\lambda\neq 0)}(\E)$ (forming the off-diagonal terms of $\J^{(\lambda)}(\E)$). Therefore, it is impossible to lower bound unitarity given the deviation. To see this more clearly, consider the following family of partial dephasing channels (which are U(1)-covariant):
\begin{equation}
	\D_p:=p\D+(1-p) \I
\end{equation}
	with
\begin{equation}
	\D(\cdot):=\sum_n \bra{E_n}(\cdot)\ket{E_n} \ketbra{E_n}{E_n}.
\end{equation}
Clearly, $P^{\D_p}_{mn}=\delta_{mn}$ and so $\Delta(\D_p)=0$. However, the unitarity varies between 1 (for $p=0$) and $1/(d+1)$ (for $p=1$). This is in accordance with our discussion in Sec.~\ref{sec:lowerbounds} concerning general non-existence of the lower bounds. Thus, we focus on deriving upper bounds.
	
We start by noting that each purity term in Eq.~\eqref{eq:u1_unitarity} is upper bounded by \mbox{$\tra{\J^{(\lambda)}(\E)}$} (this simply corresponds to \mbox{$\J^{(\lambda)}(\E)$} being unnormalised projectors). Using this observation, as well as the fact that $b^\E\geq 1$, we get
\begin{equation}
	\label{eq:u1_uni_bound}
	u(\E)\leq\frac{1}{d^2-1}\left(\sum_{\lambda} q_\lambda^2-1\right),
\end{equation}
with
\begin{equation}
	q_\lambda:=\sum_n P^\E_{n+\lambda,n}
\end{equation}
corresponding to the (unnormalised) probability of energy $\lambda$ flowing into the system due to the action of $\E$. Note that Eq.~\eqref{eq:u1_uni_bound} yields a bound on unitarity that is expressed purely in terms of $P^\E$.
	
Now, the crucial point is that for $\lambda\neq 0$ we have $q_\lambda\leq g$, with $g$ denoting the largest number of pairs of energy levels separated by the same energy difference. The minimal value of $g$ is $1$, corresponding to a Hamiltonian $H$ with non-degenerate Bohr spectrum, while the maximal value is $(d-1)$, achieved for a Hamiltonian with an equidistant spectrum. Since \mbox{$\sum_{\lambda} q_\lambda=d$}, this means that the upper bound in Eq.~\eqref{eq:u1_uni_bound} will be strictly smaller than 1 if $q_0<d$. In other words, as soon as there is any energy flow induced by $\E$ (captured by $P^\E_{nn}<1$ for at least one~$n$), unitarity $u(\E)$ will be strictly smaller than 1. 
	
Let us now relate this observation to a concrete bound on $u(\E)$ involving $\Delta(\E)$. First, we introduce the width of the energy spectrum:
\begin{equation}
	\tilde{E}:=E_d-E_1.
\end{equation}
	This allows us to  get the following bound,
\begin{equation}
	\tra{\delta H}^2 \leq \tilde{E}^2 \left( \sum_{m,n\neq m}  P^\E_{nm} \right)^2=\tilde{E}^2 \left( d-q_0 \right)^2.
\end{equation}
Similarly,
\begin{equation}
	\tra{\delta H^2} \leq \tilde{E}^2 \sum_n\left( 1-P^\E_{nn} \right)^2\leq\tilde{E}^2 \left( d-q_0 \right)^2,
\end{equation}
with the second inequality coming from the fact that the sum of squares of positive numbers is upper bounded by the square of the sum. Combining these two bounds together we arrive at
\begin{align}
	\label{eq:u1_dev_bound}
	\Delta(\E) &\leq \frac{2\tilde{E}^2}{d(d+1)} (d-q_0)^2. 	
\end{align}
	
Next, we will rewrite Eq.~\eqref{eq:u1_uni_bound} in a more convenient form as
\begin{equation}
	u(\E)\leq 1-\frac{d^2}{d^2-1}\left(1-\frac{q_0^2}{d^2}-\sum_{\lambda\neq 0} \frac{q_\lambda^2}{d^2}\right).
\end{equation}
For a fixed $q_0$ the right hand side of the above equation is maximised when for some $\lambda'$ we have $q_{\lambda'}=d-q_0$ and $q_\lambda =0$ otherwise. However, this may not be possible due to a constraint $q_\lambda\leq g$. Thus, we need to consider two separate cases. First, assume that $q_0\geq (d-g)$, so that $d-q_0\leq g$ and the constraint is satisfied. We then have
\begin{align}
	\!\!\!\!\!\!\!\!\!u(\E)&\leq 1-\frac{d^2}{d^2-1}\left(1-\frac{q_0^2}{d^2}-\frac{(d-q_0)^2}{d^2}\right)\nonumber\\
	&=1-\frac{2}{d^2-1}q_0(d-q_0)\leq 1-\frac{2(d-g)}{d^2-1}(d-q_0)\nonumber\\
	&\leq 1-\frac{2(d-g)}{d^2-1}\sqrt{\frac{d(d+1)}{2}}\frac{\sqrt{\Delta(\E)}}{\tilde{E}},
\end{align}
where the final inequality comes from Eq.~\eqref{eq:u1_dev_bound}. On the other hand, if $q_0< (d-g)$, then we can upper bound the unitarity by choosing the maximal allowed value \mbox{$q_{\lambda'}=g$}, and the remaining energy flows to \mbox{$q_{\lambda''}=d-g-q_0$}\footnote{Of course this is a very rough bound, since for small $q_0$ it may happen that the remaining energy flows $d-g-q_0$ are still larger than $g$, and should be split over more indices $\lambda$}. This means that
\begin{align}
	\!\!\!\!\!\!\!\!\!u(\E)&\leq 1-\frac{d^2}{d^2-1}\left(1-\frac{q_0^2}{d^2}-\frac{(d-g-q_0)^2}{d^2}-\frac{g^2}{d^2}\right)\nonumber\\
	&=1-\frac{2}{d^2-1}\left(g(d-g)+q_0(d-g-q_0)\right)\nonumber\\
	&\leq 1-\frac{2g(d-g)}{d^2-1}\leq 1-\frac{2g(d-g)}{d^2-1}\frac{d-q_0}{d}\nonumber\\
	&\leq 1-\frac{2g(d-g)}{d(d^2-1)} \sqrt{\frac{d(d+1)}{2}}\frac{\sqrt{\Delta(\E)}}{\tilde{E}}.
\end{align}
As we do not know what the value of $q_0$ really is (we only know what $\Delta(\E)$ is), we need to choose the weaker of the two bounds and thus we end up with
\begin{equation}
	\label{eq:u1_tradeoff}
	u(\E)\leq 1-\frac{g(d-g)}{d-1} \sqrt{\frac{2}{d(d+1)}}\frac{\sqrt{\Delta(\E)}}{\tilde{E}}.
\end{equation}

	
\subsubsection{Example}
	
\begin{figure}[t]
	\centering
	\setlength\figheight{0.45\columnwidth}
	\setlength\figwidth{0.65\columnwidth}
	\def\figscale{1}
	\includegraphics{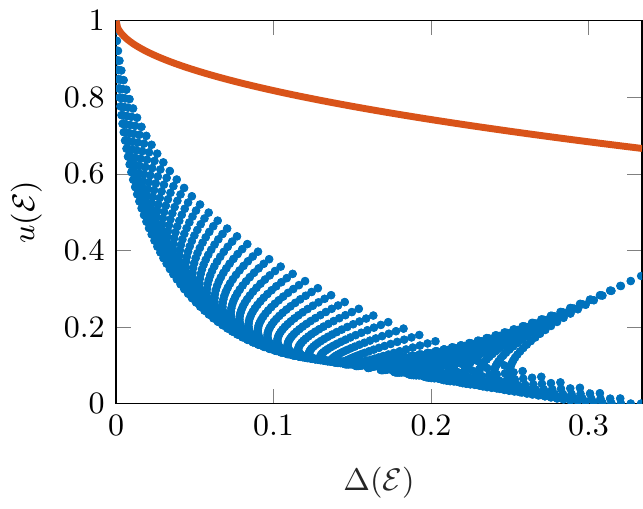}
	\caption{\label{fig:u1_tradeoff} \textbf{Trade-off between the deviation from energy conservation, $\Delta(\E)$, and unitarity $u(\E)$ for U(1)-covariant channels} Each blue dot represents $[\Delta(\E),u(\E)]$ pair for a qubit channel with fixed $P^\E$ and optimal unitarity (the two parameters defining $P^\E$, \mbox{$0\leq P^\E_{00},P^\E_{11}\leq 1$}, are taken as points from the lattice \mbox{$[0,1]\times[0,1]$} with lattice constant 0.02). The orange solid line is the upper bound from Eq.~\eqref{eq:u1_qubit_bound}. }
\end{figure}

Consider a qubit system with unit energy splitting, $\tilde{E}=1$. Average total deviation is then given by
\begin{equation}
	\Delta(\E)=\frac{\left(P_{00}^\E-P_{11}^\E\right)^2+\left(1-P_{00}^\E\right)^2+\left(1-P_{11}^\E\right)^2}{6},
\end{equation}	
while the optimal unitarity (obtained by choosing the blocks of the Jamio{\l}kowski matrix to be unnormalised projectors) for a fixed matrix $P^\E$ is given by
\begin{equation}
	\!\!u(\E)=\frac{\left(P^\E_{00}\!+\!P^\E_{11}\right)^{\!2}\!\!+\!\left(1\!-\!P^\E_{00}\right)^{\!2}\!\!+\!\left(1\!-\!P^\E_{11}\right)^{\!2}\!-b(\E)}{3},\!
\end{equation}
with
\begin{equation}
	b(\E)=1+\left(P^\E_{00}-P^\E_{11}\right)^2. 
\end{equation}
In Fig.~\ref{fig:u1_tradeoff} we present the region of all achievable pairs \mbox{$[\Delta(\E),u(\E)]$} (i.e., for each matrix $P^\E$ we plot the corresponding deviation from energy conservation and the optimal unitarity of the quantum channel transforming energy eigenstates according to $P^\E$), together with our bound from Eq.~\eqref{eq:u1_tradeoff} that for this example reads
\begin{equation}
	\label{eq:u1_qubit_bound}
	u(\E)\leq 1-\frac{1}{\sqrt{3}}\frac{\sqrt{\Delta(\E)}}{\tilde{E}}.
\end{equation}

\kk{
	\section{Directions of Application}
	\label{sec:applications}
	
	In this section we list and discuss potential applications of our results. The basic philosophy is that symmetry constraints (or the lack thereof) lead to simple consequences for easily estimated physical quantities, and thus by tracking these concrete quantities one can probe symmetry structures in the quantum dynamics. This generality suggests a range of applications, and so we shall briefly outline the following:
	
	\begin{itemize}
		\item  \emph{Symmetry-analysis for benchmarking of quantum devices for quantum technologies.}
		\item \emph{Development of universal bounds for the thermodynamics of quantum systems.}
		\item \emph{Applications in the foundations of measurement-theory and quantum tomography.}
		\item \emph{Use in symmetry-checking for Hamiltonian simulations.}
		\item \emph{Specialization to continuous-time (Markov) dynamics of open quantum systems.}
	\end{itemize}
	Concerning the last point, the discrepancy between symmetries and conservation laws for open systems has been used in the context of Lindblad master equations to examine the structure of non-equilibrium steady states~\cite{zhang2020stationary,albert2019asymptotics,albert2014symmetries,buvca2012note,manzano2014symmetry}. Our work has been phrased in terms of general quantum channels, but this can be readily adapted to continuous-time master equations involving dissipative dynamics, which corresponds to a 1-parameter family of quantum channel $\{\E_t\}_{t\ge0}$. It is therefore of interest to specialize our analysis to this regime and determine if the bounds derived here can be tightened under the assumption of Markovian dynamics.
	We leave these to future work, and now give a more detailed description of the other possible directions mentioned above.

	
	\subsection{Tools for benchmarking quantum devices}
	
	Currently, a major theoretical and experimental focus is to develop devices that can process quantum information for future quantum technologies, e.g., for communication, metrology or computing. A central challenge is to assess the degree to which one has good coherent control over the quantum device, which generically undergoes complex dynamics. One could, in principle, answer this question with full quantum process tomography; however, it \emph{scales exponentially with the system size}, requires strong assumptions on state preparation and measurement errors, and thus in many cases is infeasible~\cite{knill2008randomized}. Therefore, approximate methods have been developed that are based on easily estimated physical quantities, to shed light on the level of control over the quantum device. As such, this provides a target for applying the results developed in the current work.
	
	For example, there is a major push to develop a quantum computer that would provide computational abilities that surpass those of traditional classical computers, e.g., in the simulation of highly correlated quantum systems, for quantum chemistry, etc.~\cite{georgescu2014quantum}. Here, it is crucial to assess the degree to which a device can approximately realise a computational gate-set. Among various methods developed, \emph{randomized benchmarking} (RB) is probably the most prominent one~\cite{knill2008randomized,magesan2011scalable}. Randomized benchmarking of a quantum device can be achieved by applying a random sequence of noisy gates $\G_1, \G_2, \dots , \G_m$ from the gate-set to some initial quantum state $\rho$, and then estimating the expectation value $\<Q\>$ of some observable $Q$ on the system.
	
	From the theory of RB, we find that the expectation value, averaged over randomly selected gate sequences of length $m$, decays as~\cite{wallman2015estimating}
	\begin{equation}
	\<Q^2\> = A + B u(\E)^{m-1},
	\end{equation}
	which holds for an arbitrary observable $Q$ on the system. The central role played by the unitarity of a quantum channel in both the present analysis and the RB scenario suggests the following application. One may analyse the noise channels through the lens of particular unitary sub-group actions, which would allow a finer description of the device, essentially due to harmonic analysis. 
	
	More concretely, consider a noisy $n$-qubit system $A$ whose gate-set is subject to some noise model, described by a quantum channel $\E$ on $A$ (assume for simplicity gate-independent noise). We would like to study the structure $\E$, and we might have prior information that suggests it only weakly breaks some symmetry group $G$, e.g., it might be expected that the quantum device respects rotations about the $z$ axis of its physical qubits. Our analysis provides a simple way to test this. Let $\{J^k_A\}_{k=1}^n$ be the generators for this unitary subgroup on~$A$. Now, Theorem~\ref{thm:upperbounds} tells us that if it is the case that $\E$ is covariant with respect to this subgroup then the following relationship between $\Delta(\E)$ and $u(\E)$ must hold
	\begin{equation}
	\Delta (\E) \le 8n d_A(d_A-1) \max_k ||J^k_A||_1^2 [ 1-u(\E)].
	\end{equation}
	Note that both the unitarity $u(\E)$ and $\Delta(\E)$ can be estimated efficiently. Thus, if $\Delta(\E)$ violates the above bound then we deduce that the noise model strongly breaks this particular subgroup. From a Stinespring dilation we can further infer that the noise present involves non-trivial $G$-couplings with the environment. In a similar vein, one can exploit the multiplicity-free lower bound condition from Theorem~\ref{thm:lowerbounds},
	\begin{equation}
	\sqrt{\Delta{\E}} \ge c(d_B, J_B^k) (1-u(\E)),
	\end{equation}
	where $c(d_B, J_B^k)$ is a constant depending on the dimension of $B$. Employing it, one can assess how well the noise channel respects symmetries on multiplicity-free (sub)-systems of the quantum device.
	
	Thus estimating $\Delta(\E)$, for a choice of sub-group, supplements the existing toolkit of randomized benchmarking. Obviously, such constraints should be refined and specialised to the task in hand, but such general bounds can help one circumvent full process tomography, and address abstract structural questions of quantum channels through a simple set of expectation values of observables, and in a manner that links naturally with modern randomized benchmarking techniques.
	
	
	\subsection{Thermodynamics of general quantum systems}
	
	The reduction of a potentially complex quantum process to a small set of distinguished quantities is squarely in the spirit of thermodynamic methods. Indeed, unitarity $u(\E)$ clearly captures the degree to which the quantum channel $\E$ is irreversible~\cite{korzekwa2018coherifying}. Therefore, our analysis can be used to further develop recent works on thermodynamics from a quantum information perspective~\cite{goold2016role}. There, one models thermodynamic transformations by a distinguished set of quantum channels known as \emph{thermal operations}~\cite{horodecki2013fundamental}, which do not inject any free energy into the system or ordered quantum coherence. More precisely, a thermal operation $\E$ describes the evolution of a system with a Hamiltonian $H_S$ prepared in a state $\rho_S$ due to the energy-preserving interaction $U$ with a bath described by a Hamiltonian $H_B$ and prepared in a thermal equilibrium state $\gamma_B$:
	\begin{equation}
	\label{eq:tos}
	\E(\rho)=\trb{B}{U(\rho_S\otimes\gamma_B)U^\dagger},
	\end{equation}
	with $[U,H_S+H_B]=0$. The standard tools within this theory are based on a continuous family of single-shot entropies~\cite{brandao2013second} that are quite unwieldy for describing the coarse-grained thermodynamic behaviour of the system. However, the channels defined in Eq.~\eqref{eq:tos} naturally exhibit a time-translational symmetry~\cite{lostaglio2015description}, i.e., they are covariant with respect to U(1) group generated by the system's Hamiltonian $H_S$, and thus fall under the scope of our results. 
	
	As a consequence, the results we presented here have direct thermodynamic consequences and offer a novel approach to studying conserved charge flows and the decoherence they induce. In particular, the bound in Eq.~\eqref{eq:u1_tradeoff} relates \emph{averaged energy flows} into the system (measured by $\Delta$) during a thermodynamic process, to the amount of irreversible decoherence this process induces (measured by $u$). This means that in order to change the system's energy during a thermodynamic protocol, one has to pay the unavoidable price of deteriorating the quantum superpositions present in the system. We emphasise that this trade-off is universal, i.e., it does not depend on the structure of the thermodynamic bath or the particular interaction Hamiltonian; instead it is based on the fundamental concept of energy conservation. However it is reasonable to expect that these bounds will be most useful in the regime where the effective bath degrees of freedom involved in the process are small. Moreover, our results can also be used to investigate thermodynamics of quantum systems with multiple conserved quantities~\cite{guryanova2016thermodynamics}, e.g., Theorem~\ref{thm:su2_bounds} can be used to upper- and lower-bound the decoherence induced by transferring angular momentum from the bath to the system. We thus see that one could use this work as a starting ground to develop general thermodynamic trade-off relations that, on the one hand, are based on directly measurable physical quantities (conserved charges), and on the other hand capture the unavoidable thermodynamic irreversibility (and related quantum decoherence).
	
	
	\subsection{Measurement-theory and the WAY theorem}
	
	One motivation for the above thermodynamic models can be provided by the measurement theory in the presence of conservation laws. In particular, if one considers an additively conserved quantum $A_S$ (such as energy or momentum), then the celebrated Wigner-Araki-Yanese theorem, or WAY-theorem~\cite{busch2010translation,araki1960measurement,yanase1961optimal}, tells us that the only observables that can be measured in a repeatable manner, are those commuting with $A_S$. Otherwise, a measurement of some observable $B_S$  with $[B_S, A_S] \ne 0$ unavoidably displays a form of irreversibility, which the resource-theoretic thermodynamic framework can capture. 
	
	However, one could equally well extend the present analysis for such WAY-theorem scenarios beyond thermodynamics to general measurement theory. It is well-known that \emph{approximate measurements} of $B_S$ in a state $\ket{\psi}$ can be designed, but these require the measuring apparatus to be prepared in a state $\ket{\xi}$ with large variance of the apparatus observable $A_M$ corresponding to the conserved charge. More precisely, the root mean square error $\epsilon(B_S)$ in the
	measurement of $B_S$ is lower-bounded by~\cite{ozawa2002conservation} 
	\begin{equation}
	\label{eq:ozawa}
	\epsilon^2(B_S)\geq \frac{|\bra{\psi} [A_S,B_S]\ket{\psi}|^2}{4\sigma^2(A_S)+4\sigma^2(A_M)},
	\end{equation}
	where $\sigma(A_S)$ and $\sigma(A_M)$ denote the variances of observables $A_S$ and $A_M$ in the initial states $\ket{\psi}$ and $\ket{\xi}$. However, the original derivation of the bound~\cite{ozawa2002conservation}, as well as further designs of optimal measurement protocols~\cite{ahmadi2013wigner}, assume a perfectly closed unitary evolution of the joint $SM$ system. In this context, a natural extension of the work here would be to provide general analysis of \emph{the second-order moments} of the conserved quantities as a function of the unitarity of the channel. One can then, for example, take into account the charge-conserving coupling of the measured system and apparatus with the environment, and establish universal resource bounds on measurement schemes in the presence of noise. Specifically, such analysis would connect the measurement error of observables under symmetry constraints with both the unitarity of the joint system-apparatus evolution and the variance of the conserved observables for the apparatus. In effect, we envision that the suggested extensions would lead to a finer-grained inequality similar to Eq.~\eqref{eq:ozawa} that would also take into account the leakage of information outside the system plus apparatus setup, and which in the limit of unitary evolution would reduce to the WAY formulation.  In this manner one could obtain unitarity-based generalisations of measurement-error trade-off relations under symmetry constraints. 
	
	Beyond this second-order analysis and motivated by resource-measures, another interesting direction to be explored in the context of tomography, is to develop a more unified treatment of the $n$-th order moments, which could provide greater insight into the degree to which one can perform efficient tomography of quantum channels that are covariant with respect to some symmetry constraint.
	
	
	\subsection{Hamiltonian simulations \& error-mitigation via symmetry checking}
	
	Current quantum hardware is plagued by noise that in the regime of a few hundred qubits with relatively short coherence time cannot be feasibly corrected via fault-tolerant methods of encoding information~\cite{preskill2018quantum}. However, these noisy intermediate scale quantum (NISQ) devices are capable of demonstrating quantum advantage~\cite{arute2019quantum} and, while it is still a major open question, they may possibly lead to computational advantages. Therefore, it is important to devise methods of characterisation, mitigation and benchmarking of noise accumulated during a quantum information processing task. Such methods are useful not only for near-term applications, but they may very well pave the way for improved error-correction techniques. 
	
	In many instances, the problem we aim to solve on a quantum device has a specific structure, so that during the computation only a particular subspace of the full Hilbert space is explored. A major application is Hamiltonian simulation, where the underlying physical system often has many symmetries (particle conservation, rotational symmetry, etc.) that need to be enforced to prepare a time-evolved state or an appropriate ansatz for energy estimation (via the variational quantum eigensolver or other methods \cite{mcardle2019error, gard2020efficient}). For example, one may want to enforce conservation of particle number for a lattice gauge theory simulation, or prepare states with a fixed number of spin orbitals and enforce that wavefunctions resulting from chemistry simulations lie in the antisymmetric subspace. However, noise will affect such state preparation, and so the resulting state will be outside of the feasible subspace determined by the symmetry constraints. 
	
	Our results find a natural application in obtaining bounds on the unitarity of the global noise associated with the \emph{entire} (symmetric) circuit from directly evaluating the strength of the conservation violation. Unitarity, together with average process fidelity (which may be determined via cycle benchmarking \cite{erhard2019characterizing}), can then be used directly to obtain the bounds (or, in some restricted multiplicity-free cases, estimates) for the diamond norm distance between target symmetric unitary and the actual noisy channel implemented by the quantum device. One should note that this information about the unitarity of the noise across the entire circuit \emph{is not} accessible via the randomised benchmarking protocol for unitarity; by analogy, the difference here is the same as between error rate computed via RB giving an estimate of average fidelity and process fidelity, which is more challenging to evaluate. This also raises the question of how the benchmarking unitarity protocol can be extended to estimate the unitarity of noise across an entire specific computation.
	
	We can provide the following simple toy-model experiment to illustrate our point. Suppose $\psi_0=\ketbra{\psi_{0}}{\psi_0}$ is a two-qubit initial state prepared with support on the $\{|01\rangle , |10\rangle \}$ subspace. Let $J = Z\otimes \iden+\iden\otimes Z$ be a generator of symmetry, with $Z$ standing for the Pauli $Z$ operator. The state $\psi_0$ then undergoes an evolution by a (parametrized) unitary $V(\theta) = e^{iZ \otimes Z\theta}$, which is a symmetric unitary as it commutes with the generator. These types of unitaries appear for instance in ansatze for quantum-chemistry simulations~\cite{romero2018strategies}, and the chosen subspace can be viewed as a restriction to a fixed particle number subspace. One can check that $\Tr(J\psi_0) =0$. Under the symmetric unitary, $\ket{\psi_0}$ gets mapped to a target state \mbox{$|\psi_\theta\rangle = V(\theta)|\psi_0\rangle$}, which also satisfies $\Tr(J\psi_\theta) =0$ due to the symmetry. However, should we want to prepare such a state on an actual device, the effective map will be a noisy approximation described by a channel $\tilde{\mathcal{V}}$. The resulting noisy state \mbox{$\tilde{\psi}_\theta = \tilde{\mathcal{V}}(\psi_0)$} may not conserve $J$ if, for instance, the noisy state leaks outside the subspace $\{|01\rangle , |10\rangle \}$. A measurement in the computational basis will allow us to evaluate $\Delta(\psi_0, \tilde{\mathcal{V}}) = |\Tr(J \tilde{\psi}_{\theta})|^{2}$, which puts a direct lower bound on the distinguishability between $\tilde{\mathcal{V}}$ and the target computation $\mathcal{V}$, via Eq.~\eqref{eq:diamondupper}.
}


\section{Conclusion and Outlook}
\label{sec:conclusions}
	
We have studied the relationship between symmetry principles and conservation laws for irreversible dynamics that goes beyond Noether's principle. We established that the two questions posed in the introduction are fundamentally related. On the one hand, we provide the optimal active transformation approximating spin polarisation inversion, but this turns out to be the symmetric channel that achieves maximal deviation from the conservation law of spin angular momenta. Both of these limitations arise as fundamental constraints imposed by quantum theory on the connection between symmetry principles and conservation laws. At the core of these statements lies the convex structure of symmetric channels. 

Generally, classifying the structure of extremal (symmetric) channels~\cite{d2004extremal} is a difficult problem that remains open in the general setting~\cite{ruskai2007some}. For particular symmetries, the structure simplifies significantly and in several situations all extremal channels become isolated, forming a simplex. This was the case of symmetries described by irreducible representations of $SU(2)$ analysed in detail in Ref.~\cite{nuwairan2013su2,nuwairan2015su2}, but it can occur also for finite groups~\cite{mozrzymas2017structure} and Weyl groups~\cite{siudzinska2018quantum}. Channels that are symmetric under an irreducible representation of some compact group are of particular importance in quantum information as their classical capacity is related to their minimal output entropy~\cite{holevo2002remarks,holevo2005additivity,holevo1993note,mozrzymas2017structure}. This simpler structure was also crucial to our analysis of the robustness of conservation laws under symmetric irreversible dynamics.

This work broadly addresses structural aspects of Noether's theorem for general quantum processes, feat that connects with several important developments that aim to understand the constraints symmetry imposes on measurements via the Wigner-Araki-Yanase theorem~\cite{yanase1963information,marvian2012information, ahmadi2013wigner}, on state transformations \cite{gour2017quantum, hebdige2019classification} or the consequences of global symmetries to gauging dissipative dynamics of multipartite systems~\cite{cirstoiu2017global}. \kk{Moreover, we have discussed in more detail the relevance of our work to the development of quantum technologies (benchmarking of quantum devices, quantum thermodynamics and simulations), emphasising the importance of the fact, that the central measure we use in our results is easily accessible experimentally.}

We have restricted our analysis of Noether's principle to symmetric dynamics described by completely positive maps. Violations from conservation laws can occur in a variety of situations, including classical systems with dissipation leading to modified conserved currents and extensions of Noether's theorem for classical Markov processes~\cite{baez2013noether, gough2015noether}. In using the formalism of CPTP maps, there is an assumption that	the quantum system of interest is initially fully decoupled from its environment. A further direction to explore can be the situation when the system is coupled to the environment. This would lead to a local dynamical map corresponding to non-CP noise. We expect the stability of conservation laws under such dynamics to be difficult to characterise solely in terms of the local dynamics on the main system; we conjecture that in such case the upper bounds on deviation from conservation law in terms of unitarity of the (now non-CP) dynamics will no longer hold, due to a strong dependence on the initial system-environment interaction.

Finally, we speculate on the relevance of this work to relativistic quantum information theory, where decoherence induced by relativistic effects~\cite{pikovski2015universal} can have an impact on probing conservation laws for the quantum systems involved.

\vspace{-0.3cm}
	
\section*{Acknowledgements}
	
The authors would like to thank Benoit Collins for suggesting the excellent reference~\cite{nuwairan2013su2} and Iman Marvian for helpful discussions. CC acknowledges support from EPSRC National Quantum Technology Hub in Networked Quantum Information Technologies. This project/publication was also supported through a grant from the John Templeton Foundation.  KK acknowledges support from the ARC via the Centre of Excellence in Engineered Quantum Systems (EQUS), project number CE110001013, as well as by the Foundation for Polish Science through IRAP project co-financed by EU within Smart Growth Operational Programme (contract no. 2018/MAB/5) and through TEAM-NET project (contract no. POIR.04.04.00-00-17C1/18-00) DJ is supported by the Royal Society and also a University Academic Fellowship.


\newpage
\onecolumngrid
\appendix



\section{Unitarity of quantum channels -- alternative formulations}
\label{appendix:unitarity}
	
\begin{lemma}
	For any channel $\E:\B(\h_A)\rightarrow\B(\h_B)$ the unitarity satisfies:
	\begin{equation}
		u(\E) \leq 1
	\end{equation}
	with equality if and only if there exists an isometry $V:\h_A\rightarrow\h_B$ such that $\E(\rho) = V\rho V\hc$ for all $\rho \in \B(\h_A)$
\end{lemma}
\begin{proof}
	$\Longrightarrow$\\
	We illustrate the proof idea with $d_A = 2$ (and no restriction on $d_B$), in this case the unitarity is given by:
	\begin{equation}
		u(\E) = 2 \int \Tr(\E(\psi - \frac{\iden_{A}}{2})^2) d\, \psi
	\end{equation}
	For each $\psi\in \B(\h_A)$ we are free to write the decomposition of identity in many ways such that:
	\begin{equation}
		\iden_A = |\psi\>\<\psi| + |\psi^{\perp}\>\<\psi^{\perp}|
	\end{equation}
	where $|\psi^{\perp}\>$ is the orthogonal complement of the pure state $|\psi\>$ such that $\{|\psi\>, |\psi^{\perp}\> \}$ form an orthonormal basis for $\B(\h_A)$. Thus we can re-write the unitarity in the form:
	\begin{align}
		u(\E) &= 2\int \Tr( \E((\psi - \psi^{\perp})/2)^2 ) d\, \psi \\
		& = 2 \int \Tr( (\E(\psi) - \E(\psi^{\perp})^2/4) d\, \psi \\
		& = \frac{1}{2} \int \Tr(\E(\psi)^2)  + \Tr(\E(\psi^{\perp})^2) - 2\Tr(\E(\psi)\E(\psi^{\perp})d\,\psi
	\end{align}
		
	However for any $\psi$ and its corresponding $\psi^{\perp}$ we have $\Tr(\E(\psi)^2) \leq 1$ and $\Tr(\E(\psi^{\perp})^2) \leq 1$. Since $\E$ is a CPTP map then $\E(\psi)$ and $\E(\psi^{\perp})$ are positive operators so that $\Tr(\E(\psi)\E(\psi^{\perp})\geq 0$.  Putting everything together it follows that:
	\begin{equation}
		u(\E) \leq \frac{1}{2} \int \Tr(\E(\psi)^2)  + \Tr(\E(\psi^{\perp})^2)d\, \psi \leq 1
	\end{equation}
	with equality if and only if $\E(\psi)$ is a pure state for all pure states $\psi$. Therefore $\E$ is an isometry.
		
	More generally (for arbitrary $d_A$ and $d_B$), we have a lot more freedom in rewriting the identity in terms of and orthonormal basis containing $\psi$. Suppose that for every pure state $\psi$ we extend it to an orthonormal basis $\{\psi, \psi_1, ..., \psi_{d_A-1}\}$. With respect to this we can write:
	\begin{equation}
		\iden_A = |\psi\>\<\psi| +|\psi_1\>\<\psi_1| + ... + |\psi_{d_A-1}\>\<\psi_{d_A-1}|
	\end{equation}
	Therefore the unitarity can be written as:
	\begin{equation}
		u(\E) = \frac{d_A}{d_A-1} \int \Tr\left(\E\left((1-\frac{1}{d_A})\psi - \frac{1}{d_A}\sum_{i=1}^{d_A-1} \psi_{i} \right)^2\right) d\, \psi
	\end{equation}
		the above can be expanded since $\E$ is convex linear so that:
	\begin{equation}
		u(\E) = \frac{d_A}{d_{A}-1} \int \Tr\left(\left(\frac{d_A-1}{d_A} \E(\psi) - \frac{1}{d_A} \sum_{i=1}^{d_A-1} \E(\psi_i) \right)^2 \right)
	\end{equation}
	Or equivalently	get that:
	\begin{align*}
		u(\E) =  \int\left( \frac{d_A-1}{d_A} \Tr(\E(\psi)^2)+\frac{1}{d_A(d_A-1)}\sum_{i=1}^{d_A-1} \Tr(\E(\psi_i)^2)  \right.  \left. -\frac{2}{d_A} \sum_{i=1}^{d_A-1} \Tr(\E(\psi)\E(\psi_i))+ 2 \sum_{i,j=1}^{d_A-1 }\frac{\Tr(\E(\psi_i)\E(\psi_j)}{d_A(d_A-1)} \right) d\, \psi 
	\end{align*}
		
	The above holds equally well for any pure state $\psi$ (which is in fact a dummy variable) so the integration remains invariant under $\psi \rightarrow \psi_i = U_i\psi U_i\hc$ for some unitary $U_i$, with the rest of the basis states remaining invariant. This comes from the fact that the Haar measure is an unitarily invariant measure. In this manner we can write $u(\E)$ in $d_A$ different ways. For instance it also holds that:
	\begin{align*}
		u(\E) =  \int\left( \frac{d_A-1}{d_A} \Tr(\E(\psi_i)^2)+\frac{1}{d_A(d_A-1)}\sum_{j\neq i, j=1}^{d_A-1} \Tr(\E(\psi_j)^2)  \right.  + \frac{1}{d_A(d_A-1)} \Tr(\E(\psi)^2)- \frac{2}{d_A}\sum_{j\neq i, j=1}^{d_A-1} \Tr(\E(\psi_i)\E(\psi_j)\\
		  -\frac{2}{d_A}  \Tr(\E(\psi)\E(\psi_i))+ 2 \sum_{j,k\neq i}^{d_A-1 }\frac{\Tr(\E(\psi_k)\E(\psi_j)}{d_A(d_A-1)} 
		 \left. + 2 \sum_{j\neq i}^{d_A-1}\frac{\Tr(\E(\psi)\E(\psi_j)}{d_A(d_A-1)}\right) d\, \psi 
	\end{align*}
	The above holds for all basis states so we have in total $d_A$ equation. Summing all together we notice that we obtain the following
	\begin{align*}
		d_A u(\E) = \int \left(\Tr(\E(\psi)^2 ) + \sum_{i=1}^{d_A-1}\Tr(\E(\psi_i)^2)\right. \left. -\frac{2}{d_A-1} \sum_{j,k, \psi} \tr(\E(\psi_j)\E(\psi_k)) \right) d\, \psi
	\end{align*}
	In the above each term $\Tr(\E(\psi_j)\E(\psi_k)$ appears $d_A-2$ times with coefficient $\frac{2}{d_A(d_A-1)}$ and twice with coefficient $\frac{2}{d_A}$ the latter arising from the equations for which $i=j$ and $i=k$ and the former from the rest of the equations, where we consider the $j,k$ label to include $\psi$ as well. Putting it all together $-\frac{4}{d_A} + 2 \frac{d_A-2}{d_A(d_A-1)} = \frac{2(d_A)- 4 - 4(d_A-1)}{d_A(d_A-1)} = -\frac{2}{d_A-1}$. The quadratic terms $\Tr(\E(\psi_i)^2)$ will appear once with coefficient $\frac{d_A-1}{d_A}$ and then with coefficient $\frac{1}{d_A(d_A-1)}$ in each of the rest $d_A-1$ equations, which sums up to one. 
		
	Now it is always true that $\Tr(\E(\psi_i)^2)\leq 1$ for all $i$ and also since $\E$ is a CPTP map then $\E(\psi_i)$ is a positive operator so $\Tr(\E(\psi_j)\E(\psi_k))\geq 0$. Therefore it follows that:
	\begin{equation}
		d_A u(\E) \leq \int  \left( \Tr(\E(\psi)^2) + \sum_{i=1}^{d_A-1}\Tr(\E(\psi_i)^2) \right) \leq d_A
	\end{equation}
	with equality holding for $\Tr(\E(\psi)^2)=1$, that is whenever $\E(\psi)$ is a pure state or equivalently when $\E$ is an isometry.
		
	$\Longleftarrow$\\
	Conversely, if $\E(X) = V X V\hc$ for some isometry $V:\h_A\rightarrow \h_B$ then $V\hc V = \iden_A$ and in this case we get 
	\begin{align*}
		u(\E) &= \frac{d_A}{d_A-1}\int\Tr(V\left(\psi-\frac{\iden}{d_A}\right)^2V\hc)d\, \psi \\
		&=  \frac{d_A}{d_A-1}\int\Tr(V\left(\psi-2\frac{\psi}{d_A}+ \frac{\iden}{d_A^2}\right)V\hc)d\, \psi \\   
		&= \frac{d_A}{d_A-1}\Tr(V \int\left(\psi-2\frac{\psi}{d_A}+ \frac{\iden}{d_A^2}\right)d \, \psi\, V\hc)\\   
		&= \frac{d_A}{d_A-1}\Tr(V (\frac{\iden}{d_A} - \frac{\iden}{d_A^2}) V\hc)
	\end{align*} 
	where we have used the fact that $\int \psi d\, \psi = \iden/d_A$.  Collecting terms it follows that if $\E$ is an isometry then $u(\E) = \frac{\Tr(VV\hc)}{d_A} $. However trace preserving condition implies that $\Tr(\E(\iden/d_A)) = \Tr(VV\hc)/d_A = 1$ so that $u(\E) = 1$.\\
\end{proof}	

\begin{lemma}
	Given a channel $\E:\B(\h_A)\rightarrow\B(\h_B)$ then the unitarity can be equivalently expressed by:
	\begin{align*}
		(i)& \ \ \ u(\E) = \frac{d_A}{d_A^2-1}\left(d_A\Tr(J(\E)^2) - \Tr(\E(\iden_A/d_A)^2)\right) \\ \nonumber
		(ii)& \ \ \ u(\E) = \frac{d_A}{d_A^2-1}\left(d_A\Tr(\tilde{\E}(\iden/d_A)^2) - \Tr(\E(\iden/d_A)^2)\right) \nonumber
	\end{align*}
\end{lemma}
\begin{proof}
	(i)	Directly from the definition of unitarity we get:
	\begin{align}
		u(\E)&= \frac{d_A}{d_A-1}\left(\int\Tr(\E(\psi)^2)d\,\psi-2\int \Tr(\E(\psi)\E(\frac{\iden_A}{d_A}))d\,\psi\right. \\ \nonumber &\left.+\int\Tr(\E(\iden/d_A)^2)d\,\psi\right)
			\label{eq:unitarityrewrite}
	\end{align}
	We note that when the Haar measure over pure states is properly normalised then the following hold $\int \psi d\,\psi =\frac{\iden_A}{d_A}$ and $\int \psi^{\otimes 2} d\, \psi = \frac{1}{d_A(d_A+1)}\left(\iden_A\otimes\iden_A + SWAP_{A} \right)$ where $\psi\in\B(\h_A)$ and $SWAP_{A}\in \B(\h_A\otimes\h_A)$ is the SWAP operators defined via $SWAP_{A} := \sum_{i,j} |i_A\>|j_A\>\<j_A|\<i_A|$. We also have the following relation $\Tr(\rho^2) = \Tr(SWAP_A\rho\otimes\rho)$ for all $\rho\in \B(\h_A)$. One can similarly define the SWAP operator for system $B$. Therefore it follows that the average output purity is the purity of the Jamiolkowski operator:
	\begin{align*}
		&\int \Tr(\E(\psi)^2)d\,\psi = \int \Tr(SWAP_{B}\,\E(\psi)\otimes\E(\psi))d\,\psi\\ \nonumber
		&= \int\Tr(\E\hc\otimes\E\hc(SWAP_{B}\psi\otimes\psi))d\,\psi \\ \nonumber
		&= \Tr(\E\hc\otimes\E\hc(SWAP_B)\int\psi^{\otimes 2}d\,\psi)\\ \nonumber
		&= \frac{1}{d_A(d_A+1)}\Tr(\E\hc\otimes\E\hc(SWAP_B)(\iden_A\otimes\iden_A + SWAP_A)\\ \nonumber
		&=\frac{1}{d_A(d_A+1)}\Tr(SWAP_B\E\otimes\E((\iden_A\otimes\iden_A + SWAP_A))
	\end{align*}
	where  $SWAP_{B}\in \B(\h_B\otimes\h_B)$ is the swap operator on system $B$.
	One can also show that $\Tr(J[\E]^2) = \Tr(SWAP_B\,\E\otimes\E( SWAP_A)$ by expanding in terms of basis for $A$ and $B$. To check directly denote by $|e_{m}\>_{m=1}^{d_B}$ an orthonormal basis for system $B$. We get that
	\begin{align*}
		&\Tr(SWAP_B \E\otimes\E (SWAP_A)) =\\ &=\sum_{i,j}\Tr(SWAP_{B}\E(|i\>\<j|)\otimes\E(|j\>\<i|))\\ &=\sum_{i,j,m,n}\<e_n|\E(|i\>\<j|)|e_m\>\<e_m|\E(|j\>\<i|)|e_n\> \\ &=\sum_{i,j}\Tr(\E(|i\>\<j|)\E(|j\>\<i|) = d_A^2\Tr(\J[\E]^2)
	\end{align*}
	Putting everything together it follows that
	\begin{equation}
		\int \Tr(\E(\psi)^2)d\,\psi =\frac{1}{d_A(d_A+1)}\left(\Tr(\E(\iden_A)^2) +d_A^2\Tr(\J[\E]^2) \right)
	\end{equation}
	Similarly we have by linearity that:
	\begin{align*}
		\int \Tr(\E(\psi)\E(\iden/d_A))d\,\psi &= \Tr(\E (\int\psi d\,\psi)\E(\iden/d_A))\\
		&= \Tr(\E(\iden/d_A)^2)
	\end{align*}
	Therefore we get that the unitarity is given by
	\begin{equation}
		u(\E) = \frac{d_A^2\left(\Tr(\E(\iden/d_A)^2) + \Tr(\J[\E]^2)\right)}{(d_A^2-1)} - \frac{d_A\Tr(\E(\iden/d_A)^2)}{ (d_A-1)}
	\end{equation}
	and re-arranging we obtain
	\begin{equation}
		u(\E) = \frac{1}{d_A^2-1} \left(d_A^2\Tr(\J[\E]^2) - d_A\Tr(\E(\iden/d_A)^2) \right)
	\end{equation}
	(ii) To show the second part we just need  to check that $\Tr(\J[\E]^2) = \Tr(\tilde{\E}(\iden/d_A)^2)$. First suppose that $V:\h_A\otimes \h_B\otimes\h_E$ is a Stinespring dilation for the channel $\E$. Then the adjoint channel is $\E\hc(Y_B) =V\hc \iden_E\otimes Y_B V$. Moreover suppose that $|e_n\>_{n=1}^{d_B}$ is an orthonormal basis for system $B$ and that $SWAP_B = \sum_{n,m} |e_n\>|e_m\>\<e_m|\<e_n|$ The result then follows from the following argument:
	\begin{align*}
		d_A^2\Tr(\J[\E]^2) &= \Tr( \E\hc\otimes \E\hc (SWAP_B) SWAP_A)\\
		&= \sum_{n,m} \Tr(\E\hc(|e_m\>\<e_n|)\otimes\E\hc(|e_n\>\<e_m| SWAP_A)\\
		&= \sum_{n,m} \Tr(\E\hc(|e_m\>\<e_n|)\E\hc(|e_n\>\<e_m|)\\
		& = \sum_{n,m} \Tr(V\hc \iden_E\otimes|e_m\>\<e_n|VV\hc\iden_E\otimes|e_n\>\<e_m|V)\\
		& =  \Tr(V\hc \Tr_B(VV\hc)\otimes\iden_B V)\\
		& =  d_A\Tr(V\hc \tr_B(V\iden_A/d_A V\hc )\otimes \iden_B V)\\
		& =  d_A\Tr( V\hc \tilde{\E}(\iden_A/d_A)\otimes \iden_B V)\\
		& =  d_A \Tr(\tilde{\E}\hc (\tilde{\E}(\iden_A/d_A))) \\
		& =  d_A^2 \Tr( \tilde{\E} (\iden_A/d_A)^2).
	\end{align*}
\end{proof}
	

\section{Irreducible SU(2)-covariant channels}
\label{app:su2}

\subsection{Liouville representation for extremal SU(2)-irreducible covariant channels}
\label{LiouvilleExtremal}

In Sec.~\ref{sec:su2_structure} we have seen that the set of SU(2)-irreducibly-covariant channels between spin-$j_A$ and spin-$j_B$ systems is fully characterised by its extremal points $\E^{L}:\B(\h_{A})\rightarrow\B(\h_{B})$ with $L$ ranging from $|j_A-j_B|$ to $j_A+j_B$ in increments of one. Since the input and output spaces carry irreducible representations $j_A$ and $j_B$ of SU(2), this means that the decomposition of the operator spaces into irreducible components is multiplicity-free, and therefore the results on the structure of the corresponding Liouville operators holds. For each extremal channel $\E^{L}$ there is a unique vector $\mathbf{f}(\E^{L})$ of coefficients that fully determines it. Moreover, for SU(2) symmetries we can always construct basis of irreducible tensor operators that are Hermitian, which implies that these coefficients are real for any covariant quantum channel. Therefore, each of the vectors $\mathbf{f}(\E^{L})$ represents one of the extremal points that form a simplex in $\mathbb{R}^{d}$, where $d=2\min(j_A, j_B)$.
	
Since we have a full characterisation of the channels $\E^L$, we can give closed form formulas for the vectors $\mathbf{f}(\E^{L})$ in terms of Clebsch-Gordan coefficients. In doing so, we will make use of the Wigner-Eckart theorem. As before, let $\{T^{\lambda}_{k}\}_{k,\lambda}$ and $\{S^{\mu}_{k}\}_{k,\lambda}$ be ITO bases for $\B(\h_{A})$ and $\B(\h_{B})$, respectively. We have that $\E^{L}(T^{\lambda}_k)=f_\mu(\E^{L})S^{\mu}_{k}\delta_{\mu,\lambda}$ for any $L,\lambda,\mu$ and $k$. The vector $\mathbf{f}(\E^{L})$ has entries $f_{\lambda}(\E^{L})$ with $\lambda$ ranging from 1 to $\min(2j_A, 2j_B)$; for $\lambda=0$ trace-preserving condition implies that $f_{0}(\E^{L})=\frac{1}{2j_B+1}$ is constant for all covariant channels, so we will not include it further into the vector definition of $\mathbf{f}(\E^{L})$.
	
Concerning the angular momentum states that form the basis for $\h_{A}$ and $\h_{B}$ as in Sec.~\ref{sec:su2_structure}, for any $\lambda$-irrep there exists labels $m',n'$ and $k$ such that $\bra{j_B,n'}S^{\lambda}_{k}\ket{j_B,m'}\neq 0$. Therefore we can conveniently re-write each coefficient as: 
\begin{equation}
	f_{\lambda}(\E^{L})=\frac{\bra{j_B,n'}\E^{L}(T^{\lambda}_{k})\ket{j_B,m'}}{\bra{j_B,n'}S^{\lambda}_{k}\ket{j_B,m'}},
\end{equation} 
where we re-iterate that at the core of our analysis is that the quantity above is independent of $m',n'$ and $k$, and this is solely as a consequence of covariance of $\E^L$. The numerator can be written in an equivalent form by a basis expansion 
\begin{equation}
	\bra{j_B,n'}\E^{L}(T^{\lambda}_{k})\ket{j_B,m'}=\sum_{m,n}\bra{j_A,n}T^{\lambda}_k\ket{j_A,m}\bra{j_B,n'}\E^{L}(\ket{j_A,n}\bra{j_A,m})\ket{j_B,m'}.
\end{equation}
Therefore, by using the specific action of $\E^{L}$ on angular momentum states given in Eq.~\eqref{eq:su2_extremal} we obtain that:
\begin{align}
	f_{\lambda}(\E^{L})
	&=\sum_{m,n=-j_A}^{j_A}\frac{\bra{j_A,n}T^{\lambda}_k\ket{j_A,m}\bra{j_B,n'}\E^{L}(\ket{j_A,n}\bra{j_A,m})\ket{j_B,m'}}{\bra{j_B,n'}S^{\lambda}_{k}\ket{j_B,m'}} \nonumber \\
	& = \sum_{m,n=-j_A}^{j_A}\sum_{s=-L}^{L} \frac{\bra{j_A,n}T^{\lambda}_k\ket{j_A,m}}{\bra{j_B,n'}S^{\lambda}_{k}\ket{j_B,m'}} \clebsch{j_B}{m-s}{L}{s}{j_A}{m}\clebsch{j_B}{n-s}{L}{s}{j_A}{n}\delta_{n',n-s}\delta_{m',m-s} \nonumber \\
	& = \sum_{s=-L}^{L}\frac{\bra{j_A,n'+s}T^{\lambda}_k\ket{j_A,m'+s}}{\bra{j_B,n'}S^{\lambda}_{k}\ket{j_B,m'}}\clebsch{j_B}{m'}{L}{s}{j_A}{m'+s}\clebsch{j_B}{n'}{L}{s}{j_A}{n'+s}.
\end{align}
To simplify the above expression further we can employ the Wigner-Eckart theorem, which states that the matrix elements of an irreducible tensor operators depend on the vector component labels only trough the Clebsch-Gordan coefficients. In particular:
\begin{equation}
	\bra{j_B,n'}S^{\lambda}_{k}\ket{j_B,m'}=\clebsch{j_B}{m'}{\lambda}{k}{j_B}{n'}\bra{j_B}|S^{\lambda}|\ket{j_B},
\end{equation}
where $\bra{j_B}|S^{\lambda}|\ket{j_B}$ is the reduced matrix element which is independent of $n',m'$ or $k$. We can also write down Wigner-Eckart for the $T^{\lambda}_{k}$ irreducible operator. This leads to the following form for the vector of coefficients for the extremal channel labelled by $L$:
\begin{equation}
	f_{\lambda}(\E^{L})=\frac{\bra{j_A}|T^{\lambda}|\ket{j_A}}{\bra{j_B}|S^{\lambda}|\ket{j_B}}\sum_{s=-L}^{L}\frac{\clebsch{j_A}{m'+s}{\lambda}{k}{j_A}{n'+s}}{\clebsch{j_B}{m'}{\lambda}{k}{j_B}{n'}}\clebsch{j_B}{m'}{L}{s}{j_A}{m'+s}\clebsch{j_B}{n'}{L}{s}{j_A}{n'+s}.
\end{equation}
In particular, since the above factor has no dependence on the labels $m'$, $n'$ and $k$, without loss of generality we can take $k=0$, $m'=n'=j_B$. We thus end up with the following expression:
\begin{equation}
	f_{\lambda}(\E^{L})=\frac{\bra{j_A}|T^{\lambda}|\ket{j_A}}{\bra{j_B}|S^{\lambda}|\ket{j_B}}\sum_{s=-L}^{L}\frac{\clebsch{j_A}{j_B+s}{\lambda}{0}{j_A}{j_B+s}}{\clebsch{j_B}{j_B}{\lambda}{0}{j_B}{j_B}}\clebsch{j_B}{j_B}{L}{s}{j_A}{j_B+s}^2.
	\label{eq:coefficients_spin}
\end{equation}
	
In the particular case when the input and output spaces have the same dimension and both carry the same irrep of SU(2), $j_A=j_B=j$, we obtain the following:
\begin{equation}
	f_{\lambda}(\E^{L})=\sum_{s=-L}^{L}\frac{\clebsch{j}{j+s}{\lambda}{0}{j}{j+s}}{\clebsch{j}{j}{\lambda}{0}{j}{j}}\clebsch{j}{j}{L}{s}{j}{j+s}^{2}.
	\label{eq:coefficients_spin_same}
\end{equation}
	
	
\subsection{Maximal inversion and amplification of spin polarisation vector for SU(2)-covariant channels}
	
Here, we will characterise the range of values that the coefficient $f_{1}(\E^{L})$ takes while varying over all extremal channels $L$. This factor corresponds to how much the spin polarisation can scale (up or down) under a covariant operation. As we will see, due to the particular choice of irrep $\lambda=1$, we can significantly simplify the expressions with Clebsch-Gordan coefficients appearing in Eqs.~\eqref{eq:coefficients_spin}-\eqref{eq:coefficients_spin_same}. We will first analyse the simpler case of same input and output dimension, and then proceed to the general case. In the former case, we find that while the spin polarisation cannot increase, the spin can be inverted up to a factor that is always greater than $-1$. In other words, we show that $-\frac{j}{j+1}\leq f_{1}(\E^L)\leq 1$, where the upper bound is attained for $L=0$, i.e. the identity channel; and the lower bound is attained for $L=2j$, i.e. the extremal channel with the maximal number of Kraus operators. In the latter case, when the output dimension is larger than the input one, we will show that spin polarisation vector can actually be amplified.

	
\subsubsection{Input and output systems of the same dimension}
	
From the explicit formula for $f_{\lambda}(\E^L)$, Eq.~\eqref{eq:coefficients_spin_same}, we have that:
\begin{equation}
	f_{1}(\E^{L})=\sum_{s=-L}^{L}\frac{\clebsch{j}{j+s}{1}{0}{j}{j+s}}{\clebsch{j}{j}{1}{0}{j}{j}}\left(\clebsch{j}{j}{L}{s}{j}{j+s}\right)^{2}.
\end{equation}
Moreover,
\begin{equation}
	\frac{\clebsch{j}{j+s}{1}{0}{j}{j+s}}{\clebsch{j}{j}{1}{0}{j}{j}}=\frac{j+s}{j}
\end{equation}	
and
\begin{equation}
	\clebsch{j}{j}{L}{s}{j}{j+s}^{2}=\frac{(2j+1)!(2j+s)!(L-s)!}{(2j-L)!(L+2j+1)!(L+s)!(-s)!}
\end{equation}
is non-zero for $s\leq 0$. As a result, we have
\begin{equation}
	f_1(\E^L)=\frac{(2j+1)!}{(2j-L)!(L+2j+1)!}\sum_{s=0}^{L}\frac{(j-s)(2j-s)!(L+s)!}{j(L-s)!s!}.
\end{equation}
It turns out that the above expression can be easily evaluated in terms of products of binomial coefficients, so that
\begin{equation}
	f_1 (\E^L) = \binom{L+2j+1}{L}^{-1}\left(\sum_{s=0}^{l}\binom{2j-s}{L-s}\binom{L+s}{s} - \frac{(L+1)}{j}\binom{2j-s}{L-s}\binom{L+s}{s-1}\right).
\end{equation}
We can compute each of the two sums above separately by using combinatorial identities,
\begin{equation}
		\sum_{s=0}^{L}\binom{2j-s}{L-s}\binom{L+s}{s} = \binom{L+2j+1}{L},\qquad \sum_{s=0}^{L}\binom{2j-s}{L-s}\binom{L+s}{s-1} = \binom{L+2j+1}{L-1},
\end{equation}
to obtain a closed form formula for $f_1 (\E^L)$:
\begin{align}
	f_1 (\E^L) &= 1 - \frac{L+1}{j} \binom{L+2j+1}{L-1}\binom{L+2j+1}{L}^{-1}=1-\frac{L(L+1)}{2j(j+1)}
\end{align}
	
Therefore, under any SU(2)-covariant channel, the spin polarisation can either remain the same (whenever $L=0$, which corresponds to the identity channel), decrease by $0 \leq f_{1}(\E^{L})\leq 1$, or get inverted by $f_1(\E^{L})\leq 0$. However, in this scenario the spin polarisation will never increase. The maximal deviation from a conservation law is achieved by the extremal channel $L=2j$ which also achieves the maximal spin inversion of polarisation:
\begin{equation}
	f_{1}(\E^{2j})=-\frac{j}{j+1}.
\end{equation}
	
	
\subsubsection{Input and output systems of different dimensions}
	
We now proceed to the case $j_A\neq j_B$. From Eq.~\eqref{eq:coefficients_spin} for $\lambda= 1$ we have the following:
\begin{equation}
	f_{1}(\E^{L})\frac{\langle j_B \|S^{1}\|j_B\rangle}{\langle j_A \|T^{1}\|j_A\rangle} = \sum_{s=-L}^L \frac{\clebsch{j_A}{j_B+s}{1}{0}{j_A}{j_B+s}}{\clebsch{j_B}{j_B}{1}{0}{j_B}{j_B}}\clebsch{j_B}{j_B}{L}{s}{j_A}{j_B+s}^2.
\end{equation}
For operators on the carrier space $\h_{A}$ for the $j_A$-irrep (and similarly for $j_B$), the decomposition of $\B(\h_{A})$ contains each irreducible representation with multiplicity at most one. For any $j_A> 0$ the 1-irrep will appear once, and the corresponding subspace will be spanned by the ITOs $\{T^{1}_{k}\}_{k}$. The reduced matrix element is independent of the vector label component $k$. Therefore, due to the uniqueness of the 1-irrep, the quantity $\bra{j_A}|T^{1}|\ket{j_A}$ will be uniquely associated with the irreducible subspace of $\B(\h_{A})$ that transforms under the 1-irrep. This implies that $\bra{j_A}|T^{1}|\ket{j_A}$ is independent on the choice of orthonormal ITO basis. Analogous relation holds for $\bra{j_B}|S^{1}|\ket{j_B}$. In face, we can fairly easily determine what the constant factor $\frac{\langle j_B \|S^{1}\|j_B\rangle}{\langle j_A \|T^{1}\|j_A\rangle}$ is. For this we would again use Wigner-Eckart theorem together with the standard form for ITOs $S^1_{0}$ and $T^{1}_0$ to evaluate a particular matrix element. We then get
\begin{equation}
	\bra{j_B}|S^{1}|\ket{j_B}=\frac{\bra{j_B,m}S^{1}_{0}\ket{j_B,m}}{\clebsch{j_B}{m}{1}{0}{j_B}{m}}=\frac{\sqrt{3}}{\sqrt{2j_B+1}},
\end{equation}
where the above makes no assumption on the ITOs $S^{1}$ other than it forming an orthonormal basis for the 1-irrep component. Therefore, the ratio of the reduced matrix elements $S^{1}$ and $T^{1}$ is $\sqrt{\frac{2j_A+1}{2j_B+1}}$.
	
Now, in order to arrive at a closed form formula for $f_{1}(\E^{L})$, we need to combine the above with with binomial expansions for the Clebsch-Gordan coefficients. First notice that one of the terms in the expression for $f_1(\E^{L})$ is given by
\begin{equation}
	\clebsch{j_B}{j_B}{L}{s}{j_A}{j_B+s} ^{2}=\frac{2j_A+1}{2j_B+1}\binom{1+j_A+j_B+L}{j_A-j_B+L}^{-1}\binom{j_A+j_B+s}{L+s}\binom{L-s}{j_A-j_B-s}.
\end{equation}
Remark that for the coefficients to be non-zero we need that $-L\geq s\geq L$ and $j_A-j_B-s\geq 0$, where we recall that $L$ takes one of the positive values in the set $\{|j_A-j_B|, |j_A-j_B|+1,..., j_A+j_B\}$. Therefore, we get
\begin{equation}
	f_{1}(\E^{L})=\sqrt{\frac{j_B(j_B+1)(2j_A+1)}{j_A(j_A+1)(2j_B+1)}}\binom{1+j_A+j_B+L}{j_A-j_B+L}^{-1}\sum_{s=-L}^{j_A-j_B}\frac{j_B+s}{j_B}\binom{j_A+j_B+s}{L+s}\binom{L-s}{j_A-j_B-s},
\end{equation}
where in the summation only the terms for which the two binomials exist contribute, i.e. $j_A-j_B-s\geq 0$ (note that these correspond exactly to non-zero values of the relevant coefficients in the previous summation). Changing the dummy summation variable from $s$ to $w=s+L$ we obtain the alternative formulation:
	\begin{equation*}
	f_{1}(\E^{L})=\sqrt{\frac{j_B(j_B+1)(2j_A+1)}{j_A(j_A+1)(2j_B+1)}}\binom{1+j_A+j_B+L}{j_A-j_B+L}^{-1} \ \sum_{w=0}^{j_A-j_B+L}\frac{j_B-L+w}{j_B}\binom{j_A+j_B-L+w}{w}\binom{2l-w}{j_A-j_B+L-w}.
\end{equation*}
To compute the above, we make use of the following combinatorial property:
\begin{equation}
		\sum_{w=0}^{a}\binom{a+b+w}{w}\binom{c-w}{a-w}=\frac{(1+a+b+c)!}{a!(1+b+c)!}=\binom{1+a+b+c}{a},
\end{equation}
for $c\geq a$ and similarly 
\begin{equation}
	\sum_{w=0}^{a}w\binom{a+b+w}{w}\binom{c-w}{a-w}=\frac{(1+a+b+c)!(1+a+b)}{(a-1)!(2+b+c)!}=\binom{1+a+b+c}{a}\frac{a(1+a+b)}{2+b+c}
\end{equation}
for $a\neq 0$ and $c\geq a$ (if $a=0$ the latter sum clearly becomes zero). Now, employing this we obtain
\begin{equation}
	f_{1}(\E^{L})=\sqrt{\frac{j_B(j_B+1)(2j_A+1)}{j_A(j_A+1)(2j_B+1)}}\left(\frac{j_B-L}{j_B}+\frac{(j_A-j_B+L)(1+j_A+j_B-L)}{2j_B(1+j_B)}   \right),
\end{equation}
and after some simplification we arrive at
\begin{equation}
	f_{1}(\E^{L})=\sqrt{\frac{j_B(j_B+1)(2j_A+1)}{j_A(j_A+1)(2j_B+1)}}\left(\frac{j_A(j_A+1)+ j_B(j_B+1)-L(L+1) }{2j_B(j_B+1)} \right).
\end{equation}

For different extremal channels with $L$ between $|j_B-j_A|$ and $j_A+j_B$ the maximal value is attained for the closest valid value of $L$ to $\frac{j_A-j_B+1}{2}$.  This maximal value is attained for $L=|j_A-j_B|$. We then have two cases. If $j_A\geq j_B$ then
\begin{equation}
	f_{1}(\E^{|j_A-j_B|})=\sqrt{\frac{j_B(j_A+1)(2j_A+1)}{j_A(j_B+1)(2j_B+1)}},
\end{equation}
and if $j_A\leq j_B$ then
\begin{equation}
	f_{1}(\E^{|j_A-j_B|})=\sqrt{\frac{j_A(j_B+1)(2j_A+1)}{j_B(j_A+1)(2j_B+1)}}.
\end{equation}
The minimal value in turn will always be attained by $L=j_A+j_B$, which gives
\begin{equation}
	f_{1}(\E^{j_A+j_B})=-\frac{j_A}{j_B+1}\sqrt{\frac{j_B(j_B+1)(2j_A+1)}{j_A(j_A+1)(2j_B+1)}}.	
\end{equation}
Note that for $j_B>j_A$ there may exist an extremal channel for which the scaling coefficient will be less than $-1$. In other words, the spin polarisation can be effectively inverted in this case.

\section{Convex structure of symmetric channels}
\label{app:convexstruct}

\begin{theorem}
	Let $G$ be a compact group with representations $U_A$ and $U_B$ acting on Hilbert spaces $\h_A$ and $\h_B$. Suppose that $U_B\otimes U_A^{*}$ is a multiplicity-free tensor product representation with non-equivalent irreps labelled by elements of a set $\Lambda$, and that $U_A$ is an irrep. Then, the convex set of $G$-covariant quantum channels $\E:\B(\h_A)\rightarrow\B(\h_B)$ has $|\Lambda|$ distinct isolated extremal points given by channels $\E^{\lambda}$ for $\lambda\in \Lambda$. Each $\E^{\lambda}$ can be characterised by the following: 	
	\begin{enumerate}
		\item A unique Jamio{\l}kowski state 
		\begin{equation}
		\J(\E^{\lambda}) = \frac{\iden^{\lambda}}{d_\lambda}.
		\end{equation} 
		\item Kraus decomposition $\{E^{\lambda}_{k}\}_{k=1}^{d_\lambda}$ such that:
		\begin{equation}
		\E^{\lambda}(\rho) = \sum_{k} E^{\lambda}_{k} \rho (E^{\lambda}_k)\hc,
		\label{eq:A2}
		\end{equation}
		with $E^{\lambda}$ forming a $\lambda$-irreducible tensor operator transforming as \mbox{$U_{B}(g) E^{\lambda}_k U_{A}(g)\hc = \sum_{k'} v^{\lambda}_{k'k}(g) E^{\lambda}_{k'}$}, where $v^{\lambda}_{k'k}$ are matrix coefficients of the $\lambda$-irrep.\\
		\item A symmetric isometry $W^{\lambda}:\h_A \rightarrow\h_B\otimes\h^{\lambda}$ such that
		\begin{equation}
		\E^{\lambda}(\rho) =\Tr_{\h^{\lambda}}( W^{\lambda}\rho (W^{\lambda})\hc).
		\end{equation}
		Also, the minimal Stinespring dilation dimension for $\E^\lambda$ is given by $d_\lambda$.
	\end{enumerate}
\end{theorem}
\begin{proof}
	$1. \implies 2$
	Any square root factorisation of the Choi-Jamio{\l}kowski state gives a set of Kraus operators. In this case, it is trivial to compute the square root operator $R$ of $\J(\E^{\lambda}) = R\hc R$ and this is given simply by $R = \frac{\iden^{\lambda}}{\sqrt{d_{\lambda}}}$. Note that $R$ is not unique and any $\tilde{R} = WR$ for arbitrary unitary $W$ will also result in a valid square-root factorisation. This freedom is then reflected in the non-uniqueness of Kraus operators. Since $R$ is supported only on the $\lambda$-irrep subspace of dimension $d_{\lambda}$ this implies there will only be $d_{\lambda}$ non-zero row vectors in $R$. The non-zero row vectors of $R$ will be given by $\{\<r_i|\}_{i=1}^{d_\lambda}$, such that $\J(\E^{\lambda}) = \sum_{i=1}^{d_\lambda} |r_i\>\<r_i|$. In particular $|r_i\>\in\h_B\otimes\h_A$ and they will form an orthogonal basis for the $\lambda$-irrep subspace of $\h_B\otimes \h_A$ under the tensor product representation $U_B\otimes U_A^{*}$. This is enough to ensure that they transform irreducibly such that $U_B(g)\otimes U_A^{*}(g) |r_k\> = \sum_{k'} v^{\lambda}_{k'k} (g)|r_{k'}\>$, where $v^{\lambda}_{k'k}$ are the matrix coefficients for the $\lambda$-irrep. Under the inverse of the vectorisation operation, there exist a set of operators $E^{\lambda}_i$ represented by $d_{B}\times d_A$ matrices such that $\vc{E^{\lambda}_i} = |r_i\>$, and therefore $\{E^{\lambda}_i\}_i=1^{d_{\lambda}}$ will give a particular Kraus decomposition of $\E^{\lambda}$. Moreover since $|r_{k}\>$ transform as a $\lambda$-irrep then $\{E^{\lambda}_k\}_{k=1}^{d_{\lambda}}$ will form a $\lambda$-irrep ITO. Moreover, since the rank of $\J(\E^{\lambda})$ is $d_\lambda$ this gives a minimal Kraus representation of $\E^{\lambda}$.
	
	$2 \implies 1$ Conversely, given a Kraus decomposition as in Eqn.~\ref{eq:A2} then its corresponding Choi operator will take the form $\J(\E^{\lambda}) = \sum_{k=1}^{d_\lambda} \vc{E^{\lambda}_k}\vcb{E^{\lambda}_k}$. Moreover this is non-trivial only on the on the $\lambda$-irrep subspace of $\h_B\otimes \h_A$ under the tensor product $U_B\otimes U_A^{*}$ where it acts as the identity since this $\lambda$-irrep subspace is spanned by an orthonormal basis $\{\vc{E^{\lambda}_k}\}_{k=1}^{d_{\lambda}}$.
	
	$3\implies 2$ Given a Stinespring dilation $W^{\lambda}$ of $\E$ on an environment $\h^{\lambda}$ carrying $V^{\lambda}$ the $\lambda$-irreducible representation of $G$, then without loss of generality suppose $\{|\lambda, k\>\}_{k=1}^{d_{\lambda}}$ forms an orthonormal basis for $\h^{\lambda}$. Then define $E^{\lambda}_k : = \<\lambda,k|W^{\lambda}$. This is a linear operator from $\h_A$ to $\h_B$. Moreover, $W^{\lambda}$ is symmetric so $U_{B}(g) \otimes V^{\lambda}(g) W^{\lambda} U_{A}(g) = W^{\lambda}$ for all $g\in G$. Therefore $U_{B}\hc(g) E^{\lambda}_{k} U_{A}(g) = U_{B}\hc(g) \<\lambda,k| W^{\lambda} U_{A}(g) = U_{B}\hc(g)\<\lambda,k|U_{B}(g)\otimes V^{\lambda}(g) W^{\lambda} = \<\lambda,k|V^{\lambda}(g)|\lambda,k'\> \<\lambda,k'|W^{\lambda}$. This implies that $\{E^{\lambda}_k\}_{k=1}^{d_{\lambda}}$ transform as ITOs under the group action.
	
	$2\implies 3$ Conversely, we show that there exists a symmetric isometry $W^{\lambda}$ defined by $E^{\lambda}_{k} = \<\lambda,k|W^{\lambda}$ with $W^{\lambda}:\h_A\rightarrow \h_B\otimes \h_{\lambda}$ where $\h_{\lambda}$ has a standard orthonormal basis $|\lambda,k\>$ that transforms under the $\lambda$-irrep. Since $\E^{\lambda}$ is CPTP then $\sum_{k} (E^{\lambda}_k)\hc E^{\lambda}_k = \iden_{A}$ then $\sum_{k}(W^{\lambda})\hc |\lambda,k\>\<\lambda,k| W^{\lambda} = \iden$. However $|\lambda,k\>$ form complete orthonormal basis for $\h_{\lambda}$ so $\sum_{k}|\lambda,k\>\<\lambda,k| = \iden$ on $\h^{\lambda}$. As such $(W^{\lambda})\hc W^{\lambda} = \iden_{A}$ and therefore $W^{\lambda}$ is indeed an isometry. Moreover $W^{\lambda}$ is symmetric because $E^{\lambda}_{k}$ transform as an ITO.
	
\end{proof}


\section{Proof of Theorem~\ref{thm:upperbounds}}
\label{app:upperbounds}
	
\begin{proof}
	Since $\E$ is a $G$-covariant channel, it follows from Eq.~\eqref{eq:choigeneraldecomp} that its corresponding Jamio{\l}kowski state can be written as:  
	\begin{equation}
		\J(\E) = \bigoplus_{\lambda\in \Lambda} p_{\lambda}(\E) \frac{\iden^{\lambda}}{d_{\lambda}}\otimes \rho^{\lambda}.
		\label{eq:choidecp}
	\end{equation}
	We start by bounding unitarity in terms of $p_0$. From Lemma~\ref{lem:unitaritychoi}, the unitarity can be evaluated in terms of $\J(\E)$:
	\begin{equation}
		u(\E)=\frac{d_A}{d_A^2-1}\left[d_A\Tr(\J(\E)^2)-\Tr(\E(\iden/d_A)^2)\right].
	\end{equation} 
	As purity remains invariant under a unitary change of basis, we can compute $\Tr(\J(\E)^2)$ by using its block diagonal structure:
	\begin{align}
		\Tr(\J(\E)^2) & = \sum_{\lambda\in \Lambda} \frac{p_{\lambda}(\E)^2}{d_{\lambda}^2} \Tr(\iden^{\lambda}\otimes (\rho^{\lambda})^2) = \sum_{\lambda\in \Lambda} \frac{p_{\lambda}(\E)^2}{d_{\lambda}} \Tr( (\rho^{\lambda})^2).
	\end{align}
	Therefore we get
	\begin{equation}
		u(\E)= \frac{d_A}{d_A^2-1}\left[d_{A}\sum_{\lambda\in \Lambda} \frac{p_{\lambda}(\E)^2}{d_{\lambda}}\Tr((\rho^{\lambda})^{2})- \Tr(\E(\iden/d_A)^2)\right].
	\end{equation}
	Now, since we assumed that the output purity for the maximally mixed input state is lower bounded by \mbox{$1/d_A$} and because purity of a density matrix is always upper bounded by 1, it follows that
	\begin{equation}
		u(\E)\leq \frac{d_A}{d_A^2-1} \left[ d_A\sum_{\lambda\in \Lambda} \frac{p_{\lambda}(\E)^2}{d_{\lambda}} - \frac{1}{d_A}\right],
	\end{equation}
	or, equivalently,
	\begin{align}
		1-u(\E) \geq \frac{d_A^2}{d_A^2-1} \left[1 - \sum_{\lambda\in \Lambda} \frac{p_{\lambda}(\E)^2}{d_{\lambda}}\right].
	\end{align}
	Furthermore, recall that $p_{\lambda}(\E)$ is a normalised probability distribution, so:
	\begin{equation}
		1- \sum_{\lambda\in \Lambda} \frac{p_{\lambda}(\E)^2}{d_{\lambda}} = \sum_{\lambda\in \Lambda} \frac{d_{\lambda}-1}{d_{\lambda}} p_{\lambda}(\E)^2 + \sum_{\lambda\neq \mu} p_{\lambda}(\E) p_{\mu}(\E).
	\end{equation}
	Finally, for connected compact Lie groups, there is a single one-dimensional irreducible representation that is given by the trivial irrep. We will denote it by $\lambda=0$ for convenience. Then, for $\lambda\neq 0$ we have $\frac{d_\lambda-1}{d_\lambda}\geq \frac{1}{2}$ and so we can obtain the following lower bound:
	\begin{align}
		1-u(\E) \geq \frac{d_A^2}{d_A^2-1}\left( \frac{1}{2}\sum_{\lambda\neq 0} p_{\lambda}(\E)^2 +\frac{1}{2}\sum_{\lambda\neq \mu\neq 0} p_{\lambda}(\E) p_{\mu}(\E)   \right),
	\end{align}
	which can be conveniently rewritten as 
	\begin{equation}
		\label{eq:unitarityline}
		1-u(\E)\geq \frac{d_A^2}{2(d_A^2-1)} (1-p_0)^2.
	\end{equation}
		
	We now proceed to bounding the average total deviation $\Delta(\E)$ in terms of $p_0$. First, we will simplify the expression for $\Delta(\E)$ given in Eq.~\eqref{eq:dev_su2_partxxxxx},
	\begin{equation}
		\Delta(\E) = \frac{1}{d_A(d_A+1)} \sum_{k=1}^{n} \left(\Tr((\delta J^{k}_A)^2)+ (\Tr(\delta J{k}_A))^2\right)
	\end{equation}
	with $\delta J^{k}_A = \E\hc(J^{k}_B) -J^{k}_A$. Since $J^{k}_A$ and $J^{k}_B$ are generators of the unitary representation of the compact Lie group~$G$, they are traceless and Hermitian. Moreover, they live in the irreducible representation of $\B(\h_A)$ and $\B(\h_B)$ that is isomorphic to the adjoint representation (note that, unless $\h_A$ and $\h_B$ are trivial representations, the bounded operator spaces will always have a trivial and adjoint representation). Thus,
	\begin{equation}
		\Tr(\delta J^{k}_A) = \Tr(\E\hc(J^{k}_B)) - \Tr(J^{k}_A) = \Tr(\E\hc(J^{k}_B)) = 0.
	\end{equation}
	The last equality comes from the fact that $\E\hc$ is a symmetric operation, and it will map onto operators in $\B(\h_B)$ fully supported on the adjoint irreducible representations (and multiplicities thereof). This subspace is orthogonal (relative to the Hilbert-Schmidt norm) to the trivial representation in $\B(\h_B)$ where the identity lives. So the deviation from conservation laws reduces to:
	\begin{equation}
		\Delta(\E) = \frac{1}{d_A(d_A+1)}\sum_{k=1}^{n} \Tr((\E\hc(J^{k}_B)-J^{k}_A)^2).
	\end{equation}

	Next, since $\J(\E)$ has the block-diagonal form given in Eq.~\eqref{eq:choidecp}, we can construct CP maps \mbox{$\E^\lambda:\B(\h_A)\rightarrow \B(\h_B)$} associated to each block $\lambda\in \Lambda$, i.e., their Jamio{\l}kowski states $\J^{\lambda}(\E)$ are given by $\frac{\iden^{\lambda}}{d_{\lambda}}\otimes \rho^{\lambda}$ with  $\rho^{\lambda}(\E)$ acting on the multiplicity space. The original channel is then simply given by the convex combination,
	\begin{equation}
		\E = \sum_{\lambda\in\Lambda} p_\lambda(\E) \ \E^{\lambda}.
	\end{equation}
	Consequently, there is a Kraus decomposition for each $\E^{\lambda}$ such that
	\begin{equation}
			\E^{\lambda} = \sum_{i} E^{\lambda}_{i}\rho (E^{\lambda}_i)\hc,
	\end{equation}
	where $E^{\lambda}_i$ transform irreducibly under the group action and span an irreducible $\lambda$-subspace with multiplicity at most $m_{\lambda}$, so that $i$ can range from $1$ up to $d_{\lambda}m_{\lambda}$. In general, a given $\E^{\lambda}$ will be a trace non-increasing CP operation, and the original trace-preserving property of $\E$ can be written as
	\begin{equation}
			\sum_{\lambda\in\Lambda} p_{\lambda}(\E) \sum_{i}(E^{\lambda}_i)\hc E^{\lambda}_i = \iden.
	\end{equation}
	
	In terms of the above considerations we can re-write the crucial term appearing in the expression for $\Delta(\E)$ as follows:
	\begin{align}
		\E\hc(J_{B}^k) -J_{A}^k & = \sum_{\lambda\in\Lambda}p_{\lambda}(\E) \sum_{i} (E^{\lambda}_i)\hc J^{k}_B E^{\lambda}_i - (E^{\lambda}_i)\hc E^{\lambda}_k J^{k}_A=: \sum_{\lambda\in\Lambda}p_{\lambda}(\E) M^{\lambda,k},
		\label{eq:deviationJ}
	\end{align}
	so that
	\begin{equation}
		\label{eq:deviationJ_2}
		\Delta(\E) = \frac{1}{d_A(d_A+1)}\sum_{k=1}^{n} \tr\left(\left(\sum_{\lambda\in\Lambda}p_{\lambda}(\E) M^{\lambda,k}\right)^2\right).
	\end{equation}
	Because $J^{k}_B$ and $J^{k}_A$ are the generators of the symmetry it follows that whenever $\lambda=0$ (the trivial representation), the Kraus operators $E^{\lambda=0}_i$ transform trivially under the group action, and so $J^{k}_B E^{\lambda=0}_i = E^{\lambda=0}_i J^{k}_A$ (where $i$ in this case may label the possible multiplicities of the trivial representation). We remark that this condition is equivalent to the definition of irreducible tensor operators in terms of the generators of the symmetry. Therefore, in Eqs.~\eqref{eq:deviationJ}-\eqref{eq:deviationJ_2} the terms with $\lambda= 0$ vanish. Now, recall that the Schatten $p$-norm of a linear operator $A$ between $\h_B$ and $\h_A$ is defined as $\|A\|_{p}^{p} := \Tr(|A|^{p})$ with $|A|=\sqrt{A\hc A}$. Using H\"{o}lder's inequality followed by a triangle inequality we get: 
	\begin{align}
		\tr\left(\left(\sum_{\lambda\in\Lambda}p_{\lambda}(\E) M^{\lambda,k}\right)^2\right)
		&\leq \left\|\sum_{\lambda\neq 0} p_{\lambda}(\E)M^{\lambda,k} \right\|_{1} \left\|\sum_{\lambda\neq 0} p_{\lambda}(\E)M^{\lambda,k} \right\|_{\infty}\\
		&\leq \sum_{\lambda\neq 0} p_{\lambda}(\E)\|M^{\lambda,k} \|_{1} \sum_{\lambda\neq 0} p_{\lambda}(\E)\|M^{\lambda,k} \|_{\infty}\\
		&\leq \left(\sum_{\lambda\neq 0} p_{\lambda}(\E)\right)^2 \left(\max_{\lambda\neq 0} \|M^{\lambda,k}\|_{1}\right)^{2},
	\end{align}
	with the last inequality coming from \mbox{$\|\cdot\|_\infty\leq \|\cdot\|_1$}. Since $p_{\lambda}(\E)$ forms a probability distribution over $\lambda\in \Lambda$ we get the following bound on the deviation:
	\begin{equation}
		\Delta(\E) \leq \frac{n}{d_A(d_A+1)}(1-p_0)^2  \left(\max_{k,\lambda\neq 0}\|M^{\lambda,k}\|_1\right)^2.
	\end{equation}
	Furthermore, we can bound the term $\|M^{\lambda,k}\|_1$ for any $\lambda$ and $k$. This follows from triangle inequality and submultiplicativity of the Schatten $p$-norms:
	\begin{align}
		\|M^{\lambda,k}\|_{1} &\leq\sum_{i} \|(E^{\lambda}_i)\hc J^{k}_B E^{\lambda}_i\|_1 + \|(E^{\lambda}_i)\hc E^{\lambda}_i J^{k}_A\|_1\leq \sum_{i} \|E^{\lambda}_i\|_1^2\left( \|J^{k}_{B}\|_1 + \|J^{k}_{A}\|_1\right).
	\end{align}
	However,
	\begin{equation}
		\|E^{\lambda}_i\|_1 = \sum_{s} s\left(\sqrt{E^{\lambda\dagger}_k E^{\lambda}_k}\right)
	\end{equation}
	with $s(A)$ denoting the singular values of operator $A$. Since $E^{\lambda\dagger}_i E^{\lambda}_i\geq 0$ then it follows that $\|E^{\lambda}_i\|_{1}^{2} \leq d_{A} \Tr(E^{\lambda\dagger}_i E^{\lambda}_i)$. We also have that $\sum_{i} \Tr(E^{\lambda\dagger}_i E^{\lambda}_i) = d_A$, as the Jamio{\l}kowski states $\J(\E^{\lambda})$ satisfy $\tr(\J(\E^{\lambda})) = 1$, or equivalently $\Tr(\E^\lambda(\frac{\iden}{d_A}))=1$. Then we get the following upper bound on $M^{\lambda,k}$:
	\begin{equation}
		\|M^{\lambda}\|_1\leq d_{A}^2 \max_k (\|J^{k}_B\|_1 + \|J^{k}_A\|_1).
	\end{equation}
	Therefore, we get the following upper bound on the deviation from a conservation law:
	\begin{equation}
		\Delta(\E) \leq n\frac{d_A^3}{d_A+1} (1-p_0)^2 \max_{k} (\|J^{k}_B\|_1+\|J^{k}_A\|_1)^2,
	\end{equation}
	where we recall that $n$ is the number of generators and $p_0 = p_{0}(\E)$.
	Combining the above with Eq.~\eqref{eq:unitarityline} we finally obtain
	\begin{equation}
		\Delta(\E) \leq 2n(d_A-1)d_A \ \max_{k} (\|J^{k}_B\|_1 + \|J^{k}_A\|_1)^{2} (1-u(\E)).
	\end{equation}
\end{proof}

	
\section{Proof of Theorem~\ref{thm:lowerbounds}}
\label{app:lowerbounds}
	
\begin{proof}
	First, note that since $\B(\h_A)$ has a multiplicity-free decomposition, in particular there will be exactly one $\lambda=0$ irrep in the decomposition of $U_A\otimes U_A^{*}$, and it will correspond to the identity operator in $\B(\h_A)$. Consequently, any such symmetric channel $\E$ will necessarily be unital. Moreover, $U_A$ in this case must be an irrep. Otherwise, for each irrep appearing in the decomposition of $U_A$ there would be a trivial irrep in the decomposition of $U_A\otimes U_A^*$ and, by assumption, there is just one such trivial irrep.
		
	Recall that the deviation takes the form:
	\begin{equation}
		\Delta(\E) = \frac{1}{d_A(d_A+1)}\sum_{k=1}^{n} 	\Tr((\E\hc(J_{A}^k)-J_{A}^k)^2).
	\end{equation}
	Moreover, from Theorem~\ref{thm:extremalchannelmultipfree} it follows that $\E$ has a decomposition in terms of isolated extremal channels \mbox{$\E^{\lambda}:\B(\h_A) \rightarrow \B(\h_A)$}, with corresponding Jamio{\l}kowski states $\J(\E^{\lambda}) =\frac{\iden^{\lambda}}{d_{\lambda}}$ acting as identity on the $\lambda$-irrep subspace of $U_A\otimes U_A^{*}$. Therefore,
	\begin{equation}
		\E = \sum_{\lambda\in \Lambda} p_{\lambda}(\E) \E^{\lambda},
	\end{equation}
	with $p_{\lambda}(\E)$ being a probability distribution that depends on $\E$. Now, the multiplicity-free decomposition also ensures that $\E^{\lambda\dagger}(J_{A}^k) = f(\lambda) J_{A}^k$ for some fixed real coefficient $f(\lambda)$ that is associated with the fixed extremal point $\E^{\lambda}$ and independent of $k$. Thus,
	\begin{equation}
		\Delta(\E) = \frac{\|\boldsymbol{J}_{A}\|^{2}}{d_A(d_A+1)} \left(\sum_{\lambda} p_{\lambda} (1- f(\lambda))\right)^2,
	\end{equation}
	where analogously to the previous considerations we have introduced
	\begin{equation}
		\|\boldsymbol{J}_{A}\|^{2} :=\sum_{k=1}^{n} \operatorname{tr}\left(\left(J^{k}_{A}\right)^{2}\right).
	\end{equation}
		
	It is clear that $f(\lambda) = 1$ if and only if $\lambda = 0$. This is because for $\lambda=0$ we deal with the identity channel and so $f(\lambda)=1$; conversely, $\lambda=1$ means that the $J_k$ operators are fixed points of the unital CPTP map, and so they commute with the Kraus operators, and this happens only for Kraus operators transforming as $\lambda=0$. Moreover, without loss of generality, we may assume that $|f(\lambda)|\leq 1$. This follows from a result of Ref.~\cite{perez2006contractivity}, which states that for unital trace-preserving channels the induced $p$-norm is contractive for all $1< p \leq \infty$. That means that
	\begin{equation}
		\|\E^{\lambda\dagger}\|_{p} := \sup_{X\in\B(\h)}\frac{\|(\E^{\lambda})\hc(X)\|_p}{\|X\|_p}\leq 1,
	\end{equation}
	because $\E^{\lambda}$ are unital CPTP maps due to the fact that we deal with an irrep system. Then, it follows that
	\begin{equation}
		\Delta(\E) \geq \frac{\|\vv{J}_{A}\|^2}{d_A(d_A+1)} (1-p_0)^2 \min_{\lambda\neq 0\in \Lambda} |1-f(\lambda)|^{2}=:\frac{\|\vv{J}_{A}\|^2}{d_A(d_A+1)} (1-p_0)^2 K^2.
	\end{equation} 
	Since $K$ arises from minimisation over all $\lambda\neq 0$, it is strictly greater than zero, leading to a non-trivial lower bound on the deviation. The coefficient $K$ will be fixed for any given symmetry principle described by the representation $U_A$ of $G$.
		
	Now, according to Eq.~\eqref{eq:unitarity_jamiolkowski} and using the decomposition from Eq.~\eqref{eq:directdecompChoi}, unitarity can be expressed in terms of the probability distribution $p_{\lambda}$ as follows:
	\begin{equation}
		u(\E) = \frac{1}{d_A^2-1}\left(d_A^2\sum_{\lambda\in\Lambda}\frac{p_{\lambda}(\E)^2}{d_{\lambda}} - 1\right),
	\end{equation}
	where we have used that the channel $\E$ is unital. We can then bound it using Cauchy-Schwartz inequality in the following way:
	\begin{equation}
		u(\E)\geq \frac{1}{d_{A}^2-1}\left(d_A^2 p_0^2 + d_A^2\frac{(1-p_0)^2}{\sum_{\lambda\neq 0} d_{\lambda}} -1\right).
	\end{equation}	
	Equivalently,
	\begin{align}
		1-u(\E)&\leq \frac{d_A^2}{d_A^2-1}\left(1 - p_0^2 -\frac{(1-p_0)^2}{\sum_{\lambda\neq 0}d_{\lambda}}\right)\leq \frac{2d_A^2}{d_A^2-1}(1-p_0).
	\end{align}
	Combining the two relations results in
	\begin{equation}
		\sqrt{\Delta(\E)}\geq K \|\vv{J}_{A}\| (1-u(\E)) \frac{(d_A-1)(d_A+1)^{1/2}}{2d_A^{5/2}}.
	\end{equation}
\end{proof}


\begin{thebibliography}{71}%
\makeatletter
\providecommand \@ifxundefined [1]{%
 \@ifx{#1\undefined}
}%
\providecommand \@ifnum [1]{%
 \ifnum #1\expandafter \@firstoftwo
 \else \expandafter \@secondoftwo
 \fi
}%
\providecommand \@ifx [1]{%
 \ifx #1\expandafter \@firstoftwo
 \else \expandafter \@secondoftwo
 \fi
}%
\providecommand \natexlab [1]{#1}%
\providecommand \enquote  [1]{``#1''}%
\providecommand \bibnamefont  [1]{#1}%
\providecommand \bibfnamefont [1]{#1}%
\providecommand \citenamefont [1]{#1}%
\providecommand \href@noop [0]{\@secondoftwo}%
\providecommand \href [0]{\begingroup \@sanitize@url \@href}%
\providecommand \@href[1]{\@@startlink{#1}\@@href}%
\providecommand \@@href[1]{\endgroup#1\@@endlink}%
\providecommand \@sanitize@url [0]{\catcode `\\12\catcode `\$12\catcode
  `\&12\catcode `\#12\catcode `\^12\catcode `\_12\catcode `\%12\relax}%
\providecommand \@@startlink[1]{}%
\providecommand \@@endlink[0]{}%
\providecommand \url  [0]{\begingroup\@sanitize@url \@url }%
\providecommand \@url [1]{\endgroup\@href {#1}{\urlprefix }}%
\providecommand \urlprefix  [0]{URL }%
\providecommand \Eprint [0]{\href }%
\providecommand \doibase [0]{https://doi.org/}%
\providecommand \selectlanguage [0]{\@gobble}%
\providecommand \bibinfo  [0]{\@secondoftwo}%
\providecommand \bibfield  [0]{\@secondoftwo}%
\providecommand \translation [1]{[#1]}%
\providecommand \BibitemOpen [0]{}%
\providecommand \bibitemStop [0]{}%
\providecommand \bibitemNoStop [0]{.\EOS\space}%
\providecommand \EOS [0]{\spacefactor3000\relax}%
\providecommand \BibitemShut  [1]{\csname bibitem#1\endcsname}%
\let\auto@bib@innerbib\@empty
\bibitem [{\citenamefont {Noether}(1918)}]{noether1918invariante}%
  \BibitemOpen
  \bibfield  {author} {\bibinfo {author} {\bibfnamefont {E.}~\bibnamefont
  {Noether}},\ }\bibfield  {title} {\bibinfo {title} {Invariante
  variationsprobleme},\ }\href@noop {} {\bibfield  {journal} {\bibinfo
  {journal} {Nachrichten von der Gesellschaft der Wissenschaften zu
  G{\"o}ttingen, Mathematisch-Physikalische Klasse}\ }\textbf {\bibinfo
  {volume} {1918}},\ \bibinfo {pages} {235} (\bibinfo {year}
  {1918})}\BibitemShut {NoStop}%
\bibitem [{\citenamefont {Baez}\ and\ \citenamefont
  {Fong}(2013)}]{baez2013noether}%
  \BibitemOpen
  \bibfield  {author} {\bibinfo {author} {\bibfnamefont {J.~C.}\ \bibnamefont
  {Baez}}\ and\ \bibinfo {author} {\bibfnamefont {B.}~\bibnamefont {Fong}},\
  }\bibfield  {title} {\bibinfo {title} {A {N}oether theorem for {M}arkov
  processes},\ }\href@noop {} {\bibfield  {journal} {\bibinfo  {journal} {J.
  Math. Phys.}\ }\textbf {\bibinfo {volume} {54}},\ \bibinfo {pages} {013301}
  (\bibinfo {year} {2013})}\BibitemShut {NoStop}%
\bibitem [{\citenamefont {Gough}\ \emph {et~al.}(2015)\citenamefont {Gough},
  \citenamefont {Ratiu},\ and\ \citenamefont {Smolyanov}}]{gough2015noether}%
  \BibitemOpen
  \bibfield  {author} {\bibinfo {author} {\bibfnamefont {J.~E.}\ \bibnamefont
  {Gough}}, \bibinfo {author} {\bibfnamefont {T.~S.}\ \bibnamefont {Ratiu}},\
  and\ \bibinfo {author} {\bibfnamefont {O.~G.}\ \bibnamefont {Smolyanov}},\
  }\bibfield  {title} {\bibinfo {title} {{N}oether's theorem for dissipative
  quantum dynamical semi-groups},\ }\href@noop {} {\bibfield  {journal}
  {\bibinfo  {journal} {J. Math. Phys.}\ }\textbf {\bibinfo {volume} {56}},\
  \bibinfo {pages} {022108} (\bibinfo {year} {2015})}\BibitemShut {NoStop}%
\bibitem [{\citenamefont {Ward}(1950)}]{ward1950identity}%
  \BibitemOpen
  \bibfield  {author} {\bibinfo {author} {\bibfnamefont {J.~C.}\ \bibnamefont
  {Ward}},\ }\bibfield  {title} {\bibinfo {title} {An identity in quantum
  electrodynamics},\ }\href@noop {} {\bibfield  {journal} {\bibinfo  {journal}
  {Phys. Rev.}\ }\textbf {\bibinfo {volume} {78}},\ \bibinfo {pages} {182}
  (\bibinfo {year} {1950})}\BibitemShut {NoStop}%
\bibitem [{\citenamefont {Takahashi}(1957)}]{takahashi1957generalized}%
  \BibitemOpen
  \bibfield  {author} {\bibinfo {author} {\bibfnamefont {Y.}~\bibnamefont
  {Takahashi}},\ }\bibfield  {title} {\bibinfo {title} {On the generalized
  {W}ard identity},\ }\href@noop {} {\bibfield  {journal} {\bibinfo  {journal}
  {Il Nuovo Cimento (1955-1965)}\ }\textbf {\bibinfo {volume} {6}},\ \bibinfo
  {pages} {371} (\bibinfo {year} {1957})}\BibitemShut {NoStop}%
\bibitem [{\citenamefont {Watrous}(2018)}]{watrous2018theory}%
  \BibitemOpen
  \bibfield  {author} {\bibinfo {author} {\bibfnamefont {J.}~\bibnamefont
  {Watrous}},\ }\href@noop {} {\emph {\bibinfo {title} {Theory of Quantum
  Information}}}\ (\bibinfo  {publisher} {Cambridge University Press},\
  \bibinfo {year} {2018})\BibitemShut {NoStop}%
\bibitem [{\citenamefont {Gour}\ \emph {et~al.}(2018)\citenamefont {Gour},
  \citenamefont {Jennings}, \citenamefont {Buscemi}, \citenamefont {Duan},\
  and\ \citenamefont {Marvian}}]{gour2017quantum}%
  \BibitemOpen
  \bibfield  {author} {\bibinfo {author} {\bibfnamefont {G.}~\bibnamefont
  {Gour}}, \bibinfo {author} {\bibfnamefont {D.}~\bibnamefont {Jennings}},
  \bibinfo {author} {\bibfnamefont {F.}~\bibnamefont {Buscemi}}, \bibinfo
  {author} {\bibfnamefont {R.}~\bibnamefont {Duan}},\ and\ \bibinfo {author}
  {\bibfnamefont {I.}~\bibnamefont {Marvian}},\ }\bibfield  {title} {\bibinfo
  {title} {Quantum majorization and a complete set of entropic conditions for
  quantum thermodynamics},\ }\href@noop {} {\bibfield  {journal} {\bibinfo
  {journal} {Nature Communications}\ }\textbf {\bibinfo {volume} {9}} (\bibinfo
  {year} {2018})}\BibitemShut {NoStop}%
\bibitem [{\citenamefont {Marvian}\ and\ \citenamefont
  {Spekkens}(2014{\natexlab{a}})}]{marvian2014extending}%
  \BibitemOpen
  \bibfield  {author} {\bibinfo {author} {\bibfnamefont {I.}~\bibnamefont
  {Marvian}}\ and\ \bibinfo {author} {\bibfnamefont {R.~W.}\ \bibnamefont
  {Spekkens}},\ }\bibfield  {title} {\bibinfo {title} {Extending {N}oether's
  theorem by quantifying the asymmetry of quantum states},\ }\href@noop {}
  {\bibfield  {journal} {\bibinfo  {journal} {Nat. Commun.}\ }\textbf {\bibinfo
  {volume} {5}},\ \bibinfo {pages} {3821} (\bibinfo {year}
  {2014}{\natexlab{a}})}\BibitemShut {NoStop}%
\bibitem [{\citenamefont {Peres}(2006)}]{peres2006quantum}%
  \BibitemOpen
  \bibfield  {author} {\bibinfo {author} {\bibfnamefont {A.}~\bibnamefont
  {Peres}},\ }\href@noop {} {\emph {\bibinfo {title} {Quantum theory: concepts
  and methods}}},\ Vol.~\bibinfo {volume} {57}\ (\bibinfo  {publisher}
  {Springer Science \& Business Media},\ \bibinfo {year} {2006})\BibitemShut
  {NoStop}%
\bibitem [{\citenamefont {{Bu{\v z}ek}}\ \emph {et~al.}(1999)\citenamefont
  {{Bu{\v z}ek}}, \citenamefont {{Hillery}},\ and\ \citenamefont
  {{Werner}}}]{buzek99optimal}%
  \BibitemOpen
  \bibfield  {author} {\bibinfo {author} {\bibfnamefont {V.}~\bibnamefont
  {{Bu{\v z}ek}}}, \bibinfo {author} {\bibfnamefont {M.}~\bibnamefont
  {{Hillery}}},\ and\ \bibinfo {author} {\bibfnamefont {R.~F.}\ \bibnamefont
  {{Werner}}},\ }\bibfield  {title} {\bibinfo {title} {{Optimal manipulations
  with qubits: Universal-NOT gate}},\ }\href@noop {} {\bibfield  {journal}
  {\bibinfo  {journal} {\pra}\ }\textbf {\bibinfo {volume} {60}},\ \bibinfo
  {pages} {R2626} (\bibinfo {year} {1999})}\BibitemShut {NoStop}%
\bibitem [{Note1()}]{Note1}%
  \BibitemOpen
  \bibinfo {note} {More precisely, spin-inversion is equivalent to $\rho
  \rightarrow \rho ^T$ followed by a $\pi $-rotation around the $Y$ axis.
  However, the transpose map is known to be a positive but not
  completely-positive map~\cite {watrous2018theory}. If we restrict the state
  space to be the set of separable quantum states, however, such a strictly
  positive map will never generate negative probabilities and so would be an
  admissible physical transformation.}\BibitemShut {Stop}%
\bibitem [{\citenamefont {Holevo}(1993)}]{holevo1993note}%
  \BibitemOpen
  \bibfield  {author} {\bibinfo {author} {\bibfnamefont {A.}~\bibnamefont
  {Holevo}},\ }\bibfield  {title} {\bibinfo {title} {{A note on covariant
  dynamical semigroups}},\ }\href@noop {} {\bibfield  {journal} {\bibinfo
  {journal} {Rep. Math. Phys.}\ }\textbf {\bibinfo {volume} {32}},\ \bibinfo
  {pages} {211} (\bibinfo {year} {1993})}\BibitemShut {NoStop}%
\bibitem [{\citenamefont {Holevo}(1995)}]{holevo1995structure}%
  \BibitemOpen
  \bibfield  {author} {\bibinfo {author} {\bibfnamefont {A.}~\bibnamefont
  {Holevo}},\ }\bibfield  {title} {\bibinfo {title} {{On the structure of
  covariant dynamical semigroups}},\ }\href@noop {} {\bibfield  {journal}
  {\bibinfo  {journal} {J. Funct. Anal.}\ }\textbf {\bibinfo {volume} {131}},\
  \bibinfo {pages} {255} (\bibinfo {year} {1995})}\BibitemShut {NoStop}%
\bibitem [{\citenamefont {Nuwairan}(2013)}]{nuwairan2013su2}%
  \BibitemOpen
  \bibfield  {author} {\bibinfo {author} {\bibfnamefont {M.~A.}\ \bibnamefont
  {Nuwairan}},\ }\bibfield  {title} {\bibinfo {title} {Su (2)-irreducibly
  covariant and eposic channels},\ }\href@noop {} {\bibfield  {journal}
  {\bibinfo  {journal} {arXiv:1306.5321}\ } (\bibinfo {year}
  {2013})}\BibitemShut {NoStop}%
\bibitem [{\citenamefont {Al~Nuwairan}(2015)}]{nuwairan2015su2}%
  \BibitemOpen
  \bibfield  {author} {\bibinfo {author} {\bibfnamefont {M.}~\bibnamefont
  {Al~Nuwairan}},\ }\emph {\bibinfo {title} {SU(2)-Irreducibly Covariant
  Quantum Channels and Some Applications}},\ \href@noop {} {Ph.D. thesis},\
  \bibinfo  {school} {University of Ottawa} (\bibinfo {year}
  {2015})\BibitemShut {NoStop}%
\bibitem [{\citenamefont {Mozrzymas}\ \emph {et~al.}(2017)\citenamefont
  {Mozrzymas}, \citenamefont {Studzi{\'n}ski},\ and\ \citenamefont
  {Datta}}]{mozrzymas2017structure}%
  \BibitemOpen
  \bibfield  {author} {\bibinfo {author} {\bibfnamefont {M.}~\bibnamefont
  {Mozrzymas}}, \bibinfo {author} {\bibfnamefont {M.}~\bibnamefont
  {Studzi{\'n}ski}},\ and\ \bibinfo {author} {\bibfnamefont {N.}~\bibnamefont
  {Datta}},\ }\bibfield  {title} {\bibinfo {title} {Structure of irreducibly
  covariant quantum channels for finite groups},\ }\href@noop {} {\bibfield
  {journal} {\bibinfo  {journal} {J. Math. Phys.}\ }\textbf {\bibinfo {volume}
  {58}},\ \bibinfo {pages} {052204} (\bibinfo {year} {2017})}\BibitemShut
  {NoStop}%
\bibitem [{\citenamefont {Wallman}\ \emph {et~al.}(2015)\citenamefont
  {Wallman}, \citenamefont {Granade}, \citenamefont {Harper},\ and\
  \citenamefont {Flammia}}]{wallman2015estimating}%
  \BibitemOpen
  \bibfield  {author} {\bibinfo {author} {\bibfnamefont {J.}~\bibnamefont
  {Wallman}}, \bibinfo {author} {\bibfnamefont {C.}~\bibnamefont {Granade}},
  \bibinfo {author} {\bibfnamefont {R.}~\bibnamefont {Harper}},\ and\ \bibinfo
  {author} {\bibfnamefont {S.~T.}\ \bibnamefont {Flammia}},\ }\bibfield
  {title} {\bibinfo {title} {Estimating the coherence of noise},\ }\href@noop
  {} {\bibfield  {journal} {\bibinfo  {journal} {New J. Phys.}\ }\textbf
  {\bibinfo {volume} {17}},\ \bibinfo {pages} {113020} (\bibinfo {year}
  {2015})}\BibitemShut {NoStop}%
\bibitem [{\citenamefont {Breuer}\ and\ \citenamefont
  {Petruccione}(2002)}]{breuer2002theory}%
  \BibitemOpen
  \bibfield  {author} {\bibinfo {author} {\bibfnamefont {H.-P.}\ \bibnamefont
  {Breuer}}\ and\ \bibinfo {author} {\bibfnamefont {F.}~\bibnamefont
  {Petruccione}},\ }\href@noop {} {\emph {\bibinfo {title} {The theory of open
  quantum systems}}}\ (\bibinfo  {publisher} {Oxford University Press on
  Demand},\ \bibinfo {year} {2002})\BibitemShut {NoStop}%
\bibitem [{\citenamefont {Knill}\ \emph {et~al.}(2008)\citenamefont {Knill},
  \citenamefont {Leibfried}, \citenamefont {Reichle}, \citenamefont {Britton},
  \citenamefont {Blakestad}, \citenamefont {Jost}, \citenamefont {Langer},
  \citenamefont {Ozeri}, \citenamefont {Seidelin},\ and\ \citenamefont
  {Wineland}}]{knill2008randomized}%
  \BibitemOpen
  \bibfield  {author} {\bibinfo {author} {\bibfnamefont {E.}~\bibnamefont
  {Knill}}, \bibinfo {author} {\bibfnamefont {D.}~\bibnamefont {Leibfried}},
  \bibinfo {author} {\bibfnamefont {R.}~\bibnamefont {Reichle}}, \bibinfo
  {author} {\bibfnamefont {J.}~\bibnamefont {Britton}}, \bibinfo {author}
  {\bibfnamefont {R.~B.}\ \bibnamefont {Blakestad}}, \bibinfo {author}
  {\bibfnamefont {J.~D.}\ \bibnamefont {Jost}}, \bibinfo {author}
  {\bibfnamefont {C.}~\bibnamefont {Langer}}, \bibinfo {author} {\bibfnamefont
  {R.}~\bibnamefont {Ozeri}}, \bibinfo {author} {\bibfnamefont
  {S.}~\bibnamefont {Seidelin}},\ and\ \bibinfo {author} {\bibfnamefont
  {D.~J.}\ \bibnamefont {Wineland}},\ }\bibfield  {title} {\bibinfo {title}
  {Randomized benchmarking of quantum gates},\ }\href@noop {} {\bibfield
  {journal} {\bibinfo  {journal} {Phys. Rev. A}\ }\textbf {\bibinfo {volume}
  {77}},\ \bibinfo {pages} {012307} (\bibinfo {year} {2008})}\BibitemShut
  {NoStop}%
\bibitem [{\citenamefont {Magesan}\ \emph {et~al.}(2011)\citenamefont
  {Magesan}, \citenamefont {Gambetta},\ and\ \citenamefont
  {Emerson}}]{magesan2011scalable}%
  \BibitemOpen
  \bibfield  {author} {\bibinfo {author} {\bibfnamefont {E.}~\bibnamefont
  {Magesan}}, \bibinfo {author} {\bibfnamefont {J.~M.}\ \bibnamefont
  {Gambetta}},\ and\ \bibinfo {author} {\bibfnamefont {J.}~\bibnamefont
  {Emerson}},\ }\bibfield  {title} {\bibinfo {title} {Scalable and robust
  randomized benchmarking of quantum processes},\ }\href@noop {} {\bibfield
  {journal} {\bibinfo  {journal} {Phys. Rev. Lett.}\ }\textbf {\bibinfo
  {volume} {106}},\ \bibinfo {pages} {180504} (\bibinfo {year}
  {2011})}\BibitemShut {NoStop}%
\bibitem [{\citenamefont {Georgescu}\ \emph {et~al.}(2014)\citenamefont
  {Georgescu}, \citenamefont {Ashhab},\ and\ \citenamefont
  {Nori}}]{georgescu2014quantum}%
  \BibitemOpen
  \bibfield  {author} {\bibinfo {author} {\bibfnamefont {I.~M.}\ \bibnamefont
  {Georgescu}}, \bibinfo {author} {\bibfnamefont {S.}~\bibnamefont {Ashhab}},\
  and\ \bibinfo {author} {\bibfnamefont {F.}~\bibnamefont {Nori}},\ }\bibfield
  {title} {\bibinfo {title} {Quantum simulation},\ }\href@noop {} {\bibfield
  {journal} {\bibinfo  {journal} {Rev. Mod. Phys.}\ }\textbf {\bibinfo {volume}
  {86}},\ \bibinfo {pages} {153} (\bibinfo {year} {2014})}\BibitemShut
  {NoStop}%
\bibitem [{\citenamefont {Ozawa}(2002)}]{ozawa2002conservation}%
  \BibitemOpen
  \bibfield  {author} {\bibinfo {author} {\bibfnamefont {M.}~\bibnamefont
  {Ozawa}},\ }\bibfield  {title} {\bibinfo {title} {Conservation laws,
  uncertainty relations, and quantum limits of measurements},\ }\href@noop {}
  {\bibfield  {journal} {\bibinfo  {journal} {Phys. Rev. Lett.}\ }\textbf
  {\bibinfo {volume} {88}},\ \bibinfo {pages} {050402} (\bibinfo {year}
  {2002})}\BibitemShut {NoStop}%
\bibitem [{\citenamefont {Goold}\ \emph {et~al.}(2016)\citenamefont {Goold},
  \citenamefont {Huber}, \citenamefont {Riera}, \citenamefont {del Rio},\ and\
  \citenamefont {Skrzypczyk}}]{goold2016role}%
  \BibitemOpen
  \bibfield  {author} {\bibinfo {author} {\bibfnamefont {J.}~\bibnamefont
  {Goold}}, \bibinfo {author} {\bibfnamefont {M.}~\bibnamefont {Huber}},
  \bibinfo {author} {\bibfnamefont {A.}~\bibnamefont {Riera}}, \bibinfo
  {author} {\bibfnamefont {L.}~\bibnamefont {del Rio}},\ and\ \bibinfo {author}
  {\bibfnamefont {P.}~\bibnamefont {Skrzypczyk}},\ }\bibfield  {title}
  {\bibinfo {title} {{The role of quantum information in thermodynamics -- a
  topical review}},\ }\href@noop {} {\bibfield  {journal} {\bibinfo  {journal}
  {J. Phys. A}\ }\textbf {\bibinfo {volume} {49}},\ \bibinfo {pages} {143001}
  (\bibinfo {year} {2016})}\BibitemShut {NoStop}%
\bibitem [{\citenamefont {H{\"o}hn}\ and\ \citenamefont
  {M{\"u}ller}(2016)}]{hohn2016operational}%
  \BibitemOpen
  \bibfield  {author} {\bibinfo {author} {\bibfnamefont {P.~A.}\ \bibnamefont
  {H{\"o}hn}}\ and\ \bibinfo {author} {\bibfnamefont {M.~P.}\ \bibnamefont
  {M{\"u}ller}},\ }\bibfield  {title} {\bibinfo {title} {An operational
  approach to spacetime symmetries: {L}orentz transformations from quantum
  communication},\ }\href@noop {} {\bibfield  {journal} {\bibinfo  {journal}
  {New J. Phys.}\ }\textbf {\bibinfo {volume} {18}},\ \bibinfo {pages} {063026}
  (\bibinfo {year} {2016})}\BibitemShut {NoStop}%
\bibitem [{\citenamefont {Vanrietvelde}\ \emph {et~al.}(2018)\citenamefont
  {Vanrietvelde}, \citenamefont {H{\"o}hn},\ and\ \citenamefont
  {Giacomini}}]{vanrietvelde2018switching}%
  \BibitemOpen
  \bibfield  {author} {\bibinfo {author} {\bibfnamefont {A.}~\bibnamefont
  {Vanrietvelde}}, \bibinfo {author} {\bibfnamefont {P.~A.}\ \bibnamefont
  {H{\"o}hn}},\ and\ \bibinfo {author} {\bibfnamefont {F.}~\bibnamefont
  {Giacomini}},\ }\bibfield  {title} {\bibinfo {title} {Switching quantum
  reference frames in the {N}-body problem and the absence of global relational
  perspectives},\ }\href@noop {} {\bibfield  {journal} {\bibinfo  {journal}
  {arXiv:1809.05093}\ } (\bibinfo {year} {2018})}\BibitemShut {NoStop}%
\bibitem [{\citenamefont {Scutaru}(1979)}]{scutaru}%
  \BibitemOpen
  \bibfield  {author} {\bibinfo {author} {\bibfnamefont {H.}~\bibnamefont
  {Scutaru}},\ }\bibfield  {title} {\bibinfo {title} {Some remarks on covariant
  completely positive linear maps on c*-algebras},\ }\href@noop {} {\bibfield
  {journal} {\bibinfo  {journal} {Rep. Math. Phys.}\ }\textbf {\bibinfo
  {volume} {16}},\ \bibinfo {pages} {79} (\bibinfo {year} {1979})}\BibitemShut
  {NoStop}%
\bibitem [{\citenamefont {Keyl}\ and\ \citenamefont
  {Werner}(1999)}]{keyl1999optimal}%
  \BibitemOpen
  \bibfield  {author} {\bibinfo {author} {\bibfnamefont {M.}~\bibnamefont
  {Keyl}}\ and\ \bibinfo {author} {\bibfnamefont {R.~F.}\ \bibnamefont
  {Werner}},\ }\bibfield  {title} {\bibinfo {title} {{Optimal cloning of pure
  states, testing single clones}},\ }\href@noop {} {\bibfield  {journal}
  {\bibinfo  {journal} {J. Math. Phys.}\ }\textbf {\bibinfo {volume} {40}},\
  \bibinfo {pages} {3283} (\bibinfo {year} {1999})}\BibitemShut {NoStop}%
\bibitem [{\citenamefont {Marvian}(2012)}]{marvian-thesis}%
  \BibitemOpen
  \bibfield  {author} {\bibinfo {author} {\bibfnamefont {I.}~\bibnamefont
  {Marvian}},\ }\emph {\bibinfo {title} {Symmetry, Asymmetry and Quantum
  Information}},\ \href@noop {} {Ph.D. thesis},\ \bibinfo  {school} {University
  of Waterloo.} (\bibinfo {year} {2012})\BibitemShut {NoStop}%
\bibitem [{\citenamefont {Marvian}\ and\ \citenamefont
  {Spekkens}(2014{\natexlab{b}})}]{marvian2014modes}%
  \BibitemOpen
  \bibfield  {author} {\bibinfo {author} {\bibfnamefont {I.}~\bibnamefont
  {Marvian}}\ and\ \bibinfo {author} {\bibfnamefont {R.~W.}\ \bibnamefont
  {Spekkens}},\ }\bibfield  {title} {\bibinfo {title} {Modes of asymmetry: the
  application of harmonic analysis to symmetric quantum dynamics and quantum
  reference frames},\ }\href@noop {} {\bibfield  {journal} {\bibinfo  {journal}
  {Phys. Rev. A}\ }\textbf {\bibinfo {volume} {90}},\ \bibinfo {pages} {062110}
  (\bibinfo {year} {2014}{\natexlab{b}})}\BibitemShut {NoStop}%
\bibitem [{\citenamefont {{van Enk}}(2005)}]{vanEnk}%
  \BibitemOpen
  \bibfield  {author} {\bibinfo {author} {\bibfnamefont {S.~J.}\ \bibnamefont
  {{van Enk}}},\ }\bibfield  {title} {\bibinfo {title} {{Relations between
  Cloning and the Universal NOT Derived from Conservation Laws}},\ }\href@noop
  {} {\bibfield  {journal} {\bibinfo  {journal} {\prl}\ }\textbf {\bibinfo
  {volume} {95}},\ \bibinfo {pages} {010502} (\bibinfo {year}
  {2005})}\BibitemShut {NoStop}%
\bibitem [{\citenamefont {{D'Ariano}}\ and\ \citenamefont
  {{Perinotti}}(2009)}]{NoStretching}%
  \BibitemOpen
  \bibfield  {author} {\bibinfo {author} {\bibfnamefont {G.~M.}\ \bibnamefont
  {{D'Ariano}}}\ and\ \bibinfo {author} {\bibfnamefont {P.}~\bibnamefont
  {{Perinotti}}},\ }\bibfield  {title} {\bibinfo {title} {{Quantum
  no-stretching: A geometrical interpretation of the no-cloning theorem}},\
  }\href@noop {} {\bibfield  {journal} {\bibinfo  {journal} {Phys. Lett. A}\
  }\textbf {\bibinfo {volume} {373}},\ \bibinfo {pages} {2416} (\bibinfo {year}
  {2009})}\BibitemShut {NoStop}%
\bibitem [{\citenamefont {Blume-Kohout}\ \emph {et~al.}(2010)\citenamefont
  {Blume-Kohout}, \citenamefont {Ng}, \citenamefont {Poulin},\ and\
  \citenamefont {Viola}}]{blume2010information}%
  \BibitemOpen
  \bibfield  {author} {\bibinfo {author} {\bibfnamefont {R.}~\bibnamefont
  {Blume-Kohout}}, \bibinfo {author} {\bibfnamefont {H.~K.}\ \bibnamefont
  {Ng}}, \bibinfo {author} {\bibfnamefont {D.}~\bibnamefont {Poulin}},\ and\
  \bibinfo {author} {\bibfnamefont {L.}~\bibnamefont {Viola}},\ }\bibfield
  {title} {\bibinfo {title} {Information-preserving structures: A general
  framework for quantum zero-error information},\ }\href@noop {} {\bibfield
  {journal} {\bibinfo  {journal} {Phys. Rev. A}\ }\textbf {\bibinfo {volume}
  {82}},\ \bibinfo {pages} {062306} (\bibinfo {year} {2010})}\BibitemShut
  {NoStop}%
\bibitem [{\citenamefont {Hall}(2013)}]{hall2013lie}%
  \BibitemOpen
  \bibfield  {author} {\bibinfo {author} {\bibfnamefont {B.~C.}\ \bibnamefont
  {Hall}},\ }\bibfield  {title} {\bibinfo {title} {Lie groups, lie algebras,
  and representations},\ }in\ \href@noop {} {\emph {\bibinfo {booktitle}
  {Quantum Theory for Mathematicians}}}\ (\bibinfo  {publisher} {Springer},\
  \bibinfo {year} {2013})\ pp.\ \bibinfo {pages} {333--366}\BibitemShut
  {NoStop}%
\bibitem [{\citenamefont {Albert}\ and\ \citenamefont
  {Jiang}(2014)}]{albert2014symmetries}%
  \BibitemOpen
  \bibfield  {author} {\bibinfo {author} {\bibfnamefont {V.~V.}\ \bibnamefont
  {Albert}}\ and\ \bibinfo {author} {\bibfnamefont {L.}~\bibnamefont {Jiang}},\
  }\bibfield  {title} {\bibinfo {title} {Symmetries and conserved quantities in
  lindblad master equations},\ }\href@noop {} {\bibfield  {journal} {\bibinfo
  {journal} {Physical Review A}\ }\textbf {\bibinfo {volume} {89}},\ \bibinfo
  {pages} {022118} (\bibinfo {year} {2014})}\BibitemShut {NoStop}%
\bibitem [{\citenamefont {Zhang}\ \emph {et~al.}(2020)\citenamefont {Zhang},
  \citenamefont {Tindall}, \citenamefont {Mur-Petit}, \citenamefont {Jaksch},\
  and\ \citenamefont {Bu{\v{c}}a}}]{zhang2020stationary}%
  \BibitemOpen
  \bibfield  {author} {\bibinfo {author} {\bibfnamefont {Z.}~\bibnamefont
  {Zhang}}, \bibinfo {author} {\bibfnamefont {J.}~\bibnamefont {Tindall}},
  \bibinfo {author} {\bibfnamefont {J.}~\bibnamefont {Mur-Petit}}, \bibinfo
  {author} {\bibfnamefont {D.}~\bibnamefont {Jaksch}},\ and\ \bibinfo {author}
  {\bibfnamefont {B.}~\bibnamefont {Bu{\v{c}}a}},\ }\bibfield  {title}
  {\bibinfo {title} {Stationary state degeneracy of open quantum systems with
  non-abelian symmetries},\ }\href@noop {} {\bibfield  {journal} {\bibinfo
  {journal} {Journal of Physics A: Mathematical and Theoretical}\ }\textbf
  {\bibinfo {volume} {53}},\ \bibinfo {pages} {215304} (\bibinfo {year}
  {2020})}\BibitemShut {NoStop}%
\bibitem [{\citenamefont {Bu{\v{c}}a}\ and\ \citenamefont
  {Prosen}(2012)}]{buvca2012note}%
  \BibitemOpen
  \bibfield  {author} {\bibinfo {author} {\bibfnamefont {B.}~\bibnamefont
  {Bu{\v{c}}a}}\ and\ \bibinfo {author} {\bibfnamefont {T.}~\bibnamefont
  {Prosen}},\ }\bibfield  {title} {\bibinfo {title} {A note on symmetry
  reductions of the lindblad equation: transport in constrained open spin
  chains},\ }\href@noop {} {\bibfield  {journal} {\bibinfo  {journal} {New
  Journal of Physics}\ }\textbf {\bibinfo {volume} {14}},\ \bibinfo {pages}
  {073007} (\bibinfo {year} {2012})}\BibitemShut {NoStop}%
\bibitem [{\citenamefont {Harrow}(2013)}]{harrow2013church}%
  \BibitemOpen
  \bibfield  {author} {\bibinfo {author} {\bibfnamefont {A.~W.}\ \bibnamefont
  {Harrow}},\ }\bibfield  {title} {\bibinfo {title} {The church of the
  symmetric subspace},\ }\href@noop {} {\bibfield  {journal} {\bibinfo
  {journal} {arXiv:1308.6595}\ } (\bibinfo {year} {2013})}\BibitemShut
  {NoStop}%
\bibitem [{\citenamefont {Stembridge}(2003)}]{stembridge2003multiplicity}%
  \BibitemOpen
  \bibfield  {author} {\bibinfo {author} {\bibfnamefont {J.~R.}\ \bibnamefont
  {Stembridge}},\ }\bibfield  {title} {\bibinfo {title} {Multiplicity-free
  products and restrictions of weyl characters},\ }\href@noop {} {\bibfield
  {journal} {\bibinfo  {journal} {Represent. Theory}\ }\textbf {\bibinfo
  {volume} {7}},\ \bibinfo {pages} {404} (\bibinfo {year} {2003})}\BibitemShut
  {NoStop}%
\bibitem [{Note2()}]{Note2}%
  \BibitemOpen
  \bibinfo {note} {More precisely: for a system exhibiting cyclic dynamics,
  i.e., such that there exists time $t_0$ for which $e^{iH_At_0}=\protect
  \mathbb {I}_A$.}\BibitemShut {Stop}%
\bibitem [{\citenamefont {Lostaglio}\ \emph
  {et~al.}(2015{\natexlab{a}})\citenamefont {Lostaglio}, \citenamefont
  {Korzekwa}, \citenamefont {Jennings},\ and\ \citenamefont
  {Rudolph}}]{lostaglio2015quantum}%
  \BibitemOpen
  \bibfield  {author} {\bibinfo {author} {\bibfnamefont {M.}~\bibnamefont
  {Lostaglio}}, \bibinfo {author} {\bibfnamefont {K.}~\bibnamefont {Korzekwa}},
  \bibinfo {author} {\bibfnamefont {D.}~\bibnamefont {Jennings}},\ and\
  \bibinfo {author} {\bibfnamefont {T.}~\bibnamefont {Rudolph}},\ }\bibfield
  {title} {\bibinfo {title} {Quantum coherence, time-translation symmetry, and
  thermodynamics},\ }\href@noop {} {\bibfield  {journal} {\bibinfo  {journal}
  {Phys. Rev. X}\ }\textbf {\bibinfo {volume} {5}},\ \bibinfo {pages} {021001}
  (\bibinfo {year} {2015}{\natexlab{a}})}\BibitemShut {NoStop}%
\bibitem [{\citenamefont {Lostaglio}\ \emph {et~al.}(2017)\citenamefont
  {Lostaglio}, \citenamefont {Korzekwa},\ and\ \citenamefont
  {Milne}}]{lostaglio2017markovian}%
  \BibitemOpen
  \bibfield  {author} {\bibinfo {author} {\bibfnamefont {M.}~\bibnamefont
  {Lostaglio}}, \bibinfo {author} {\bibfnamefont {K.}~\bibnamefont
  {Korzekwa}},\ and\ \bibinfo {author} {\bibfnamefont {A.}~\bibnamefont
  {Milne}},\ }\bibfield  {title} {\bibinfo {title} {Markovian evolution of
  quantum coherence under symmetric dynamics},\ }\href@noop {} {\bibfield
  {journal} {\bibinfo  {journal} {Phys. Rev. A}\ }\textbf {\bibinfo {volume}
  {96}},\ \bibinfo {pages} {032109} (\bibinfo {year} {2017})}\BibitemShut
  {NoStop}%
\bibitem [{\citenamefont {Korzekwa}\ \emph {et~al.}(2018)\citenamefont
  {Korzekwa}, \citenamefont {Czach{\'o}rski}, \citenamefont {Pucha{\l}a},\ and\
  \citenamefont {{\.Z}yczkowski}}]{korzekwa2018coherifying}%
  \BibitemOpen
  \bibfield  {author} {\bibinfo {author} {\bibfnamefont {K.}~\bibnamefont
  {Korzekwa}}, \bibinfo {author} {\bibfnamefont {S.}~\bibnamefont
  {Czach{\'o}rski}}, \bibinfo {author} {\bibfnamefont {Z.}~\bibnamefont
  {Pucha{\l}a}},\ and\ \bibinfo {author} {\bibfnamefont {K.}~\bibnamefont
  {{\.Z}yczkowski}},\ }\bibfield  {title} {\bibinfo {title} {Coherifying
  quantum channels},\ }\href@noop {} {\bibfield  {journal} {\bibinfo  {journal}
  {New J. Phys.}\ }\textbf {\bibinfo {volume} {20}},\ \bibinfo {pages} {043028}
  (\bibinfo {year} {2018})}\BibitemShut {NoStop}%
\bibitem [{Note3()}]{Note3}%
  \BibitemOpen
  \bibinfo {note} {\protect \leavevmode {\protect \color {black}Note that we
  can invoke linearity here as we restricted to the convex hull of
  spin-coherent states (i.e a classical state space), and on this subspace the
  anti-linear property of the map $ {\protect \mathcal T} $ is not relevant
  since the scalar field is real.}}\BibitemShut {Stop}%
\bibitem [{\citenamefont {Albert}(2019)}]{albert2019asymptotics}%
  \BibitemOpen
  \bibfield  {author} {\bibinfo {author} {\bibfnamefont {V.~V.}\ \bibnamefont
  {Albert}},\ }\bibfield  {title} {\bibinfo {title} {Asymptotics of quantum
  channels: conserved quantities, an adiabatic limit, and matrix product
  states},\ }\href@noop {} {\bibfield  {journal} {\bibinfo  {journal}
  {Quantum}\ }\textbf {\bibinfo {volume} {3}},\ \bibinfo {pages} {151}
  (\bibinfo {year} {2019})}\BibitemShut {NoStop}%
\bibitem [{Note4()}]{Note4}%
  \BibitemOpen
  \bibinfo {note} {Of course this is a very rough bound, since for small $q_0$
  it may happen that the remaining energy flows $d-g-q_0$ are still larger than
  $g$, and should be split over more indices $\lambda $}\BibitemShut {NoStop}%
\bibitem [{\citenamefont {Manzano}\ and\ \citenamefont
  {Hurtado}(2014)}]{manzano2014symmetry}%
  \BibitemOpen
  \bibfield  {author} {\bibinfo {author} {\bibfnamefont {D.}~\bibnamefont
  {Manzano}}\ and\ \bibinfo {author} {\bibfnamefont {P.~I.}\ \bibnamefont
  {Hurtado}},\ }\bibfield  {title} {\bibinfo {title} {Symmetry and the
  thermodynamics of currents in open quantum systems},\ }\href@noop {}
  {\bibfield  {journal} {\bibinfo  {journal} {Physical Review B}\ }\textbf
  {\bibinfo {volume} {90}},\ \bibinfo {pages} {125138} (\bibinfo {year}
  {2014})}\BibitemShut {NoStop}%
\bibitem [{\citenamefont {{Horodecki}}\ and\ \citenamefont
  {{Oppenheim}}(2013)}]{horodecki2013fundamental}%
  \BibitemOpen
  \bibfield  {author} {\bibinfo {author} {\bibfnamefont {M.}~\bibnamefont
  {{Horodecki}}}\ and\ \bibinfo {author} {\bibfnamefont {J.}~\bibnamefont
  {{Oppenheim}}},\ }\bibfield  {title} {\bibinfo {title} {{Fundamental
  limitations for quantum and nanoscale thermodynamics}},\ }\href@noop {}
  {\bibfield  {journal} {\bibinfo  {journal} {Nat. Commun.}\ }\textbf {\bibinfo
  {volume} {4}} (\bibinfo {year} {2013})}\BibitemShut {NoStop}%
\bibitem [{\citenamefont {{Brand\~ao}}\ \emph {et~al.}(2015)\citenamefont
  {{Brand\~ao}}, \citenamefont {{Horodecki}}, \citenamefont {{Ng}},
  \citenamefont {{Oppenheim}},\ and\ \citenamefont
  {{Wehner}}}]{brandao2013second}%
  \BibitemOpen
  \bibfield  {author} {\bibinfo {author} {\bibfnamefont {F.~G.~S.~L.}\
  \bibnamefont {{Brand\~ao}}}, \bibinfo {author} {\bibfnamefont
  {M.}~\bibnamefont {{Horodecki}}}, \bibinfo {author} {\bibfnamefont
  {N.~H.~Y.}\ \bibnamefont {{Ng}}}, \bibinfo {author} {\bibfnamefont
  {J.}~\bibnamefont {{Oppenheim}}},\ and\ \bibinfo {author} {\bibfnamefont
  {S.}~\bibnamefont {{Wehner}}},\ }\bibfield  {title} {\bibinfo {title} {{The
  second laws of quantum thermodynamics}},\ }\href@noop {} {\bibfield
  {journal} {\bibinfo  {journal} {Proc. Natl. Acad. Sci. U.S.A.}\ }\textbf
  {\bibinfo {volume} {112}},\ \bibinfo {pages} {3275} (\bibinfo {year}
  {2015})}\BibitemShut {NoStop}%
\bibitem [{\citenamefont {Lostaglio}\ \emph
  {et~al.}(2015{\natexlab{b}})\citenamefont {Lostaglio}, \citenamefont
  {Jennings},\ and\ \citenamefont {Rudolph}}]{lostaglio2015description}%
  \BibitemOpen
  \bibfield  {author} {\bibinfo {author} {\bibfnamefont {M.}~\bibnamefont
  {Lostaglio}}, \bibinfo {author} {\bibfnamefont {D.}~\bibnamefont
  {Jennings}},\ and\ \bibinfo {author} {\bibfnamefont {T.}~\bibnamefont
  {Rudolph}},\ }\bibfield  {title} {\bibinfo {title} {Description of quantum
  coherence in thermodynamic processes requires constraints beyond free
  energy},\ }\href@noop {} {\bibfield  {journal} {\bibinfo  {journal} {Nat.
  Commun.}\ }\textbf {\bibinfo {volume} {6}},\ \bibinfo {pages} {6383}
  (\bibinfo {year} {2015}{\natexlab{b}})}\BibitemShut {NoStop}%
\bibitem [{\citenamefont {Guryanova}\ \emph {et~al.}(2016)\citenamefont
  {Guryanova}, \citenamefont {Popescu}, \citenamefont {Short}, \citenamefont
  {Silva},\ and\ \citenamefont {Skrzypczyk}}]{guryanova2016thermodynamics}%
  \BibitemOpen
  \bibfield  {author} {\bibinfo {author} {\bibfnamefont {Y.}~\bibnamefont
  {Guryanova}}, \bibinfo {author} {\bibfnamefont {S.}~\bibnamefont {Popescu}},
  \bibinfo {author} {\bibfnamefont {A.~J.}\ \bibnamefont {Short}}, \bibinfo
  {author} {\bibfnamefont {R.}~\bibnamefont {Silva}},\ and\ \bibinfo {author}
  {\bibfnamefont {P.}~\bibnamefont {Skrzypczyk}},\ }\bibfield  {title}
  {\bibinfo {title} {Thermodynamics of quantum systems with multiple conserved
  quantities},\ }\href@noop {} {\bibfield  {journal} {\bibinfo  {journal} {Nat.
  Commun.}\ }\textbf {\bibinfo {volume} {7}},\ \bibinfo {pages} {12049}
  (\bibinfo {year} {2016})}\BibitemShut {NoStop}%
\bibitem [{\citenamefont {Busch}(2010)}]{busch2010translation}%
  \BibitemOpen
  \bibfield  {author} {\bibinfo {author} {\bibfnamefont {P.}~\bibnamefont
  {Busch}},\ }\bibfield  {title} {\bibinfo {title} {Translation of "die
  {M}essung quantenmechanischer {O}peratoren" by {E}{P} {W}igner},\ }\href@noop
  {} {\bibfield  {journal} {\bibinfo  {journal} {arXiv:1012.4372}\ } (\bibinfo
  {year} {2010})}\BibitemShut {NoStop}%
\bibitem [{\citenamefont {Araki}\ and\ \citenamefont
  {Yanase}(1960)}]{araki1960measurement}%
  \BibitemOpen
  \bibfield  {author} {\bibinfo {author} {\bibfnamefont {H.}~\bibnamefont
  {Araki}}\ and\ \bibinfo {author} {\bibfnamefont {M.~M.}\ \bibnamefont
  {Yanase}},\ }\bibfield  {title} {\bibinfo {title} {Measurement of quantum
  mechanical operators},\ }\href@noop {} {\bibfield  {journal} {\bibinfo
  {journal} {Phys. Rev.}\ }\textbf {\bibinfo {volume} {120}},\ \bibinfo {pages}
  {622} (\bibinfo {year} {1960})}\BibitemShut {NoStop}%
\bibitem [{\citenamefont {Yanase}(1961)}]{yanase1961optimal}%
  \BibitemOpen
  \bibfield  {author} {\bibinfo {author} {\bibfnamefont {M.~M.}\ \bibnamefont
  {Yanase}},\ }\bibfield  {title} {\bibinfo {title} {Optimal measuring
  apparatus},\ }\href@noop {} {\bibfield  {journal} {\bibinfo  {journal} {Phys.
  Rev.}\ }\textbf {\bibinfo {volume} {123}},\ \bibinfo {pages} {666} (\bibinfo
  {year} {1961})}\BibitemShut {NoStop}%
\bibitem [{\citenamefont {Ahmadi}\ \emph {et~al.}(2013)\citenamefont {Ahmadi},
  \citenamefont {Jennings},\ and\ \citenamefont {Rudolph}}]{ahmadi2013wigner}%
  \BibitemOpen
  \bibfield  {author} {\bibinfo {author} {\bibfnamefont {M.}~\bibnamefont
  {Ahmadi}}, \bibinfo {author} {\bibfnamefont {D.}~\bibnamefont {Jennings}},\
  and\ \bibinfo {author} {\bibfnamefont {T.}~\bibnamefont {Rudolph}},\
  }\bibfield  {title} {\bibinfo {title} {The wigner--araki--yanase theorem and
  the quantum resource theory of asymmetry},\ }\href@noop {} {\bibfield
  {journal} {\bibinfo  {journal} {New J. Phys.}\ }\textbf {\bibinfo {volume}
  {15}},\ \bibinfo {pages} {013057} (\bibinfo {year} {2013})}\BibitemShut
  {NoStop}%
\bibitem [{\citenamefont {Preskill}(2018)}]{preskill2018quantum}%
  \BibitemOpen
  \bibfield  {author} {\bibinfo {author} {\bibfnamefont {J.}~\bibnamefont
  {Preskill}},\ }\bibfield  {title} {\bibinfo {title} {Quantum computing in the
  nisq era and beyond},\ }\href@noop {} {\bibfield  {journal} {\bibinfo
  {journal} {Quantum}\ }\textbf {\bibinfo {volume} {2}},\ \bibinfo {pages} {79}
  (\bibinfo {year} {2018})}\BibitemShut {NoStop}%
\bibitem [{\citenamefont {Arute}\ \emph {et~al.}(2019)\citenamefont {Arute},
  \citenamefont {Arya}, \citenamefont {Babbush}, \citenamefont {Bacon},
  \citenamefont {Bardin}, \citenamefont {Barends}, \citenamefont {Biswas},
  \citenamefont {Boixo}, \citenamefont {Brandao}, \citenamefont {Buell} \emph
  {et~al.}}]{arute2019quantum}%
  \BibitemOpen
  \bibfield  {author} {\bibinfo {author} {\bibfnamefont {F.}~\bibnamefont
  {Arute}}, \bibinfo {author} {\bibfnamefont {K.}~\bibnamefont {Arya}},
  \bibinfo {author} {\bibfnamefont {R.}~\bibnamefont {Babbush}}, \bibinfo
  {author} {\bibfnamefont {D.}~\bibnamefont {Bacon}}, \bibinfo {author}
  {\bibfnamefont {J.~C.}\ \bibnamefont {Bardin}}, \bibinfo {author}
  {\bibfnamefont {R.}~\bibnamefont {Barends}}, \bibinfo {author} {\bibfnamefont
  {R.}~\bibnamefont {Biswas}}, \bibinfo {author} {\bibfnamefont
  {S.}~\bibnamefont {Boixo}}, \bibinfo {author} {\bibfnamefont {F.~G.}\
  \bibnamefont {Brandao}}, \bibinfo {author} {\bibfnamefont {D.~A.}\
  \bibnamefont {Buell}}, \emph {et~al.},\ }\bibfield  {title} {\bibinfo {title}
  {Quantum supremacy using a programmable superconducting processor},\
  }\href@noop {} {\bibfield  {journal} {\bibinfo  {journal} {Nature}\ }\textbf
  {\bibinfo {volume} {574}},\ \bibinfo {pages} {505} (\bibinfo {year}
  {2019})}\BibitemShut {NoStop}%
\bibitem [{\citenamefont {McArdle}\ \emph {et~al.}(2019)\citenamefont
  {McArdle}, \citenamefont {Yuan},\ and\ \citenamefont
  {Benjamin}}]{mcardle2019error}%
  \BibitemOpen
  \bibfield  {author} {\bibinfo {author} {\bibfnamefont {S.}~\bibnamefont
  {McArdle}}, \bibinfo {author} {\bibfnamefont {X.}~\bibnamefont {Yuan}},\ and\
  \bibinfo {author} {\bibfnamefont {S.}~\bibnamefont {Benjamin}},\ }\bibfield
  {title} {\bibinfo {title} {Error-mitigated digital quantum simulation},\
  }\href@noop {} {\bibfield  {journal} {\bibinfo  {journal} {Phys. Rev. Lett.}\
  }\textbf {\bibinfo {volume} {122}},\ \bibinfo {pages} {180501} (\bibinfo
  {year} {2019})}\BibitemShut {NoStop}%
\bibitem [{\citenamefont {Gard}\ \emph {et~al.}(2020)\citenamefont {Gard},
  \citenamefont {Zhu}, \citenamefont {Barron}, \citenamefont {Mayhall},
  \citenamefont {Economou},\ and\ \citenamefont {Barnes}}]{gard2020efficient}%
  \BibitemOpen
  \bibfield  {author} {\bibinfo {author} {\bibfnamefont {B.~T.}\ \bibnamefont
  {Gard}}, \bibinfo {author} {\bibfnamefont {L.}~\bibnamefont {Zhu}}, \bibinfo
  {author} {\bibfnamefont {G.~S.}\ \bibnamefont {Barron}}, \bibinfo {author}
  {\bibfnamefont {N.~J.}\ \bibnamefont {Mayhall}}, \bibinfo {author}
  {\bibfnamefont {S.~E.}\ \bibnamefont {Economou}},\ and\ \bibinfo {author}
  {\bibfnamefont {E.}~\bibnamefont {Barnes}},\ }\bibfield  {title} {\bibinfo
  {title} {Efficient symmetry-preserving state preparation circuits for the
  variational quantum eigensolver algorithm},\ }\href@noop {} {\bibfield
  {journal} {\bibinfo  {journal} {Npj Quantum Inf.}\ }\textbf {\bibinfo
  {volume} {6}},\ \bibinfo {pages} {1} (\bibinfo {year} {2020})}\BibitemShut
  {NoStop}%
\bibitem [{\citenamefont {Erhard}\ \emph {et~al.}(2019)\citenamefont {Erhard},
  \citenamefont {Wallman}, \citenamefont {Postler}, \citenamefont {Meth},
  \citenamefont {Stricker}, \citenamefont {Martinez}, \citenamefont
  {Schindler}, \citenamefont {Monz}, \citenamefont {Emerson},\ and\
  \citenamefont {Blatt}}]{erhard2019characterizing}%
  \BibitemOpen
  \bibfield  {author} {\bibinfo {author} {\bibfnamefont {A.}~\bibnamefont
  {Erhard}}, \bibinfo {author} {\bibfnamefont {J.~J.}\ \bibnamefont {Wallman}},
  \bibinfo {author} {\bibfnamefont {L.}~\bibnamefont {Postler}}, \bibinfo
  {author} {\bibfnamefont {M.}~\bibnamefont {Meth}}, \bibinfo {author}
  {\bibfnamefont {R.}~\bibnamefont {Stricker}}, \bibinfo {author}
  {\bibfnamefont {E.~A.}\ \bibnamefont {Martinez}}, \bibinfo {author}
  {\bibfnamefont {P.}~\bibnamefont {Schindler}}, \bibinfo {author}
  {\bibfnamefont {T.}~\bibnamefont {Monz}}, \bibinfo {author} {\bibfnamefont
  {J.}~\bibnamefont {Emerson}},\ and\ \bibinfo {author} {\bibfnamefont
  {R.}~\bibnamefont {Blatt}},\ }\bibfield  {title} {\bibinfo {title}
  {Characterizing large-scale quantum computers via cycle benchmarking},\
  }\href@noop {} {\bibfield  {journal} {\bibinfo  {journal} {Nat. Commun.}\
  }\textbf {\bibinfo {volume} {10}},\ \bibinfo {pages} {1} (\bibinfo {year}
  {2019})}\BibitemShut {NoStop}%
\bibitem [{\citenamefont {Romero}\ \emph {et~al.}(2018)\citenamefont {Romero},
  \citenamefont {Babbush}, \citenamefont {McClean}, \citenamefont {Hempel},
  \citenamefont {Love},\ and\ \citenamefont
  {Aspuru-Guzik}}]{romero2018strategies}%
  \BibitemOpen
  \bibfield  {author} {\bibinfo {author} {\bibfnamefont {J.}~\bibnamefont
  {Romero}}, \bibinfo {author} {\bibfnamefont {R.}~\bibnamefont {Babbush}},
  \bibinfo {author} {\bibfnamefont {J.~R.}\ \bibnamefont {McClean}}, \bibinfo
  {author} {\bibfnamefont {C.}~\bibnamefont {Hempel}}, \bibinfo {author}
  {\bibfnamefont {P.~J.}\ \bibnamefont {Love}},\ and\ \bibinfo {author}
  {\bibfnamefont {A.}~\bibnamefont {Aspuru-Guzik}},\ }\bibfield  {title}
  {\bibinfo {title} {Strategies for quantum computing molecular energies using
  the unitary coupled cluster ansatz},\ }\href@noop {} {\bibfield  {journal}
  {\bibinfo  {journal} {Quantum Sci. Technol.}\ }\textbf {\bibinfo {volume}
  {4}},\ \bibinfo {pages} {014008} (\bibinfo {year} {2018})}\BibitemShut
  {NoStop}%
\bibitem [{\citenamefont {D'Ariano}(2004)}]{d2004extremal}%
  \BibitemOpen
  \bibfield  {author} {\bibinfo {author} {\bibfnamefont {G.~M.}\ \bibnamefont
  {D'Ariano}},\ }\bibfield  {title} {\bibinfo {title} {Extremal covariant
  quantum operations and positive operator valued measures},\ }\href@noop {}
  {\bibfield  {journal} {\bibinfo  {journal} {J. Math. Phys.}\ }\textbf
  {\bibinfo {volume} {45}},\ \bibinfo {pages} {3620} (\bibinfo {year}
  {2004})}\BibitemShut {NoStop}%
\bibitem [{\citenamefont {Ruskai}(2007)}]{ruskai2007some}%
  \BibitemOpen
  \bibfield  {author} {\bibinfo {author} {\bibfnamefont {M.~B.}\ \bibnamefont
  {Ruskai}},\ }\bibfield  {title} {\bibinfo {title} {Some open problems in
  quantum information theory},\ }\href@noop {} {\bibfield  {journal} {\bibinfo
  {journal} {arXiv:0708.1902}\ } (\bibinfo {year} {2007})}\BibitemShut
  {NoStop}%
\bibitem [{\citenamefont {Siudzi{\'n}ska}\ and\ \citenamefont
  {Chru{\'s}ci{\'n}ski}(2018)}]{siudzinska2018quantum}%
  \BibitemOpen
  \bibfield  {author} {\bibinfo {author} {\bibfnamefont {K.}~\bibnamefont
  {Siudzi{\'n}ska}}\ and\ \bibinfo {author} {\bibfnamefont {D.}~\bibnamefont
  {Chru{\'s}ci{\'n}ski}},\ }\bibfield  {title} {\bibinfo {title} {Quantum
  channels irreducibly covariant with respect to the finite group generated by
  the weyl operators},\ }\href@noop {} {\bibfield  {journal} {\bibinfo
  {journal} {J. Math. Phys.}\ }\textbf {\bibinfo {volume} {59}},\ \bibinfo
  {pages} {033508} (\bibinfo {year} {2018})}\BibitemShut {NoStop}%
\bibitem [{\citenamefont {Holevo}(2002)}]{holevo2002remarks}%
  \BibitemOpen
  \bibfield  {author} {\bibinfo {author} {\bibfnamefont {A.}~\bibnamefont
  {Holevo}},\ }\bibfield  {title} {\bibinfo {title} {Remarks on the classical
  capacity of quantum channel},\ }\href@noop {} {\bibfield  {journal} {\bibinfo
   {journal} {arXiv:quant-ph/0212025}\ } (\bibinfo {year} {2002})}\BibitemShut
  {NoStop}%
\bibitem [{\citenamefont {Holevo}(2005)}]{holevo2005additivity}%
  \BibitemOpen
  \bibfield  {author} {\bibinfo {author} {\bibfnamefont {A.}~\bibnamefont
  {Holevo}},\ }\bibfield  {title} {\bibinfo {title} {Additivity conjecture and
  covariant channels},\ }\href@noop {} {\bibfield  {journal} {\bibinfo
  {journal} {Int. J. Quantum Inf.}\ }\textbf {\bibinfo {volume} {3}},\ \bibinfo
  {pages} {41} (\bibinfo {year} {2005})}\BibitemShut {NoStop}%
\bibitem [{\citenamefont {Wigner}\ and\ \citenamefont
  {Yanase}(1963)}]{yanase1963information}%
  \BibitemOpen
  \bibfield  {author} {\bibinfo {author} {\bibfnamefont {E.}~\bibnamefont
  {Wigner}}\ and\ \bibinfo {author} {\bibfnamefont {M.}~\bibnamefont
  {Yanase}},\ }\bibfield  {title} {\bibinfo {title} {Information contents of
  distributions},\ }\href@noop {} {\bibfield  {journal} {\bibinfo  {journal}
  {Proc. Natl. Acad. Sci. (USA)}\ }\textbf {\bibinfo {volume} {49}},\ \bibinfo
  {pages} {910} (\bibinfo {year} {1963})}\BibitemShut {NoStop}%
\bibitem [{\citenamefont {Marvian}\ and\ \citenamefont
  {Spekkens}(2012)}]{marvian2012information}%
  \BibitemOpen
  \bibfield  {author} {\bibinfo {author} {\bibfnamefont {I.}~\bibnamefont
  {Marvian}}\ and\ \bibinfo {author} {\bibfnamefont {R.~W.}\ \bibnamefont
  {Spekkens}},\ }\bibfield  {title} {\bibinfo {title} {An information-theoretic
  account of the wigner-araki-yanase theorem},\ }\href@noop {} {\bibfield
  {journal} {\bibinfo  {journal} {arXiv:1212.3378}\ } (\bibinfo {year}
  {2012})}\BibitemShut {NoStop}%
\bibitem [{\citenamefont {Hebdige}\ and\ \citenamefont
  {Jennings}(2019)}]{hebdige2019classification}%
  \BibitemOpen
  \bibfield  {author} {\bibinfo {author} {\bibfnamefont {T.}~\bibnamefont
  {Hebdige}}\ and\ \bibinfo {author} {\bibfnamefont {D.}~\bibnamefont
  {Jennings}},\ }\bibfield  {title} {\bibinfo {title} {On the classification of
  two-qubit group orbits and the use of coarse-grained'shape'as a
  superselection property},\ }\href@noop {} {\bibfield  {journal} {\bibinfo
  {journal} {Quantum}\ }\textbf {\bibinfo {volume} {3}},\ \bibinfo {pages}
  {119} (\bibinfo {year} {2019})}\BibitemShut {NoStop}%
\bibitem [{\citenamefont {Cirstoiu}\ and\ \citenamefont
  {Jennings}(2017)}]{cirstoiu2017global}%
  \BibitemOpen
  \bibfield  {author} {\bibinfo {author} {\bibfnamefont {C.}~\bibnamefont
  {Cirstoiu}}\ and\ \bibinfo {author} {\bibfnamefont {D.}~\bibnamefont
  {Jennings}},\ }\bibfield  {title} {\bibinfo {title} {Global and local gauge
  symmetries beyond lagrangian formulations},\ }\href@noop {} {\bibfield
  {journal} {\bibinfo  {journal} {arXiv:1707.09826}\ } (\bibinfo {year}
  {2017})}\BibitemShut {NoStop}%
\bibitem [{\citenamefont {Pikovski}\ \emph {et~al.}(2015)\citenamefont
  {Pikovski}, \citenamefont {Zych}, \citenamefont {Costa},\ and\ \citenamefont
  {Brukner}}]{pikovski2015universal}%
  \BibitemOpen
  \bibfield  {author} {\bibinfo {author} {\bibfnamefont {I.}~\bibnamefont
  {Pikovski}}, \bibinfo {author} {\bibfnamefont {M.}~\bibnamefont {Zych}},
  \bibinfo {author} {\bibfnamefont {F.}~\bibnamefont {Costa}},\ and\ \bibinfo
  {author} {\bibfnamefont {{\v{C}}.}~\bibnamefont {Brukner}},\ }\bibfield
  {title} {\bibinfo {title} {Universal decoherence due to gravitational time
  dilation},\ }\href@noop {} {\bibfield  {journal} {\bibinfo  {journal} {Nat.
  Phys.}\ }\textbf {\bibinfo {volume} {11}},\ \bibinfo {pages} {668} (\bibinfo
  {year} {2015})}\BibitemShut {NoStop}%
\bibitem [{\citenamefont {Perez-Garcia}\ \emph {et~al.}(2006)\citenamefont
  {Perez-Garcia}, \citenamefont {Wolf}, \citenamefont {Petz},\ and\
  \citenamefont {Ruskai}}]{perez2006contractivity}%
  \BibitemOpen
  \bibfield  {author} {\bibinfo {author} {\bibfnamefont {D.}~\bibnamefont
  {Perez-Garcia}}, \bibinfo {author} {\bibfnamefont {M.~M.}\ \bibnamefont
  {Wolf}}, \bibinfo {author} {\bibfnamefont {D.}~\bibnamefont {Petz}},\ and\
  \bibinfo {author} {\bibfnamefont {M.~B.}\ \bibnamefont {Ruskai}},\ }\bibfield
   {title} {\bibinfo {title} {Contractivity of positive and trace-preserving
  maps under l p norms},\ }\href@noop {} {\bibfield  {journal} {\bibinfo
  {journal} {J. Math. Phys.}\ }\textbf {\bibinfo {volume} {47}},\ \bibinfo
  {pages} {083506} (\bibinfo {year} {2006})}\BibitemShut {NoStop}%
\end{thebibliography}%
%

\end{document}